\documentclass[11pt,reqno]{amsart}
\usepackage{amsmath, amsthm, amssymb}
\usepackage{enumitem}
\usepackage[foot]{amsaddr}

\usepackage{fullpage}
\usepackage[many]{tcolorbox}
\usepackage{xcolor}
\usepackage{wasysym}
\usepackage{mathtools}
\usepackage{scalerel}
\usepackage{bbm}

\usepackage{tikz}
\usetikzlibrary{arrows,backgrounds,patterns.meta}
\usetikzlibrary{positioning,shadings,cd}
\usetikzlibrary{shapes}
\usetikzlibrary{backgrounds}
\usetikzlibrary{decorations,decorations.pathreplacing,decorations.markings,decorations.pathmorphing}
\usetikzlibrary{fit,calc,through}
\usetikzlibrary{external}
\newcommand{\tikzmath}[2][]
{\vcenter{\hbox{\begin{tikzpicture}[#1]#2\end{tikzpicture}}}
}
\newcommand{\roundNbox}[6]{
	\draw[rounded corners=5pt, very thick, #1] ($#2+(-#3,-#3)+(-#4,0)$) rectangle ($#2+(#3,#3)+(#5,0)$);
	\coordinate (ZZa) at ($#2+(-#4,0)$);
	\coordinate (ZZb) at ($#2+(#5,0)$);
	\node at ($1/2*(ZZa)+1/2*(ZZb)$) {#6};
}
\tikzset{super thick/.style={line width=3pt}}
\tikzstyle{far>}=[decoration={markings, mark=at position 0.75 with {\arrow{>}}}, postaction={decorate}]
\tikzstyle{mid>}=[decoration={markings, mark=at position 0.55 with {\arrow{>}}}, postaction={decorate}]
\tikzstyle{mid<}=[decoration={markings, mark=at position 0.55 with {\arrow{<}}}, postaction={decorate}]
\tikzset{super thick/.style={line width=3pt}}
\tikzstyle{far>}=[decoration={markings, mark=at position 0.75 with {\arrow{>}}}, postaction={decorate}]
\tikzstyle{mid>}=[decoration={markings, mark=at position 0.55 with {\arrow{>}}}, postaction={decorate}]
\tikzstyle{mid<}=[decoration={markings, mark=at position 0.55 with {\arrow{<}}}, postaction={decorate}]
\tikzstyle{knot}=[preaction={super thick, white, draw}]
\tikzstyle{coupon}=[draw, very thick, rectangle, rounded corners=5pt]
\tikzset{Rightarrow/.style={double equal sign distance,>={Implies},->},
triplecd/.style={-,preaction={draw,Rightarrow}},
quadruplecd/.style={preaction={draw,Rightarrow,
shorten >=0pt
},
shorten >=1pt,
-,double,double
distance=0.2pt}}
\tikzset{
    tripleline/.style args={[#1] in [#2] in [#3]}{
        #1,preaction={preaction={draw,#3},draw,#2}
    }
}
\tikzstyle{triple}=[tripleline={[line width=.15mm,black] in
      [line width=.7mm,white] in
      [line width=1mm,black]}] 
\tikzset{
    quadrupleline/.style args={[#1] in [#2] in [#3] in [#4]}{
        #1,preaction={preaction={preaction={draw,#4},draw,#3}, draw,#2}
    }
}
\tikzstyle{quadruple}=[quadrupleline={[line width=.3mm,white] in
      [line width=.6mm,black] in
      [line width=1.2mm,white] in
      [line width=1.5mm,black]}]

\definecolor{violet}{RGB}{148,0,211}
\definecolor{DarkGreen}{RGB}{0,150,0}
\definecolor{rufous}{HTML}{A81C07}

\usepackage[plainpages=false,hypertexnames=false,pdfpagelabels]{hyperref}
\definecolor{medium-blue}{rgb}{0,0,.8}
\hypersetup{colorlinks, linkcolor={purple}, citecolor={medium-blue}, urlcolor={medium-blue}}
\newcommand{\arxiv}[1]{\href{http://arxiv.org/abs/#1}{\tt arXiv:\nolinkurl{#1}}}
\newcommand{\arXiv}[1]{\href{http://arxiv.org/abs/#1}{\tt arXiv:\nolinkurl{#1}}}

\newcommand{\doi}[1]{\href{http://dx.doi.org/#1}{{\tt DOI:#1}}}

\DeclareMathOperator{\Ad}{Ad}
\DeclareMathOperator{\Aut}{Aut}
\DeclareMathOperator{\coev}{coev}
\DeclareMathOperator{\End}{End}
\DeclareMathOperator{\eval}{eval}

\DeclareMathOperator{\ev}{ev}

\DeclareMathOperator{\Hom}{Hom}
\DeclareMathOperator{\id}{id}

\DeclareMathOperator{\Irr}{Irr}

\DeclareMathOperator{\op}{op}
\DeclareMathOperator{\Path}{Path}

\DeclareMathOperator{\spec}{spec}

\DeclareMathOperator{\Tr}{Tr}
\DeclareMathOperator{\tr}{tr}

\newcommand{\set}[2]{\left\{#1 \middle| #2\right\}}
\newcommand{\pbk}[2]{\left| #1 \middle\rangle \middle\langle #2\right|}

\newcommand{\DHR}{\mathsf{DHR}}
\newcommand{\Fun}{\mathsf{Fun}}
\newcommand{\fdHilb}{\mathsf{Hilb_{fd}}}

\newcommand{\Vect}{\mathsf{Vect}}

\newcommand{\Rep}{\mathsf{Rep}}

\def\semicolon{;}
\def\applytolist#1{
    \expandafter\def\csname multi#1\endcsname##1{
        \def\multiack{##1}\ifx\multiack\semicolon
            \def\next{\relax}
        \else
            \csname #1\endcsname{##1}
            \def\next{\csname multi#1\endcsname}
        \fi
        \next}
    \csname multi#1\endcsname}

\def\calc#1{\expandafter\def\csname c#1\endcsname{{\mathcal #1}}}
\applytolist{calc}QWERTYUIOPLKJHGFDSAZXCVBNM;
\def\bbc#1{\expandafter\def\csname bb#1\endcsname{{\mathbb #1}}}
\applytolist{bbc}QWERTYUIOPLKJHGFDSAZXCVBNM;
\def\bfc#1{\expandafter\def\csname bf#1\endcsname{{\mathbf #1}}}
\applytolist{bfc}QWERTYUIOPLKJHGFDSAZXCVBNM;
\def\sfc#1{\expandafter\def\csname s#1\endcsname{{\sf #1}}}
\applytolist{sfc}QWERTYUIOPLKJHGFDSAZXCVBNM;
\def\fc#1{\expandafter\def\csname f#1\endcsname{{\mathfrak #1}}}
\applytolist{fc}QWERTYUIOPLKJHGFDSAZXCVBNM;
\def\rmc#1{\expandafter\def\csname rm#1\endcsname{{\mathrm #1}}}
\applytolist{rmc}QWERTYUIOPLKJHGFDSAZXCVBNM;

\numberwithin{equation}{section}

\theoremstyle{plain}
\newtheorem{thm}[equation]{Theorem}
\newtheorem*{thm*}{Theorem}
\newtheorem{cor}[equation]{Corollary}
\newtheorem{lem}[equation]{Lemma}
\newtheorem{prop}[equation]{Proposition}
\newtheorem{conj}[equation]{Conjecture}

\newtheorem*{claim*}{Claim}

\newtheorem{thmalpha}{Theorem}

\newtheorem{coralpha}[thmalpha]{Corollary}
\theoremstyle{definition}
\newtheorem{defn}[equation]{Definition}
\newtheorem{algorithm}[equation]{Algorithm}

\newtheorem*{trick*}{Trick}
\newtheorem{construction}[equation]{Construction}
\newtheorem{nota}[equation]{Notation}

\newtheorem{ex}[equation]{Example}

\newtheorem{rem}[equation]{Remark}

\title{Local topological order and boundary algebras}
\date{\today}
\begin{document}
\author{Corey Jones$^1$}
\address{$^1$ Department of Mathematics, North Carolina State University, Raleigh, NC 27695, USA}
\author{Pieter Naaijkens$^2$}
\address{$^2$ School of Mathematics, Cardiff University, Cardiff, CF24 4AG, United Kingdom}
\author{David Penneys$^3$}
\author{Daniel Wallick$^3$}
\address{$^3$ Department of Mathematics, The Ohio State University, Columbus, OH 43210, USA}
\makeatletter
\let\@wraptoccontribs\wraptoccontribs
\makeatother
\contrib[with an appendix by]{Masaki Izumi$^4$}
\address{$^4$ Graduate School of Science, Kyoto University, Sakyo-ku, Kyoto 606-8502, Japan}

\begin{abstract}
We introduce a set of axioms for locally topologically ordered quantum spin systems in terms of nets of local ground state projections, and we show they are satisfied by Kitaev's Toric Code and Levin-Wen type models. 
For a locally topologically ordered spin system on $\mathbb{Z}^{k}$, we define a local net of boundary algebras on $\mathbb{Z}^{k-1}$, which provides a mathematically precise algebraic description of the holographic dual of the bulk topological order.
We construct a canonical quantum channel so that states on the boundary quasi-local algebra parameterize bulk-boundary states without reference to a boundary Hamiltonian. 
As a corollary, we obtain a new proof of a recent result of Ogata [Ann. H. Poincar\'e 25, 2024] that the bulk cone von Neumann algebra in the Toric Code is of type $\rm{II}$, and we show that Levin-Wen models can have cone algebras of type $\rm{III}$.
Finally, we argue that the braided tensor category of DHR bimodules for the net of boundary algebras characterizes the bulk topological order in (2+1)D, and can also be used to characterize the topological order of boundary states.
\end{abstract}
\maketitle
\tableofcontents
\section{Introduction}

In 2+1 dimensions, topologically ordered spin systems display a number of interesting phenomena, from non-trivial braiding statistics of quasi-particles to robust error correction properties. 
The most widely studied class of topological spin systems are exactly solvable with commuting projector Hamiltonians \cite{cond-mat/0506438,PhysRevB.71.045110}. These have the property that the useful error correction features of the system are present in the local ground state spaces, which has led to these systems being called $\textit{local topological ordered}$.

In this article, we propose an axiomatization for local topological order (LTO) in terms of nets of projections in the quasi-local algebra (Definition \ref{defn:StrongLocalTopologicalOrder}).
Our axioms are stronger than previous axiomatizations of topological quantum order (the `TQO conditions')~\cite{MR2742836,MR2842961}, but we show that our stronger axioms hold for the Levin--Wen models~\cite{PhysRevB.71.045110} and Kitaev's Toric Code~\cite{MR1951039}.\footnote{Our stronger axioms were recently shown to hold for Kitaev's Quantum Double model in \cite{MR4814524} and for twisted quantum double models in \cite{2411.08675}.}
The primary motivation for our stronger axiomatization is that LTOs in our sense give rise to a local net of boundary algebras (Construction \ref{const:LTO-BoundaryNet}). These are nets of $\rm C^*$-algebras defined on a lattice in one spatial dimension lower than the original LTO. 
In general, the local boundary algebras do not factorize as tensor products of algebras assigned to sites, and thus carry important topological information about the bulk system. 
In general, they do not embed as unital subalgebras of the original quasi-local algebra. 

The first main result of our paper is that Kitaev's Toric Code model and the Levin-Wen models satisfy our axiomatization. 
This provides a new proof (in the case of the Levin-Wen model) for the quantum error correction properties for these classes~\cite{MR1951039,Cui2020kitaevsquantum,QIU2020168318}.  
We also identify the boundary nets as \emph{fusion categorical nets}.
Such nets first emerged from subfactor theory ~\cite{MR1642584, MR3406647}, and have recently also found use in applications to topological phases of matter~\cite{2211.03777,2112.09091,MR4109480,MR4272039,MR4814692} and have connections to conformal field theory \cite{MR3719546,PhysRevLett.128.231603,PhysRevLett.128.231602,MR4580529}.

\begin{thmalpha}
\label{thm:LTOExamples}
The Toric Code and Levin-Wen models satisfy the LTO axioms \ref{LTO:Hastings}--\ref{LTO:Injective}.
The boundary nets are fusion categorical nets over the lattice $\bbZ$.
\end{thmalpha}

Our primary motivation for the stronger axiomatization of local topological order is that the resulting boundary net gives us a powerful tool to rigorously analyze the entire system in at least two ways.

\begin{enumerate}
\item
States on the boundary algebra correspond to states on the bulk-boundary system which restrict to the canonical ground state in the bulk. 
This gives a Hamiltonian-free approach to boundary states.
\item 
The boundary algebra of an LTO can be viewed as a discrete algebraic quantum field theory in one dimension lower, holographically dual to the bulk theory. 
The topological order of the bulk should then be algebraically characterized by the category of DHR bimodules of the net of boundary algebras (see Conjecture \ref{conj:BulkBoundary}), giving a precise mathematical formulation of topological holography in the sense of \cite{2310.05790, PhysRevB.107.155136}. 
\end{enumerate}

With a concrete description of the boundary nets in hand given by Theorem \ref{thm:LTOExamples}, we can explicitly study boundary states and their relation to topological order. 
There are several salient boundary states to investigate. 
Of particular interest are the canonical boundary states, obtained from the LTO axioms by simply `compressing' (that is, projecting the observables down to a suitable subspace) the canonical bulk ground states. 
We show that these states can be viewed as equilibrium states (or more properly, KMS states) for natural locally representable 1-parameter automorphism groups. 
We use this to prove the following surprising theorem.

\begin{thmalpha}
For a Levin-Wen model over a fusion category $\mathcal{C}$, the canonical boundary state is a factor state on the quasi-local algebra of the boundary. 
The corresponding factor is of type $\rm II_{1}$ if and only if all simple objects in $\mathcal{C}$ are invertible ($d_c=1$ for all $c\in \Irr(\cC)$); otherwise, it is of type $\rm III$.
\end{thmalpha}

In the approach to the superselection theory of topologically ordered spin systems introduced in \cite{MR2804555, MR3426207}, a fundamental role is played by the \emph{cone algebras}. 
These are the von Neumann algebras constructed from completing (in the weak or strong operator topology) the $\rm C^*$-algebras assigned to an infinite cone in the GNS representation of the bulk ground state.
Ogata showed that for Kitaev quantum double models, the cone algebras with rough edges are type $\rm II_{\infty}$~\cite{MR4721705}.\footnote{Ogata only claims the result for abelian quantum double models, but the result can be obtained in the more general setting by combining her proof with remarks made in \cite{MR3426207}.} 
As an application of our previous theorem, we have the following result.

\begin{coralpha}
\label{cor:ConeAlgebraTypes}
For Levin-Wen models over the fusion category $\mathcal{C}$, the cone algebras are of type $\rm II_\infty$ if and only if all simple objects in $\mathcal{C}$ are invertible ($d_c=1$ for all $c\in \Irr(\cC)$); otherwise they are of type $\rm III$.
\end{coralpha}

This should be contrasted with the conformal nets or local algebras in algebraic quantum field theory, which are generically type $\rm III$ independent of the corresponding category of superselection sectors.
In the appendix, authored by Masaki Izumi, the analysis of the type $\rm III$ case is refined to obtain the type ${\rm III}_\lambda$.
We note that the type of the von Neumann algebra is of relevance in the context of embezzlement of entanglement~\cite{2401.07299}.
In particular, any normal state on a type ${\rm III}_1$-factor is a so-called \emph{universal embezzler}.

One of the main motivations for studying the boundary nets is topological holography.
Recall that anyonic excitations in the Levin-Wen model for $\cC$ are described by the quantum double/Drinfeld center $Z(\cC)$.
If the boundary algebra is truly a holographic dual of the bulk theory, it should recover the topological order. 
We show this for (2+1)D Levin-Wen models based on a fusion category $\cC$ and for Kitaev's Toric Code model on the plane.\footnote{This is was also recently proven for Kitaev's Quantum Double model in \cite{MR4814524}.}
For a net of algebras $\fF$ over $\mathbbm{Z}^{n}$, there is braided tensor category $\DHR(\fF)$ of \emph{DHR bimodules}~\cite{MR4814692}, inspired by the Doplicher--Haag--Roberts theory of superselection sectors (see~\cite{MR1405610,math-ph/0602036} for an introduction). For nets of algebras $\fF$ built from a fusion category $\cC$, $\DHR(\fF)\cong Z(\cC)$.
Our characterization of the boundary algebras in the Levin-Wen and Toric Code models leads to the following corollary:
\begin{coralpha}
    \label{cor:BulkBoundary}
    Let $\fF$ be the boundary net of the Levin-Wen model over the fusion category $\cC$ (respectively the boundary net for the Toric Code).
    Then $\DHR(\fF) \cong Z(\cC)$ (with $\cC = \fdHilb(\mathbb{Z}/2)$ for the Toric Code model).
\end{coralpha}
This result immediately leads to an interesting observation in light of Corollary \ref{cor:ConeAlgebraTypes}, given that inequivalent fusion categories can have the same Drinfeld center $Z(\cC)$.
In particular, there are two Levin-Wen models that yield the Drinfeld double $\cD(G)$ topological order for a finite group $G$: $\cC=\fdHilb(G)$ and $\cC=\Rep(G)$.
The first is pointed (all simple objects are invertible), but the second is only pointed when $G$ is abelian.
Thus, the type of the cone algebras is specific to the model, and does not only depend on the bulk topological order.
We come back to this point later in Remark~\ref{rem:FiniteDepthQuantumCircuit}.

Finally, we can use the above to give a categorical analysis of the superselection theory of boundary states. 
Given a boundary state $\phi$ on the boundary net $\fB$, its superselection category, called the `boundary order,' is the category of representations of the boundary quasi-local algebra that are quasi-contained in the GNS representation of $\phi$ restriced to the algebras of operators localized outside any sufficiently large interval (see Definition~\ref{defn:BoundarySuperSelection}). 
The superselection category naturally forms a module category over $\DHR(\fB)$, and choosing the GNS representation as a distinguished object, taking internal end produces a $\rm W^*$-algebra object $A_{\phi}$ internal to the unitary tensor category $\DHR(\fB)$ in the sense of \cite{MR3687214}.
We say the boundary order of a state $\phi$ is \textit{topological} if $A_{\phi}$ is a Lagrangian algebra object~\cite[Defn. 4.6]{MR3039775}). 
This is the case we expect to match with `gapped boundaries' in the usual sense \cite{MR2942952}.
We show that for the boundary state associated to the vacuum in Levin-Wen type models built from the fusion category $\cC$, this algebra is indeed the canonical Lagrangian algebra in $\DHR(\fF)\cong Z(\cC)$.

The examples we consider in this paper are mostly restricted to 2D topological order,
but we note that our nets of boundary algebras work in arbitrary dimensions.
In higher dimensions, topological order is characterized by a braided fusion $n$-category rather than simply a braided fusion category \cite{MR4444089}.
Thus the tensor category of DHR bimodules at the boundary is not sufficient to fully characterize the topological order. 
However, we expect that the category of DHR bimodules has a higher categorical generalization, which we plan to pursue in future work.

In \S\ref{sec:LTO}, we introduce the basic setup and our local topological order axioms.
We use these axioms to show that we have a canonical state in the bulk, and a quantum channel from boundary to bulk states.
The reader mostly interested in applications to physical systems may wish to skip over these operator-algebraic constructions at first reading, and jump straight to \S\ref{sec:toric} and \S\ref{sec:LWStringNet}, where we discuss the Toric Code and Levin-Wen examples.
In \S\ref{sec:BoundaryStates} we take a closer look at states on the boundary algebra, and study the type of the von Neumann algebra generated by the `canonical' boundary state.
Finally, details on the bulk-boundary correspondence are given in~\S\ref{sec:BulkBoundaryCorrespondence}.

\subsection*{Acknowledgements}

The authors would like to thank
Dave Aasen,
Sven Bachmann,
Jeongwan Haah,
Peter Huston,
Theo Johnson-Freyd,
Kyle Kawagoe,
Brent Nelson,
David P{\'e}rez-Garc{\'i}a
and
Shuqi Wei
for helpful conversations.
Corey Jones was supported by NSF grant DMS 2100531.
David Penneys and Daniel Wallick were supported by NSF grant DMS 2154389.

\section{Nets of algebras and local topological order}
\label{sec:LTO}

In this section, we work with nets of $\rm C^*$-algebras on square lattices, i.e., $\bbZ^\ell$ for some $\ell$.
The methods here work in greater generality, but passing to more general lattices would require substantially more space and heavier notation.
We therefore restrict to the simpler case for clarity.

\subsection{Nets of algebras and nets of projections}

\begin{defn}[Nets of algebras]
\label{defn:LocalNet}
Suppose $\cL$ is the $\bbZ^\ell$ lattice.
An $\ell$D \emph{(local) net of algebras}\footnote{A local net of algebras such that $\fA(\Lambda)$ is finite dimensional for all $\Lambda$ could be called an \emph{abstract spin system}, cf.~the following Example~\ref{ex:SpinSystem} of a (concrete) quantum spin system where the local algebras $\fA(\Lambda)$ are tensor products of full matrix algebras.} 
on $\cL$ in the (unital) ambient $\rm C^*$-algebra $\fA$ (called the \emph{quasi-local algebra}) 
is an assignment of a $\rm C^*$-subalgebra $\fA(\Lambda)\subset \fA$ to each bounded rectangle $\Lambda\subset \cL$ such that
\begin{enumerate}[label=(N\arabic*)]
    \item $\fA(\emptyset)=\bbC 1_{\fA}$,
    \label{N:empty}
    \item 
    \label{N:inclusion}
    if $\Lambda\subset \Delta$, then $\fA(\Lambda)\subset \fA(\Delta)$,
    \item 
    \label{N:local}
    if $\Lambda\cap \Delta= \emptyset$, then $[\fA(\Lambda), \fA(\Delta)]=0$, and
    \item 
    \label{N:dense}
    $\bigcup_\Lambda \fA(\Lambda)$ is norm dense in $\fA$.
\end{enumerate}
The first and second conditions above are equivalent to the data of a functor from the poset of rectangles in $\cL$ ordered by inclusion to the poset of unital $\rm C^*$-subalgebras of $\fA$ ordered by inclusion.
\end{defn}

We will only be considering nets of algebras which satisfy the locality condition~\ref{N:local}, and will subsequently drop the adjective `local'.

\begin{rem}
\label{rem:SufficientlyLargeRectangle}
In some circumstances, we are not concerned with all $\ell$D rectangles, but only rectangles $\Lambda$ which are \emph{sufficiently large}, meaning there is a global constant $r>0$ such that $\Lambda$ contains a closed $r^\ell$-cube.
An assignment of $\rm C^*$-subalgebras $\fA(\Lambda)\subset \fA$ for sufficiently large rectangles $\Lambda$ satisfying \ref{N:inclusion}, \ref{N:local}, and \ref{N:dense} can be canonically augmented to a net of algebras for all rectangles by defining $\fA(\Delta):=\bbC 1_\fA$ whenever $\Delta$ is not sufficiently large.
\end{rem}

\begin{defn}
Let $\cL$ be the $\bbZ^\ell$ lattice, and write $\Aut_{\tr}(\cL)$ for its group of translation symmetries, where we write $\Lambda \mapsto g+\Lambda$ for $g\in \Aut_{\tr}(\cL)$.
A net of algebras $\Lambda \mapsto \fA(\Lambda)$ is called \emph{translation invariant} if there is an $\Aut_{\tr}(\cL)$-action on $\fA$ by unital $*$-automorphisms such that $g\cdot\fA(\Lambda) = \fA(g+\Lambda)$ for all $g\in \Aut_{\tr}(\cL)$.
\end{defn}

\begin{ex}
\label{ex:SpinSystem}
The canonical example of a net of algebras on a lattice is a \emph{quantum spin system}. 
We take the lattice $\cL=\bbZ^{\ell}$, and at each site/vertex, we view a copy of $\bbC^d$.
For each bounded rectangle $\Lambda \subset \cL$, we define $\fA(\Lambda):=\bigotimes_{v\in \Lambda} M_d(\bbC)$.
When $\Lambda \subset \Delta$, we have obvious inclusions $\fA(\Lambda)\subseteq \fA(\Delta)$. 
The quasi-local algebra $\fA$ is the colimit of this directed system in the category of $\rm C^*$-algebras.
Observe that $\fA$ carries a canonical action of $\bbZ^\ell=\Aut_{\tr}(\bbZ^\ell)$, and translation invariance means with respect to this canonical action.
\end{ex}

\begin{ex}
\label{ex:TensorCat}
We define a 1D net $I\mapsto \fF(I)$ from a unitary tensor category $\cC$ and a choice of object $X\in \cC$. 
For any interval $I\subseteq \mathbbm{Z}$ with $n$ points, we set $\fF(I):=\End_\cC(X^{n})$. 
If $I\subseteq J$, there are natural inclusions $\fF(I)\hookrightarrow \fF(J)$ given by tensoring with $\id_{X}$ on the left and/or right depending on the relative position of $I$ and $J$. 
The quasi-local algebra $\fF$ is the colimit in the category of unital $\rm C^*$-algebras, and we identify $\fF(I)$ with its image in this colimit. 
That is, we regard $\fF(I)$ as a $\rm{C}^*$-subalgebra of $\fF$.

When $\cC$ is a unitary fusion category, we call nets constructed in this way \textit{fusion categorical nets}. 
We will see such nets arise as the net of boundary algebras for Levin-Wen models (see \S\ref{sec:LWStringNet}, Remark~\ref{rem:IdentifyBoundaryLW}) and the Toric Code (in \S\ref{sec:toric}, Construction~\ref{construction:ToricCodeInjective}).
We note that if $X$ strongly generates the fusion category (i.e., there exists an $n$ such that every simple is isomorphic to a direct summand of $X^n$), then the net of algebras $\fF$ satisfies \emph{weak algebraic Haag duality} by~\cite[Prop. 4.3]{MR4814692}, which allows one to define and identify its category of DHR bimodules; we refer the reader to \S\ref{sec:BulkBoundaryCorrespondence} below for more details.
\end{ex}

A natural equivalence relation between nets of algebras on a lattice is \textit{bounded spread} isomorphism.
Many interesting properties, such as the category of DHR bimodules that we introduce later, are preserved under bounded spread isomorphism. This seems to be a good notion of equivalence for discrete nets, so we include it here.

\begin{defn}[Bounded spread isomorphism]
\label{defn:BoundedSpread}
Suppose we have two nets of algebras, 
$\Lambda\mapsto \fA(\Lambda)$ in $\fA$ 
and 
$\Lambda\mapsto \fB(\Lambda)$ in $\fB$, 
on the same lattice $\cL=\bbZ^\ell$.
A unital $*$-algebra isomorphism $\Psi:\fA\to \fB$ is said to have \emph{bounded spread} if there is an $s>0$ such that 
$\Psi(\fA(\Lambda)) \subseteq \fB(\Lambda^{+s})$ 
and
$\Psi^{-1}(\fB(\Lambda)) \subseteq \fA(\Lambda^{+s})$ 
for all rectangles $\Lambda$, where $\Lambda^{+s}$ is the smallest rectangle containing $\Lambda$ and all points of distance at most $s$.
\end{defn}


\begin{nota}
Suppose $\cL$ is the $\bbZ^\ell$ lattice.
We write $\partial \Lambda$ for the vertices (or sites) at the boundary of a rectangle $\Lambda$.
We say that a rectangle \emph{$\Delta$ surrounds $\Lambda$ by $s>0$}
if 
\begin{itemize}
\item 
$\Lambda \subset \Delta$,
\item
$\partial \Lambda \cap \partial \Delta$ is either empty or an $(\ell-1)$D rectangle, and
\item
Every vertex $v\in\Delta\setminus \Lambda$ is contained in some closed $s^\ell$-cube contained entirely in $\Delta\setminus \Lambda$.
\end{itemize}
If $\partial\Lambda \cap \partial \Delta=\emptyset$, we say \emph{$\Delta$ completely surrounds $\Lambda$ by $s>0$}, and we denote this by $\Lambda \ll_s \Delta$.
If $\partial\Lambda \cap \partial \Delta$ is non-empty, we denote this by $\Lambda \Subset_s \Delta$.
Here are two examples where $s=2$ and $\ell = 2$:
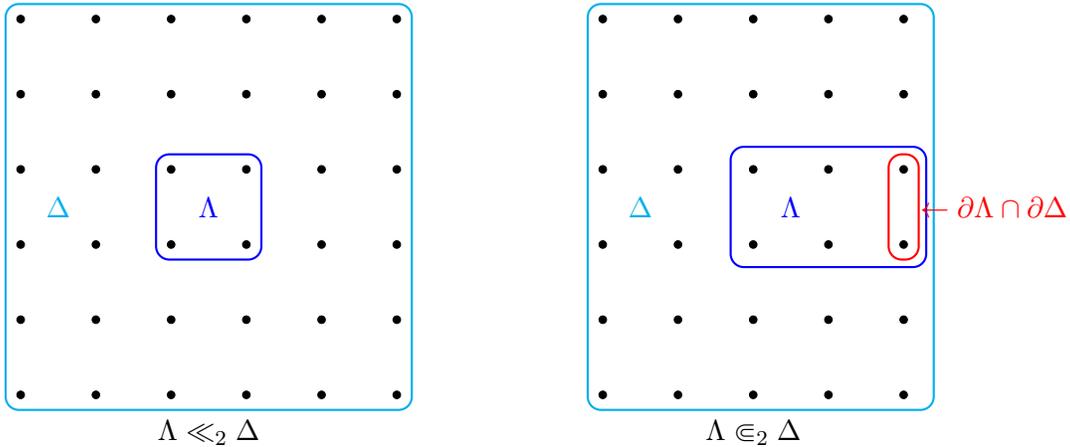
$$
\tikzmath{
\foreach \y in {0,1,2,3,4,5}{
\foreach \x in {0,1,2,3,4,5}{
\filldraw (\x,\y) circle (.05cm);
}
}
\draw[thick, blue, rounded corners=5pt] (3.2,1.8) rectangle (1.8,3.2);
\node[blue] at (2.5,2.5) {$\Lambda$};
\draw[thick, cyan, rounded corners=5pt] (-.2,-.2) rectangle (5.2,5.2);
\node[cyan] at (.5,2.5) {$\Delta$};
\node at (2.5,-.5) {$\Lambda\ll_2 \Delta$};
}
\qquad\qquad\qquad
\tikzmath{
\foreach \y in {0,1,2,3,4,5}{
\foreach \x in {0,1,2,3,4}{
\filldraw (\x,\y) circle (.05cm);
}
}
\draw[thick, blue, rounded corners=5pt] (4.3,1.7) rectangle (1.7,3.3);
\node[blue] at (2.5,2.5) {$\Lambda$};
\draw[thick, cyan, rounded corners=5pt] (-.2,-.2) rectangle (4.4,5.2);
\node[cyan] at (.5,2.5) {$\Delta$};
\draw[thick, red, rounded corners=5pt] (4.2,1.8) rectangle (3.8,3.2);
\node[red] at (5.2,2.5) {$\leftarrow \partial \Lambda \cap \partial \Delta$};
\node at (2,-.5) {$\Lambda\Subset_2 \Delta$};
}
$$
\end{nota}

\begin{defn}
Suppose we have a translation-invariant net of algebras $\Lambda\mapsto \fA(\Lambda) \subset \fA$ on the lattice $\cL=\bbZ^\ell $.
By convention, whenever $\Lambda$ is not sufficiently large, $\fA(\Lambda)=\bbC 1_\fA$, and thus for the empty rectangle, $\fA(\emptyset):=\bbC 1_\fA$.

A \emph{net of projections} on $\cL$ is an assignment of a non-zero orthogonal projection $p_\Lambda\in \fA(\Lambda)$ associated to every rectangle $\Lambda$ in $\cL$ ordered by reverse inclusion, i.e., $\Lambda \subset \Delta$ implies $p_\Delta \leq p_\Lambda$.
A net of projections is called \emph{translation invariant} if $g\cdot p_\Lambda = p_{g+\Lambda}$ for all rectangles $\Lambda$.
\end{defn}

\subsection{Local topological order axioms}

In this section, we assume $\fA$ is a translation invariant net of algebras and $p=(p_\Lambda)$ is a translation invariant net of projections on $\cL=\bbZ^\ell$.

The following algebras play an important role in the local topological order conditions below.

\begin{defn}
Given $s > 0$ and subsets $\Lambda$ and $\Delta$ of $\cL$ with $\Lambda \Subset_s \Delta$, we define
$$
\fB(\Lambda \Subset_s \Delta)
:=
\set{xp_\Delta}{x\in p_{\Lambda} \fA(\Lambda)p_\Lambda\text{ and }xp_{\Delta'}=p_{\Delta'}x\text{ whenever }\Lambda \Subset_s \Delta'\text{ with }\partial \Lambda \cap \partial \Delta' = \partial \Lambda \cap \partial \Delta}.
$$
\end{defn}

\noindent Observe that $\fB(\Lambda \Subset_s \Delta)$ is a unital $*$-algebra with unit $p_\Delta$. Similar algebras were considered for annular regions in \cite{MR3465431,MR4071117}.

We now have introduced all necessary notation to state the local topological order axioms, which are the main object of study in this paper.
\begin{defn}
\label{defn:StrongLocalTopologicalOrder}
We say $(\fA,p)$ is \emph{locally topologically ordered} if it satisfies the following four axioms for sufficiently large $\Lambda$ (for $r>0$) and a globally fixed `surrounding constant' $s>0$:
\begin{enumerate}[label=(LTO\arabic*)]
\item 
\label{LTO:Hastings}
Whenever $\Lambda \ll_s \Delta$, $p_\Delta \fA(\Lambda)p_\Delta = \bbC p_\Delta$.
\item 
\label{LTO:Boundary}
Whenever $\Lambda \Subset_s \Delta$, $p_\Delta \fA(\Lambda)p_\Delta 
= 
\fB(\Lambda \Subset_s \Delta)$
(which is equal to 
$\fB(\Lambda \Subset_s \Delta) p_\Delta$).
\item
\label{LTO:Surjective}
Whenever
$\Lambda_1\subset  \Lambda_2\Subset_s \Delta$ with 
$\partial \Lambda_1 \cap \partial\Delta= \partial \Lambda_2 \cap \partial \Delta$,
$\fB(\Lambda_1 \Subset_s \Delta)=\fB(\Lambda_2 \Subset_s \Delta)$.
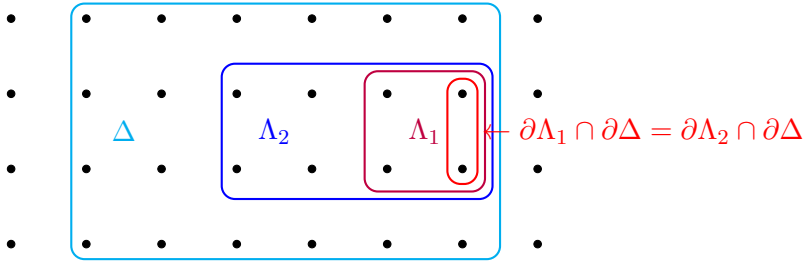
$$
\tikzmath{
\foreach \y in {0,1,2,3,4,5}{
\foreach \x in {-3,-2,-1,0,1,2,3,4}{
\filldraw (\x,\y) circle (.05cm);
}
}
\draw[thick, purple, rounded corners=5pt] (3.3,1.7) rectangle (1.7,3.3);
\node[purple] at (2.5,2.5) {$\Lambda_1$};
\draw[thick, blue, rounded corners=5pt] (3.4,1.6) rectangle (-.2,3.4);
\node[blue] at (.5,2.5) {$\Lambda_2$};
\draw[thick, cyan, rounded corners=5pt] (3.5,.8) rectangle (-2.2,4.2);
\node[cyan] at (-1.5,2.5) {$\Delta$};
\draw[thick, red, rounded corners=5pt] (3.2,1.8) rectangle (2.8,3.2);
\node[red] at (5.4,2.5) {$\leftarrow \partial \Lambda_1 \cap \partial \Delta=\partial \Lambda_2\cap \partial \Delta$};
}
$$
\item
\label{LTO:Injective}
Whenever $\Lambda \Subset_s \Delta_1\subset \Delta_2$ with 
$\partial \Lambda \cap \partial\Delta_1= \partial \Lambda \cap \partial \Delta_2$,
if $x\in \fB(\Lambda \Subset_s \Delta_1)$ with $xp_{\Delta_2}=0$, then $x=0$.
$$
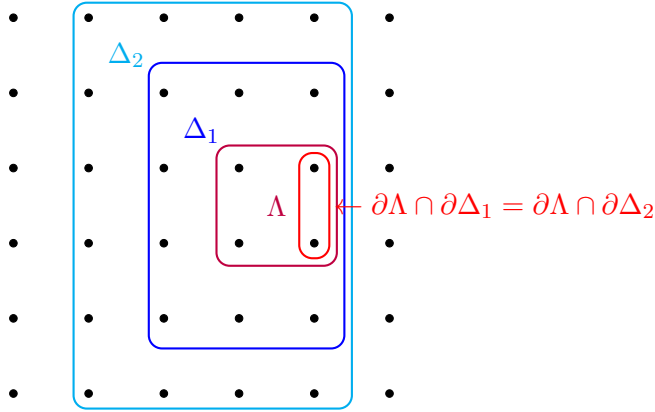
\tikzmath{
\foreach \y in {-1,0,1,2,3,4,5,6}{
\foreach \x in {-1,0,1,2,3,4}{
\filldraw (\x,\y) circle (.05cm);
}
}
\draw[thick, purple, rounded corners=5pt] (3.3,1.7) rectangle (1.7,3.3);
\node[purple] at (2.5,2.5) {$\Lambda$};
\draw[thick, blue, rounded corners=5pt] (3.4,.6) rectangle (.8,4.4);
\node[blue] at (1.5,3.5) {$\Delta_1$};
\draw[thick, cyan, rounded corners=5pt] (3.5,-.2) rectangle (-.2,5.2);
\node[cyan] at (.5,4.5) {$\Delta_2$};
\draw[thick, red, rounded corners=5pt] (3.2,1.8) rectangle (2.8,3.2);
\node[red] at (5.4,2.5) {$\leftarrow \partial \Lambda \cap \partial \Delta_1=\partial \Lambda\cap \partial \Delta_2$};
}
$$
\end{enumerate}
Observe that the algebra $\fB(\Lambda \Subset_s \Delta)$ in \ref{LTO:Boundary} plays the role of $\bbC$ from \ref{LTO:Hastings}.
\end{defn}

The first condition \ref{LTO:Hastings} implies the topological quantum order conditions (TQO1) and (TQO2) of~\cite{MR2742836}, and hence the LTO axioms are stronger (see Proposition~\ref{prop:LTQO} below).
Nevertheless, these stronger conditions are satisfied by the Toric Code and Levin-Wen examples, as we show in \S\ref{sec:toric} and \S\ref{sec:LWStringNet} respectively below.
The remaining three LTO axioms allow us to talk about operators acting along a boundary of a region that are `compatible' with the bulk.
In concrete examples, these are the operators that create excitations along the boundary only, but not in the bulk.
The definition of boundary algebras, however, depends on the choice of region $\Delta$, which should not matter as long as $\Delta$ is large enough, and we should be able to identify the algebras for different choices of $\Delta$.
Axioms~\ref{LTO:Surjective} and~\ref{LTO:Injective} guarantee that this identification can be done consistently, allowing us to define an abstract local net of boundary operators.

\begin{ex}
Our first example in this article is Kitaev's Toric Code~\cite{MR1951039}, which appears in \S\ref{sec:toric} below.
\end{ex}

\begin{ex}
Our main example in this article is the (2+1)D Levin-Wen string net model~\cite{PhysRevB.71.045110} associated to a unitary fusion category, which appears in \S\ref{sec:LWStringNet} below.
\end{ex}

For the above models,
we verify a simplified/stronger version of \ref{LTO:Boundary}--\ref{LTO:Injective} based on the following observation.

\begin{lem}
\label{lem:StabilityOfBoundaryAlgebras}
Suppose we have a quadrilateral of rectangles
\[
\begin{tikzcd}[column sep=0, row sep =0]
\Lambda_2 & \Subset_s & \Delta_2
\\
\cup &&\cup
\\
\Lambda_1 & \Subset_s & \Delta_1
\end{tikzcd}
\qquad\qquad
\text{such that} 
\qquad
\partial \Lambda_1 \cap \partial \Delta_1 = \partial \Lambda_1\cap \partial \Delta_2.
\]
The map $\fB(\Lambda_1\Subset_s \Delta_1) \to \fB(\Lambda_2\Subset_s \Delta_2)$ given by multiplication by $p_{\Delta_2}$ is an injective $*$-algebra map onto $\fB(\Lambda_1\Subset_s \Delta_2)$. 
If moreover $\partial \Lambda_1 \cap \partial \Delta_1 = \partial \Lambda_2\cap \partial \Delta_2$, then this map is an isomorphism.
\end{lem}
\begin{proof}
By definition, $p_{\Delta_2}$ commutes with $\fB(\Lambda_1 \Subset_s \Delta_1)=p_{\Delta_1}
\fA(\Lambda_1)
p_{\Delta_1}$ (by \ref{LTO:Boundary}) inside $\fA(\Delta_2)$.
Again by \ref{LTO:Boundary}, we see
$$
\fB(\Lambda_1 \Subset_s \Delta_1)p_{\Delta_2}
\underset{\text{\ref{LTO:Boundary}}}{=}
p_{\Delta_2}
p_{\Delta_1}
\fA(\Lambda_1)
p_{\Delta_1}
p_{\Delta_2}
=
p_{\Delta_2}
\fA(\Lambda_1)
p_{\Delta_2}
\underset{\text{\ref{LTO:Boundary}}}{=}
\fB(\Lambda_1 \Subset_s \Delta_2).
$$
Thus multiplication by $p_{\Delta_2}$ is a well-defined surjective unital $*$-algebra map
$\fB(\Lambda_1 \Subset_s \Delta_1)\to\fB(\Lambda_1 \Subset_s \Delta_2)$.
By \ref{LTO:Injective} applied to $\Lambda_1\Subset_s \Delta_1\subset \Delta_2$,
this map is also injective.
Since $\partial \Lambda_1 \cap \partial \Delta_1 = \partial \Lambda_1 \cap \partial \Delta_2$,
clearly
$$
\fB(\Lambda_1\Subset_s \Delta_2)
\underset{\text{\ref{LTO:Boundary}}}{=}
p_{\Delta_2}\fA(\Lambda_1)p_{\Delta_2}
\subseteq 
p_{\Delta_2}\fA(\Lambda_2)p_{\Delta_2}
\underset{\text{\ref{LTO:Boundary}}}{=}
\fB(\Lambda_2\Subset_s \Delta_2).
$$
When in addition $\partial \Lambda_1 \cap \partial \Delta_1 = \partial \Lambda_2\cap \partial \Delta_2$,
$\fB(\Lambda_1 \Subset_s \Delta_2)
=
\fB(\Lambda_2 \Subset_s \Delta_2)
$
by \ref{LTO:Surjective} applied to $\Lambda_1 \subset \Lambda_2\Subset \Delta_2$.
\end{proof}

The above lemma says that $\fB(\Lambda \Subset_s \Delta)$ really only depends on sites near the boundary interval $I:= \partial \Lambda \cap \partial \Delta$, together with a chosen `side of $I$' on which $\Lambda\Subset_s \Delta$ live.
Indeed, let $\Lambda_I\subset \Lambda$ to be the smallest sufficiently large rectangle with $\partial \Lambda_I \cap \partial \Delta = I$,
and let $\Delta_I\subset \Delta$ be the smallest rectangle such that $\Lambda_I \Subset_s \Delta_I$.
Setting
\begin{equation}
\label{eq:BoundaryAlgebra}
\fB(I):= \fB(\Lambda_I \Subset_s \Delta_I),
\end{equation}
Lemma \ref{lem:StabilityOfBoundaryAlgebras} says that
$\fB(I)p_\Delta = \fB(\Lambda\Subset_s \Delta)$.
We suppress the dependence on the fixed `side of $I$' in the notation.\footnote{The algebras $\fB(I)$ depend on a choice of half-space in which $\Lambda_I\Subset_s\Delta_I$ live.
It is possible one could choose $\widetilde{\Lambda}_I\Subset_s\widetilde{\Delta}_I$ with
$\partial \widetilde{\Lambda}_I \cap \partial \widetilde{\Delta}_I = I$
on the `other side' of $I$
such that $\fB(\widetilde{\Lambda}_I\Subset_s\widetilde{\Delta}_I)$
may not be isomorphic to $\fB(I)$.
}
This leads to an alternative characterisation of~\ref{LTO:Boundary}--\ref{LTO:Injective}.

\begin{prop}
The axioms \ref{LTO:Boundary},  \ref{LTO:Surjective}, and \ref{LTO:Injective} are equivalent to the following two axioms:
{\textup{
\begin{enumerate}[label=(LTO\arabic*$'$)]
\setcounter{enumi}{1}
\item 
\label{LTO':B(I)}
Whenever $\Lambda \Subset_s \Delta $ with $\partial \Lambda \cap \partial \Delta =I$,
$p_\Delta \fA(\Lambda)p_\Delta = \fB(I) p_\Delta$.
\setcounter{enumi}{3}
\item
\label{LTO':injective}
Whenever $\Lambda_I \Subset_s \Delta_I \subset \Delta$,
$xp_\Delta = 0$ implies $x=0$ for all $x\in \fB(I)$.
\end{enumerate}
}}
\end{prop}
\begin{proof}
The forward direction is immediate by Lemma \ref{lem:StabilityOfBoundaryAlgebras} as noted right above the proposition.

Now suppose \ref{LTO':B(I)} and \ref{LTO':injective} hold.
To see \ref{LTO:Boundary}, observe that
$$
p_\Delta \fA(\Lambda) p_\Delta
\underset{\text{\ref{LTO':B(I)}}}{=}
\fB(I)p_\Delta
=
p_\Lambda\fB(\Lambda_I \Subset_s \Delta_I)p_\Lambda p_\Delta
\subseteq
\fB(\Lambda\Subset_s \Delta)p_\Delta
\subseteq
p_\Delta \fA(\Lambda) p_\Delta,
$$
so equality follows.

Note that \ref{LTO:Surjective} is now immediate as whenever
$\Lambda_1\subset  \Lambda_2\Subset_s \Delta$ with 
$\partial \Lambda_1 \cap \partial\Delta= I=\partial \Lambda_2 \cap \partial \Delta$,
$$
\fB(\Lambda_1\Subset_s \Delta)
\underset{\text{\ref{LTO:Boundary}}}{=}
p_\Delta \fA(\Lambda_1)p_\Delta
\underset{\text{\ref{LTO':B(I)}}}{=}
\fB(I)p_\Delta
\underset{\text{\ref{LTO':B(I)}}}{=}
p_\Delta \fA(\Lambda_2)p_\Delta
\underset{\text{\ref{LTO:Boundary}}}{=}
\fB(\Lambda_2 \Subset_s \Delta).
$$

Finally, to prove \ref{LTO:Injective},
suppose $\Lambda \Subset_s \Delta_1\subset \Delta_2$ with 
$\partial \Lambda \cap \partial\Delta_1= I=\partial \Lambda \cap \partial \Delta_2$.
Since $p_{\Delta_2}\leq p_{\Delta_1}$, the map 
$$
\fB(\Lambda \Subset_s \Delta_1)
=
\fB(I)p_{\Delta_1}\xrightarrow{\cdot  p_{\Delta_2}} \fB(I)p_{\Delta_2}=\fB(\Lambda \Subset_s \Delta_2)
$$ 
given by multiplication by $p_{\Delta_2}$ is always surjective.
If
$x\in \fB(\Lambda \Subset_s \Delta_1)=\fB(I)p_{\Delta_1}$ with $xp_{\Delta_2}=0$,
let $y\in \fB(I)$ such that $x=yp_{\Delta_1}$.
Then 
$$
yp_{\Delta_2}=yp_{\Delta_1}p_{\Delta_2}=xp_{\Delta_2}=0,
$$ 
so $y=0$ by \ref{LTO':injective}.
We conclude that $x=yp_{\Delta_1}=0$, and thus \ref{LTO:Injective} holds.
\end{proof}

\begin{rem}
\label{rem:StrategyForExamples}
For our examples in \S\ref{sec:toric} and \S\ref{sec:LWStringNet} below, we actually prove something slightly stronger.
For each interval $I$, let $\Lambda_I$ be the minimal sufficiently large rectangle (chosen such that $I$ is always on the same side of $\Lambda_I$) as mentioned above equation~\eqref{eq:BoundaryAlgebra}.
We identify an abstract $\rm C^*$-algebra $\fC(I)\subset p_{\Lambda_I}\fA(\Lambda_I)p_{\Lambda_I}$ which commutes with every $p_\Lambda$ such that $\Lambda_I \subset \Lambda$ with $\partial \Lambda_I \cap \partial \Lambda=I$.
We then show that whenever $\Lambda \Subset_s \Delta$ with $\partial \Lambda \cap \partial \Delta=I$,
then $p_\Delta \fA(\Lambda) p_\Delta = \fC(I) p_\Delta$, 
and that
$xp_\Delta = 0$ implies $x=0$ for $x\in \fC(I)$.
Thus $\fC(I)\cong \fC(I)p_{\Delta_I}=\fB(I)$, but it lives inside $p_{\Lambda_I}\fA(\Lambda_I)p_{\Lambda_I}$ rather than $p_{\Delta_I}\fA(\Lambda_I)p_{\Delta_I}$.
\end{rem}

Before getting to our examples, we now analyze a canonical state on $\fA$ from \ref{LTO:Hastings} and a canonical quantum channel from $\fA$ to a `quasi-local boundary algebra' $\fB=\varinjlim \fB(I)$ coming from \ref{LTO:Boundary}, \ref{LTO:Surjective}, and \ref{LTO:Injective}.
If the reader is more interested in the examples, they may skip directly to \S\ref{sec:toric} and \S\ref{sec:LWStringNet} below.

\subsection{Canonical state of a locally topologically ordered net of projections}
\label{sec:CanonicalState}

We now show that a net of projections $(p_\Lambda)$ on a net of algebras $\fA$ satisfying \ref{LTO:Hastings} has a canonical state.
Recall that if ground states are locally indistinguishable, all of them will converge to the same ground state when taking the thermodynamic limit.
To see this, consider an increasing sequence $\Lambda_1 \subset \Lambda_2 \subset \dots$ of rectangles exhausting $\cL$.
Suppose that $\phi_n : \fA \to \mathbb{C}$ is a sequence of states such that $\phi_n|_{\fA(\Lambda)}$ is a ground state for the local dynamics $H_n$ on $\fA(\Lambda_n)$.
Then the weak-* limit $\phi$ is a ground state of the thermodynamic limit.
Suppose that $\psi_n$ is another such sequence.
Then for $x \in \fA(\Lambda)$, by local indistinguishability we have that $\psi_n(x) = \phi_n(x)$ for $n$ large enough.
Hence the weak-* limits coincide.
Coming back to our setting of nets of projections satisfying our LTO axioms,
if the $(p_\Lambda)$ are the local ground state projections of a quantum spin model with local topological quantum order, this state is precisely the canonical state that we define here (see Remark~\ref{rem:PsiIsUniqueTranslationInvariantGroundState} below).

Below, we write $\ll$ instead of $\ll_s$ to ease the notation.

\begin{lem}
\label{lem:StateIndependentOfDelta}
Suppose $(\fA,p)$ satisfies \ref{LTO:Hastings}.
For $\Lambda$ sufficiently large with $\Lambda \ll \Delta$ and $x\in \fA(\Lambda)$, define $\psi_\Delta(x)\in \bbC$ as the scalar such that $p_\Delta xp_\Delta=\psi_\Delta(x)p_\Delta$.
Then $\psi_\Delta(x)$ is independent of $\Delta$.
We may thus denote $\psi_\Delta(x)$ simply by $\psi(x)$.
\end{lem}
\begin{proof}
If $\Lambda \ll \Delta_i$ for $i=1,2$, then pick $\Delta_3$ containing $\Delta_1\cup \Delta_2$.
Since $p_{\Delta_3}\leq p_{\Delta_i}$ for $i=1,2$,
\begin{equation*}
p_{\Delta_3} x p_{\Delta_3}
=
p_{\Delta_3}p_{\Delta_i} x p_{\Delta_i}p_{\Delta_3}
=
p_{\Delta_3}\psi_{\Delta_i}(x)p_{\Delta_i}p_{\Delta_3}
=
\psi_{\Delta_i}(x)p_{\Delta_3}
\qquad\qquad
\text{for }
i=1,2,
\end{equation*}
and so $\psi_{\Delta_1}(x)=
\psi_{\Delta_3}(x)=\psi_{\Delta_2}(x)$.
\end{proof}

The following lemma and corollary are certainly known to experts.
We include a proof for convenience and completeness.

\begin{lem}
\label{lem:PositiveAFLimit}
Suppose $A=\varinjlim A_n$ is a unital AF $\rm C^*$-algebra where each $A_n$ is a finite dimensional $\rm C^*$-algebra.
Then $A^+ = \varinjlim A_n^+$.
\end{lem}
\begin{proof}
Suppose $a_n\in A_n$ with $a_n \to a\in A^+$.
Then $a_n^*\to a$, so $b_n:=\frac{a_n+a_n^*}{2}\to a$.
This means for every $\varepsilon>0$, there is an $N>0$ such that $n>N$ implies 
$$
\spec(b_n)\subset N_{\varepsilon/2}(\spec(a))\subset [-\varepsilon/2, \|a\|+\varepsilon/2],
$$
where $N_{\varepsilon/2}(\spec(a))$ is an $\varepsilon/2$ neighborhood of $\spec(a)$.
Since each $A_n$ is closed under functional calculus, we see that applying the function $(\,\cdot\,)_+:r\mapsto \max\{0,r\}$ to $b_n$ gives a positive sequence $(b_n)_+$ such that $\|(b_n)_+-b_n\| \leq \varepsilon/2$ for all $n>N$.
Picking $N'>N$ such that $\|b_n-a\|<\varepsilon/2$,
for all $n>N'$,
$$
\|(b_n)_+ - a\|
\leq
\|(b_n)_+ - b_n\|+\|b_n-a\|
<\varepsilon.
$$
Hence $(b_n)_+\to a$, and the result follows.
\end{proof}

\begin{cor}
\label{cor:ExtendAF-UCP}
Suppose $A=\varinjlim A_n$ is a unital AF $\rm C^*$-algebra and $B$ is another unital $\rm C^*$-algebra.
A unital (completely) positive map $\phi: \bigcup A_n \to B$
uniquely extends to a unital (completely) positive map $A\to B$.
\end{cor}
\begin{proof}
Since $\phi$ is unital, for all $a\in \bigcup_n A_n$,
$$
\phi(a^*a) \leq \|a^*a\| \phi(1) = \|a^*a\|.
$$
Since positives in $\bigcup_n A_n$ span $\bigcup_n A_n$, we see $\phi$ is bounded.
Hence $\phi$ uniquely extends to a map $A\to B$, and (complete) positivity follows directly from Lemma \ref{lem:PositiveAFLimit}.
\end{proof}

\begin{defn}
\label{defn:CanonicalState}
By Lemma \ref{lem:StateIndependentOfDelta}, $\psi$ is a well-defined positive linear functional on $\bigcup_\Lambda \fA(\Lambda)$ such that $\psi(1_\fA)=1_\bbC$, and thus extends to a unique state on $\fA$ by Corollary \ref{cor:ExtendAF-UCP}.
We call this the \emph{canonical state} associated to the net $(p_\Lambda)$.
\end{defn}

\begin{cor}
\label{cor:TranslationInvariant}
When $\fA$ and $(p_\Lambda)$ are translation invariant, so is the state $\psi$.
\end{cor}
\begin{proof}
For every rectangle $\Lambda$ and $g\in \Aut_{\tr}(\cL)$ there is a rectangle $\Delta$ with both $\Lambda \ll \Delta$ 
and 
$g+\Lambda \ll \Delta \Leftrightarrow \Lambda \ll -g+\Delta$.
Then using the notation from Lemma \ref{lem:StateIndependentOfDelta},
\[
\psi(g\cdot x)
=
\psi_\Delta(g\cdot x)
=
\psi_{-g+\Delta}(x) 
=
\psi_\Delta(x) 
=
\psi(x)
\qquad\qquad
\forall\, x\in \fA(\Lambda).
\qedhere
\]
\end{proof}

The next lemma follows by a simple application of the Cauchy-Schwarz inequality.

\begin{lem}[{\cite[\textsection 2.1.1]{MR2345476}}]
\label{lem:AFHCauchySchwarz}
Let $A$ be a unital $\rmC^*$-algebra and $\phi \colon A \to \bbC$ be a state.  
Suppose $x \in A$ satisfies $x \leq 1_A$ and $\phi(x) = 1$.  
Then for all $y \in A$, 
\[
\phi(xy) = \phi(yx) = \phi(y).
\]
\end{lem} 

\begin{cor}
For every rectangle $\Lambda$,
$\psi(x) = \psi(p_\Lambda x) = \psi(xp_\Lambda)$ for all $x\in \fA$.
\end{cor}
\begin{proof}
Whenever $\Lambda \ll \Delta$, 
$p_\Delta p_\Lambda p_\Delta = p_\Delta$.
Thus $\psi(p_\Lambda)=1$ for every rectangle $\Lambda$.
Now apply Lemma~\ref{lem:AFHCauchySchwarz}.
\end{proof}

\begin{cor}
\label{cor:UniqueStateTakingOneOnNet}
If $\phi$ is a state on $\fA$ satisfying $\phi(p_\Lambda) = 1$ for all rectangles $\Lambda$, then $\phi=\psi$.
In particular, $\psi$ is pure.
\end{cor}
\begin{proof}
It suffices to prove that $\phi=\psi$ on every $\fA(\Lambda)$ for $\Lambda$ sufficiently large.
Pick any rectangle $\Delta$ with $\Lambda \ll \Delta$.
Then for all $x\in \fA(\Lambda)$,
\[
\phi(x)
\underset{\text{\scriptsize (Lem.~\ref{lem:AFHCauchySchwarz})}}{=}
\phi(p_\Delta x p_\Delta) 
\underset{\text{\scriptsize (Lem.~\ref{lem:StateIndependentOfDelta})}}{=}
\psi(x)
\cdot 
\phi(p_\Delta) 
\underset{\text{\scriptsize (Lem.~\ref{lem:AFHCauchySchwarz})}}{=}
\psi(x).
\]
Purity of $\psi$ now follows quickly.  
Indeed, suppose $\varphi \colon \fA \to \bbC$ is a functional satisfying $0 \leq \varphi \leq \psi$.  
Then for all rectangles $\Lambda$, 
\[
0
\leq
\varphi(1_\fA - p_\Lambda)
\leq
\psi(1_\fA - p_\Lambda)
=
0.
\]
Hence $\varphi(p_\Lambda) = \varphi(1)$ for all rectangles $\Lambda$,
so $\varphi = \varphi(1_\fA)\cdot\psi$.
\end{proof}

While the quasi-local algebra $\fA$ carries a canonical pure state $\psi$, we do not, \emph{a priori}, have a local Hamiltonian for which $\psi$ is the ground state.

\begin{rem}
\label{rem:PsiIsUniqueTranslationInvariantGroundState}
Suppose that $\fA$ arises from a translation-invariant frustration free local Hamiltonian $H$ on a spin system.
Recall that an interaction is an assigment $X \mapsto \Phi(X) = \Phi(X)^* \in \fA(X)$ to each finite subset $X \subset \cL$.
The interaction $X \mapsto \Phi(X)$ is called \emph{frustration free} if $\Phi(X) \geq 0$ and the ground states of the local Hamiltonians are given by $\ker(H_\Lambda)$, with $H_\Lambda = \sum_{X \subset \Lambda} \Phi(X)$.
This implies that if $p_\Lambda$ is the projection onto the local ground space of $H_\Lambda$, we have $p_\Delta = p_\Delta p_\Lambda = p_\Lambda p_\Delta$ if $\Lambda \subset \Delta$, and so $(p_\Lambda)$ is a net of projections.
Suppose that $(p_\Lambda)$ satisfies
\ref{LTO:Hastings}. 
We claim that the canonical state $\psi$ is the unique translation-invariant ground state for $H$.
Translation invariance follows immediately from Corollary \ref{cor:TranslationInvariant}.

It is easy to check that $\psi(H_\Lambda)=0$ for every sufficiently large $\Lambda$, which is the minimum possible value a state can take on $H_\Lambda$.
Indeed, using the spectral theorem locally, we can write $H_\Lambda = \lambda_0p_\Lambda + \sum \lambda_i q_i$ where each $\lambda_i > \lambda_0= 0$, and the $p_\Lambda,q_i$ are commuting non-zero projections that sum to 1.
Since 
$$
\psi(q_i) 
= 
\psi(p_\Lambda q_i) 
= 
\psi(0)=0
\qquad
\forall\, i,
\qquad
\Longrightarrow
\qquad
\psi(H_\Lambda)
=
\lambda_0 \psi(p_\Lambda) + \sum_i \lambda_i \psi(q_i)
=
0.
$$
This immediately implies $\psi(H_\Lambda)=0$ for every rectangle $\Lambda$ by considering $\Lambda \subset \Delta$ with $\Delta$ suitably large.
Moreover, if $\phi$ is a state on $\fA$ such that $\phi(H_\Lambda)=0$ for every $\Lambda$, then necessarily $\phi(q_i)=0$ for all $i$, so
$$
\phi(p_\Lambda)
=
\phi(p_\Lambda)
+
\sum_i \phi(q_i)
=
\phi(1_\fA) 
=
1
\qquad\qquad\qquad
\forall\, \Lambda
$$
and thus $\phi=\psi$ by Corollary \ref{cor:UniqueStateTakingOneOnNet}.
It now follows by \cite[Thm.~6.2.58]{MR1441540} that $\psi$ is the unique translation-invariant ground state for $H$.

\end{rem}

The article \cite{MR2742836} considers a finite 2D quantum spin system defined on $\Gamma := \mathbb{Z}_L \times \mathbb{Z}_L$ for suffieciently large integers $L$, with periodic boundary conditions.
Furthermore, they assume in \cite[\S~II.A]{MR2742836} that they have a frustration-free commuting projector local Hamiltonian.
In this setting, for each rectangle $\Lambda \subset \Gamma$, we define a projection $p_\Lambda$ onto the ground space of the local Hamiltonian supported on $\Lambda$.
This gives a net $\Lambda \mapsto p_\Lambda$ of projections as above.
The authors then define a pair of \emph{topological quantum order conditions}.
We paraphrase them here for convenience using our notation:
\begin{enumerate}[label=(TQO\arabic*)]
    \item For $a \in \fA(\Lambda)$ with $\Lambda$ small enough compared to the system size, we have $p_\Gamma a p_\Gamma \in \bbC p_\Gamma$.
    \item Let $\Lambda$ be square that is small enough compared to the system size. Then $\ker \Tr_{\Lambda^c}(p_\Gamma) = \ker \Tr_{\Lambda^c}(p_{\Lambda^{+1}})$.
\end{enumerate}
Here $\Tr_{\Lambda^c} : \fA \to \fA(\Lambda)$ is the $\tr$-preserving conditional expectation which traces out the degrees of freedom localized in $\Lambda^c := \Gamma \setminus \Lambda$.
The condition (TQO2) can be interpreted as the ``local'' ground spaces being compatible with the ``global'' one.
Our~\ref{LTO:Hastings} condition implies both of these conditions.
\begin{prop}
\label{prop:LTQO}
The condition~\ref{LTO:Hastings} in the setup of \cite{MR2742836}
implies both topological quantum order conditions (TQO1) and (TQO2) of~\cite{MR2742836}.
\end{prop}
\begin{proof}
Since we now consider a finite quantum system, $p_\Gamma$ is a projection in $\fA$.
If a rectangle $\Lambda$ is small compared to the system size $L$, we have for all $x \in \fA(\Lambda)$
\[
    p_\Gamma x p_\Gamma = \psi(x) p_\Gamma
\]
as before.
This is precisely (TQO1).

Now let $\tr$ be the unique tracial state on $\fA$ and suppose that $\Lambda$ again is a rectangle which is small compared to $L$.
Let $\Delta$ be such that $\Lambda \ll_2 \Delta$.\footnote{The $\ll_2$ matches the assumptions on the locality of the local terms in the Hamiltonian in~\cite[\S~II.B]{MR2742836}.}
Consider the state on $\fA$ given by $\phi(x) = (\tr(p_\Delta))^{-1} \cdot \tr(p_\Delta x)$.
Then by Lemma \ref{lem:StateIndependentOfDelta},
\[
\phi(x)
=
\frac{\tr(p_\Delta x p_\Delta)}{\tr(p_\Delta)} 
=
\frac{\psi(x) \tr(p_\Delta)}{\tr(p_\Delta)}
=
\psi(x)
\qquad\qquad\qquad
\forall\,
x \in \fA(\Lambda).
\]
Since 
$\phi|_{\fA(\Lambda)}(x) = \tr(p_\Delta)^{-1}\tr(\Tr_{\Lambda^c}(p_\Delta)x)$
and
$\psi|_{\fA(\Lambda)}(x) = \tr(p_\Gamma)^{-1}\tr(\Tr_{\Lambda^c}(p_\Gamma)x)$,
and both these are equal for all $x\in \fA(\Lambda)$,
we see that 
$\tr(p_\Gamma)^{-1}\Tr_{\Lambda^c}(p_\Gamma) 
=
\tr(p_\Delta)^{-1}\Tr_{\Lambda^c}(p_\Delta)$.
In particular, $\Tr_{\Lambda^c}(p_\Gamma)$ and $\Tr_{\Lambda^c}(p_\Delta)$ have the same kernel, which is (TQO2) of~\cite{MR2742836}.\footnote{Note that we have shown something strictly stronger than (TQO2); in fact, not only do these operators share the same kernel, but they are actually proportional.}
\end{proof}

\subsection{The boundary net}
\label{sec:BoundaryNet}

In this section, we use \ref{LTO:Boundary}, \ref{LTO:Surjective}, and \ref{LTO:Injective} to build a canonical boundary net of algebras $\fB$ on a codimension one $\bbZ^{\ell-1}$ sublattice $\cK$ of $\cL$.
For a choice of half-plane $\bbH$ bounded by $\cK$, setting 
$$
\fA_\bbH
:=
\varinjlim_{\substack{
\Lambda\subset \bbH
\\
\partial \Lambda \cap \cK \neq \emptyset
}} 
\fA(\Lambda),
$$
we construct a unital completely positive (ucp) map\footnote{Such a map is also called a \emph{quantum channel} in the Heisenberg picture.} $\bbE:\fA_\bbH \to \fB$
satisfying
\begin{equation}
\label{eq:DefiningConditionOfE}
p_\Delta x p_\Delta = \bbE(x)p_\Delta
\qquad\qquad\qquad\qquad
\forall\, x\in \fA(\Lambda),
\quad
\forall\, \Lambda\Subset\Delta \text{ with }\partial \Delta \cap \cK\neq \emptyset.
\end{equation}
Thus one way to think of $\fB$ 
is as a generalization of the role played by $\bbC$ in \ref{LTO:Hastings} as a receptacle for the canonical state $\psi$, whose role here is played by $\bbE$.

\begin{construction}[\textbf{The boundary net}.]
\label{const:LTO-BoundaryNet}
Consider a net of projections $(\fA, p)$ satisfying the axioms \ref{LTO:Hastings}--\ref{LTO:Injective}.
For an interval $I\subset \cK$, 
let $\Lambda_I\subset \bbH$ be the smallest sufficiently large rectangle with $\partial \Lambda_I\cap \cK = I$.
Let $\Delta_I$ be the smallest rectangle with $\Lambda_I \Subset_s \Delta_I$ and $\partial \Lambda_I \cap \partial \Delta_I =I$.
As in \eqref{eq:BoundaryAlgebra}, we define
$
\fB(I):=\fB(\Lambda_I \Subset_s \Delta_I).
$

We now show that $I\mapsto \fB(I)$ defines a net of algebras.
By convention, $\Lambda_\emptyset:=\emptyset$, so $\fB(\emptyset)=\bbC$, and \ref{N:empty} holds.

If $I\subset J\subset \cK$, the map $\fB(I) \to \fB(J)$ given by $x\mapsto xp_{\Delta_J}$ is a well-defined injective $*$-algebra homomorphism by Lemma \ref{lem:StabilityOfBoundaryAlgebras}.
Since $p_{\Delta_K}\leq p_{\Delta_J}$ whenever $J\subset K\subset \cK$,
the algebras $\fB(I)$ form an inductive limit system.
Setting $\fB := \varinjlim \fB(I)$, we see \ref{N:inclusion} and \ref{N:dense} hold.

Finally, suppose we have disjoint intervals $I\cap J = \emptyset$ in $\cK$, so that also $\Lambda_I\cap \Lambda_J = \emptyset$.
For $xp_{\Delta_I}\in \fB(I)$ 
with $x\in p_{\Lambda_I}\fA(\Lambda_I)p_{\Lambda_I}$
and $yp_{\Delta_J}\in \fB(J)$
with $y\in p_{\Lambda_J}\fA(\Lambda_J)p_{\Lambda_J}$,
$[x,y]=0$.
Thus for any interval $K\subset \cK$ containing $I\cup J$,
$
xy p_{\Delta_K}
=
yx p_{\Delta_K}
$,
so
\ref{N:local} holds.
\end{construction}

\begin{lem}
\label{lem:EIndependentOfDelta}
Suppose $(\fA,p)$ satisfies \ref{LTO:Boundary}, \ref{LTO:Surjective}, and \ref{LTO:Injective}.
For $\Lambda \Subset_s \Delta$ with $\partial\Lambda \cap \partial \Delta = I$,
and $x\in \fA(\Lambda)$,
the operator $\bbE_\Delta(x)\in \fB(I)$ satisfying $p_\Delta x p_\Delta = \bbE_\Delta(x)p_\Delta$ is independent of the choice of $\Delta$.
We may thus denote $\bbE_\Delta(x)$ simply by $\bbE(x)$.
\end{lem}
\begin{proof}
Suppose $\Lambda\Subset_s \Delta_i$ with $I=\partial \Lambda \cap \partial \Delta_i \neq \emptyset$ for $i=1,2$,
and let $\bbE_i(x)\in \fB(I)$ such that $p_{\Delta_i}xp_{\Delta_i}=\bbE_i(x)p_{\Delta_i}$.
Picking $\Delta_3$ containing $\Delta_1\cup\Delta_2$ such that $I=\partial \Lambda\cap \partial \Delta_3$, since $p_{\Delta_3}\leq p_{\Delta_i}$ for both $i=1,2$, we have
\begin{equation*}
p_{\Delta_3} x p_{\Delta_3}
=
p_{\Delta_3}p_{\Delta_i} x p_{\Delta_i}p_{\Delta_3}
=
p_{\Delta_3}\bbE_i(x)p_{\Delta_i}p_{\Delta_3}
=
\bbE_i(x)p_{\Delta_3}
\qquad\qquad
\text{for }
i=1,2.
\end{equation*}
Hence 
$(\bbE_1(x)-\bbE_2(x))p_{\Delta_3}=0$, and so 
$\bbE_1(x)=\bbE_2(x)$.
Hence $\bbE_\Delta$ is independent of $\Delta$ as claimed.
\end{proof}

\begin{defn}
Identifying each boundary algebra $\fB(I)$ with its image in $\fB = \varinjlim \fB(I)$,
by Lemma \ref{lem:EIndependentOfDelta},
we get a well-defined map
$$
\bbE:\bigcup_{\substack{
\Lambda\subset \bbH
\\
\partial \Lambda \cap \cK \neq \emptyset
}}
\fA(\Lambda)\to \fB.
$$ 
satisfying the formula \eqref{eq:DefiningConditionOfE} above.
Observe that $\bbE$ is manifestly ucp, as it is defined by compressing by a projection.
Thus $\bbE$ uniquely extends to a ucp map $\bbE:\fA_\bbH\to \fB$ by Corollary \ref{cor:ExtendAF-UCP}.
\end{defn}

\begin{rem}
Although $\fB$ is not a unital subalgebra of $\fA$, so $\bbE$ is not technically a conditional expectation, we do have the property that 
$\bbE(x)=x$ for all $x\in \fB(I)$, which we show below.
This means $\bbE$ is like a conditional expectation, but onto a subalgebra with a different unit than the ambient algebra.
Indeed, since $[x,p_\Delta]=0$ whenever $\Lambda \Subset_s \Delta$ with $\partial \Lambda \cap \partial \Delta = I$,
we have 
$$
\bbE(x)p_\Delta
=
p_\Delta xp_\Delta 
= 
xp_\Delta.
$$
Thus $(\bbE(x)-x)p_\Delta = 0$, which implies $\bbE(x)=x$.
\end{rem}

\begin{rem}
Observe that if $\Lambda \subset \bbH$ is far enough from $\cK$, i.e., there is a $\Delta \subset \bbH$ with $\Lambda \ll_s \Delta$, then $\bbE|_{\fA(\Lambda)} = \psi$.
Indeed, pick $\Lambda'\Subset_s \Delta'$ with $\partial \Lambda' \cap \partial \Delta'\subset \cK$ and $\Lambda \subset \Lambda'$ and $\Delta\subset \Delta'$.
Then for $x\in \fA(\Lambda)$,
$$
\bbE(x)p_{\Delta'}
=
p_{\Delta'}xp_{\Delta'}
=
p_{\Delta'}p_\Delta xp_\Delta p_{\Delta'}
=
\psi(x)p_{\Delta'}.
$$
By \ref{LTO:Injective}, $\bbE(x)=\psi(x)$ as claimed.

This means that \emph{every} state on the boundary algebra $\fB$ canonically extends to a state on $\fA_\bbH$ which looks like the canonical state $\psi$ in the bulk.
Indeed, for an arbitrary state $\phi_\fB$ on $\fB$, we define $\phi: \fA_\bbH \to \bbC$ by $\phi:= \phi_\fB \circ \bbE$.
Thus the boundary algebra $\fB$ gives us a \emph{state-based approach} to boundary conditions.
We will study such states in more detail in \S\ref{sec:BoundaryStates}.
\end{rem}

\begin{ex}
\label{ex:BoundaryState}
Observe that $\psi$ gives a canonical translation-invariant state $\psi_\fB$  on the boundary net of algebras $\fB$
by $\psi_\fB(xp_{\Delta_I}):= \psi(xp_{\Delta_I}) = \psi(x)$ by Lemma \ref{lem:AFHCauchySchwarz}.
\end{ex}

At this time, we do  not know if \ref{LTO:Boundary}--\ref{LTO:Injective} imply $\psi_\fB$ is faithful on $\fB$.
In our examples, $\psi_\fB$ is a KMS state and $\fB$ is simple which implies $\psi_\fB$ is faithful by \cite[Cor.~5.3.9]{MR1441540}.
See \S\ref{sec:ToricCodeBoundaryStates}--\S\ref{sec:KMS} for more details.

\section{Example: Kitaev's Toric Code}
\label{sec:toric}
As a first example, we consider Kitaev's Toric Code~\cite{MR1951039}.
To follow conventions used in most of the literature on the Toric Code, here we will not exactly follow the definitions of~\S\ref{sec:LTO}; the spins/degrees of freedom now live on the edges of a $\bbZ^2$ lattice instead of the vertices, and we will still use rectangles along these edges for regions.
In Remark \ref{rem:DiagonalRegions} below, we make the connection to the exact setup of~\S\ref{sec:LTO}.

The Toric Code is defined on a square $\bbZ^2$ lattice with a copy of $\bbC^2$ placed at each edge of the lattice.
If $\Lambda$ is a finite subset of edges, we have $\fA(\Lambda) \cong \bigotimes_{\ell \in \Lambda} M_2(\mathbb{C})$.
A \emph{star} $s$ is the set of four edges incident to some vertex $v$, and a \emph{plaquette} $p$ consists of the four edges around a face/plaquette.
We define star operators $A_s$ and plaquette operators $B_p$ as
\[
    A_s := \bigotimes_{j \in s} \sigma_j^X, \quad\quad B_p := \bigotimes_{j \in p} \sigma_j^Z,
\]
where $\sigma_j^X$ and $\sigma_j^Z$ denote the Pauli matrices $\sigma^X$ and $\sigma^Z$ acting on the site $j$.
The local Hamiltonians are then defined as
\begin{equation}
    \label{eq:tcHamiltonian}
    H_\Lambda := \sum_{s \subset \Lambda} (1-A_s) + \sum_{p \subset \Lambda} (1-B_p).
\end{equation}
This model has been studied extensively in the literature.
Here we just mention that in the thermodynamic limit, the model has a unique frustration-free ground state (which is also translation invariant), as well as non-frustration-free (and non-invariant) ground states associated to the superselection sectors of the model~\cite{MR2345476,MR3764565}.

The Toric Code satisfies the local TQO conditions~\cite{Cui2020kitaevsquantum}, and we will exploit this to define a net of projections satisfying \ref{LTO:Hastings}--\ref{LTO:Injective}.
If $\Lambda \subset \cL$ is a rectangle, define the projections
\begin{equation}
    \label{eq:tcproj}
    p_\Lambda := \prod_{s \subset \Lambda} \left(\frac{1+A_s}{2}\right) \prod_{p \subset \Lambda} \left(\frac{1+B_p}{2}\right),
\end{equation}
which locally project onto the ground state of our Hamiltonian \eqref{eq:tcHamiltonian}.

We briefly discuss the intuition behind the algebras of the form $p_\Delta \fA(\Lambda) p_\Delta$ with $\Lambda \Subset_2 \Delta$, which goes back to methods employed in~\cite{MR2345476}.
Recall that (pairs of) excitations are created by path operators.
To a path $\xi$ on the lattice, we can associate an operator $F_\xi$ that acts with a $\sigma^Z$ on each of the edges in the path.
Similarly, to a path $\xi$ on the dual lattice, we can associate an operator $\widehat{F}_\xi$ acting with $\sigma^X$ on all edges that the dual path crosses.
It is easy to check that if the path is not closed, $F_\xi$ anti-commutes with $A_s$ at the start and endpoints of the path, and commutes with all other $A_s$.
The same is true for $\widehat{F}_\xi$ and the plaquette operators.
Moreover, path operators associated to closed paths (closed paths on the dual lattice) are a product of plaquette (star) operators.

The local algebras are linear spans of products of path operators on the lattice and dual lattice.
If $x\in \fA(\Lambda)$ is such a product, it either commutes or anti-commutes with any given star or plaquette operator.
Suppose it anti-commutes with $A_s$ for some $s \subset \Delta$ where $\Lambda \Subset_2 \Delta$.
Then we have
\[
    (1+A_s) x (1+A_s) = (1+A_s)(1-A_s) x = 0,
\]
and hence $p_\Delta x p_\Delta = 0$.
Note that this has a clear physical interpretation: $x$ creates an excitation at the star $s$, so it takes us out of the ground state space.
If there is no star or plaquette operator in $\fA(\Delta)$ that anti-commutes with $x$ and $\partial\Lambda \cap \partial \Delta = \emptyset$,
it follows that $x$ is a product of star and plaquette operators supported on $\Lambda$~\cite{MR2345476} (see also Algorithm \ref{alg:PauliReduction} below).
Thus $p_\Delta x p_\Delta = p_\Delta$, i.e., the unit of the compressed algebra $p_\Delta \fA(\Lambda)p_\Delta$.
Observe that this argument above did not depend on $\Delta$ beyond relying on the condition that $\Lambda \Subset_2 \Delta$ with $\partial \Lambda \cap \partial \Delta = \emptyset$.

However, this argument breaks down if $\partial\Lambda \cap \partial\Delta=:I\neq \emptyset$,
as in this case, the plaquette operators near the boundary $I$ no longer appear in $p_\Delta$.
In this case, $p_\Delta \fA(\Lambda) p_\Delta$ can be identified with the operators that create excitations at sites of $I$, but leave the bulk untouched.
We provide a proof of this in Algorithm \ref{alg:PauliReduction} below.

There are two cases for the 1D boundary $\cK$: rough (shown on the left below) and smooth (in the right picture) depending on the choice of $\bbZ$ hyperplane in $\bbZ^2$.
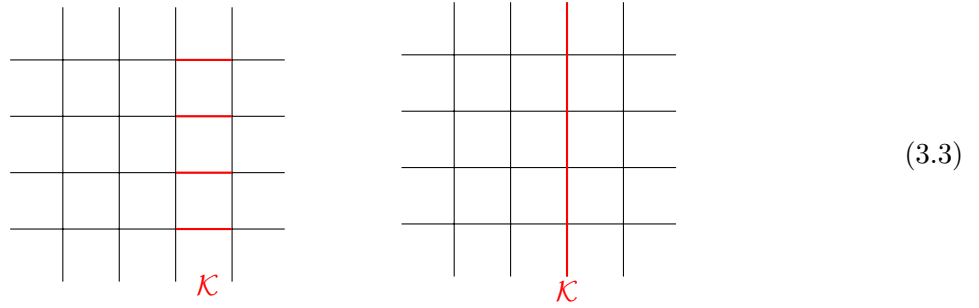
\begin{equation}
\label{eq:TC-Boundaries}
\tikzmath{
\draw[step=.75] (.05,.05) grid (3.7,3.7);
\foreach \y in {.75,1.5,2.25,3}{
    \draw[thick, red] (2.25,\y) -- (3,\y);
}
\node[red] at (2.675,0) {$\cK$};
}
\qquad\qquad
\tikzmath{
\draw[step=.75] (.05,.05) grid (3.7,3.7);
\draw[thick, red] (2.25,.05) -- (2.25,3.7);
\node[red] at (2.25,-.15) {$\cK$};
}
\end{equation}
While these two boundaries appear to break the translation symmetry assumed for our net of algebras $\fA$, there is an additional $\frac{1}{2}$-translation dualizing symmetry which shifts all edges $45^\circ$ to the northeast and swaps horizontal and vertical edges.
This $\frac{1}{2}$-translation dualizing symmetry maps between these two 1D boundaries.
$$
\tikzmath{
\draw[step=.75] (.8,.8) grid (3.7,3.7);
\draw[thick, red] (2.25,.8) -- (2.25,3.7);
}
=
\tikzmath{
\draw[step = .75] (.8,.8) grid (3.7,3.7);
\foreach \x in {1.125,1.875,2.625,3.375}{
\foreach \y in {1.5,2.25,3}{
\filldraw (\x,\y) circle (.05cm);
\filldraw (\y,\x) circle (.05cm);
}}
\draw[thick, red] (2.25,.8) -- (2.25,3.7);
\foreach \y in {1.125,1.875,2.625,3.375}{
\filldraw[red] (2.25,\y) circle (.05cm);
\draw[thick, blue, ->] (2.25,\y) -- ($ (2.25,\y) + (.375,.375) $);
}
}
\begin{tikzcd}
\mbox{}
\arrow[r,squiggly]
&
\mbox{}
\end{tikzcd}
\tikzmath{
\draw[step = .75] (.8,.8) grid (3.7,3.7);
\foreach \x in {1.125,1.875,2.625,3.375}{
\foreach \y in {1.5,2.25,3}{
\filldraw (\x,\y) circle (.05cm);
\filldraw (\y,\x) circle (.05cm);
}}
\foreach \y in {1.5,2.25,3}{
\draw[thick, red] (2.25,\y) -- (3,\y);
\filldraw[red] (2.625,\y) circle (.05cm);
}
}
=
\tikzmath{
\draw[step=.75] (.8,.8) grid (3.7,3.7);
\foreach \y in {1.5,2.25,3}{
    \draw[thick, red] (2.25,\y) -- (3,\y);
}
}
$$

In either case, we can fully describe the string operators which create excitations at the boundary.
These boundary operators are supported not only on $\cK$, but at sites in the bulk closest to $\cK$ as well.
For a rough boundary interval $I$ of $\Lambda$, we write $\widetilde{I}$ for $I$ union the next row or column of $\Lambda$ adjacent to $I$.
Similarly, for a smooth boundary interval $J$ of $\Lambda$, we write $\widetilde{J}$ for $J$ union the next row or column of $\Lambda$ adjacent to $J$.
In the diagrams below depicting the regions $\widetilde{I},\widetilde{J}$, we assume that the rectangle $\Lambda$ meeting $\cK$ at the boundary has interior on the left of $\cK$.
We also define corresponding $\rm{C}^*$-algebras $\fC(I)$ and $\fD(J)$ as follows:
\begin{equation}
\label{eq:TC-BoundaryExcitations}
\tikzmath{
\draw (1.5,.05) -- (1.5,3.7);
\foreach \y in {.75,1.5,2.25,3}{
    \draw (2.25,\y) -- (1.5,\y);
}
\draw[thick, red] (2.25,1.5) -- (1.5,1.5) -- (1.5,2.25) -- (2.25,2.25);
\node[red] at (1.3,1.875) {$\scriptstyle \sigma^Z$};
\node[red] at (1.875,1.3) {$\scriptstyle \sigma^Z$};
\node[red] at (1.875,2.45) {$\scriptstyle \sigma^Z$};
\node[red] at (2.5,1.875) {$D_p$};
\node[red] at (1.875,1.875) {$\scriptstyle p$};
\draw[thick, orange] (1.5,3) -- (2.25,3);
\node[orange] at (1.875,3.2) {$\scriptstyle \sigma^X$};
\node[orange] at (2.5,3) {$C_\ell$};
\node at (2,-.3) {$\fC(I):=
{\rm C^*}\set{C_\ell,D_p}{\ell\subset I,\,p \subset \widetilde{I}}$};
\draw[thick, blue, rounded corners=5pt] (1,3.7) -- (2.25,3.7) -- (2.25,.05) -- (1,.05);
\node[blue] at (1.25,3.5) {$\Lambda$}; 
}
\qquad\qquad
\tikzmath{
\draw (2.25,.05) -- (2.25,3.7);
\foreach \y in {.75,1.5,2.25,3}{
    \draw (2.25,\y) -- (1.5,\y);
}
\draw[thick, red] (2.25,.75) -- (2.25,1.5);
\node[red] at (2.05,1.175) {$\scriptstyle \sigma^Z$};
\node[red] at (2.6,1.175) {$D_\ell$};
\draw[thick, orange] (1.5,3) -- (2.25,3);
\draw[thick, orange] (2.25,2.25) -- (2.25,3.7);
\node[orange] at (1.5,3.1) {$\scriptstyle \sigma^X$};
\node[orange] at (2,2.675) {$\scriptstyle \sigma^X$};
\node[orange] at (2,3.375) {$\scriptstyle \sigma^X$};
\node[orange] at (2.6,3) {$C_s$};
\node at (2,-.3) {$\fD(J):=
{\rm C^*}\set{C_s,D_\ell}{\ell\subset J,\,s\subset \widetilde{J}}$};
\draw[thick, blue, rounded corners=5pt] (1.25,3.7) -- (2.35,3.7) -- (2.35,.05) -- (1.25,.05);
\node[blue] at (1.5,3.5) {$\Lambda$}; 
}.
\end{equation}

Here, the operators $C_\ell,D_p\in \fC(I)$ are the portions of the corresponding star and plaquette terms which are included in $\Lambda$, i.e., $C_\ell=\sigma^X_\ell$ for the edge $\ell$, and $D_p=\bigotimes_{j\in p} \sigma^Z_j$ where $p$ is a truncated plaquette.
The operators $C_s, D_\ell\in \fD(J)$ are defined similarly.

While they have different abstract descriptions, $\fC(I)$ and $\fD(J)$ are isomorphic when the intervals $I$ and $J$ contain the same number of sites (even if $\widetilde{I}$ and $\widetilde{J}$ contain a different number of sites).
We omit the proof of the following lemma, which is straightforward.

\begin{lem}
\label{lem:AbstractBoundaryAlgebra}
When a rough boundary interval $I$ has $n+1$ horizontal sites, $\fC(I)$ has the following abstract presentation as a $*$-algebra:
\begin{itemize}
    \item generators: $x_1,\dots, x_{n+1}, y_1,\dots, y_{n}$ 
    \item relations:
    \begin{enumerate}
        \item The $x_i, y_j$ are self-adjoint unitaries: $x_i=x_i^*$, $x_i^2=1$, $y_j=y_j^*$, and $y_j^2=1$,
        \item $[x_i, x_j]=0$,
        \item $[y_i, y_j]=0$,
        \item $\{x_{i\pm 1}, y_i\}=x_{i\pm 1} y_i + y_i x_{i\pm 1}=0$, and
        \item $[x_i, y_j]=0$ whenever $|i-j|\geq 2$.
    \end{enumerate}
\end{itemize}
The same presentation also holds for $\fD(J)$ when $J$ is a smooth boundary interval with $n+1$ vertical sites.

In either case, a canonical basis for this $*$-algebra is given by the monomials $x_1^{a_1}\cdots x_{n+1}^{a_{n+1}} y_1^{b_1}\cdots y_n^{b_n}$
with $a_i,b_j\in \{0,1\}$.
Thus this $*$-algebra has dimension $2^{2n+1}$ and is isomorphic to $M_{2^n}(\bbC)\oplus M_{2^n}(\bbC)$.
\end{lem}

\begin{cor}
\label{cor:Rough==Smooth}
There is an isomorphism of the nets of algebras $\fC$ and $\fD$.
\end{cor}

\begin{rem}
\label{rem:C(I)ParityPreservation}
There is a nice description of the $*$-algebra $\fC(I)$ in terms of the operators from the transverse-field Ising model.
Identify $I$ with $n+1=\#I$ contiguous sites on a 1D lattice, where each site hosts $\bbC^2$-spins.
The $*$-algebra $\fC(I)$ is isomorphic to the algebra $\fC'(I)$ generated by the operators $\sigma^X_i$ acting at site $i\in \{1,\dots, n{+}1\}$ and $\sigma^Z_j \sigma^Z_{j+1}$ acting at sites $j$ and $j+1$ for $j\in \{1,\dots, n\}$.
Indeed, this just corresponds to forgetting the third $\sigma^Z$ operator for each $D_p$ which lives on $\widetilde{I}\setminus I$, which plays no role in the abstract characterization of $\fC(I)$.

Now consider the $|\pm \rangle = \frac{1}{\sqrt{2}}(|0\rangle \pm |1\rangle)$ ONB which diagonalizes $\sigma^X_i$ (so $\bbC^2 = \bbC|+\rangle \oplus \bbC|-\rangle$).
Working in this computational basis for $\bigotimes^{n+1} \bbC^2 \cong \bbC^{2^{n+1}}$, 
we see that every operator in $\fC'(I)$ preserves the subspaces with even numbers of $|-\rangle$ and odd numbers of $|-\rangle$, which exactly corresponds to the direct sum decomposition
$\fC(I) = M_{2^n}(\bbC) \oplus M_{2^n}(\bbC)\subset M_{2^{n+1}}(\bbC)$.

Similarly, we have an isomorphism $\fD(J)\cong \fD'(J)$ where the latter algebra is generated by operators $\sigma^Z_j$ at each site $j\in\{1,\dots, n{+}1\}$ and $\sigma^X_j\sigma^X_{j+1}$ at each site $j\in\{1,\dots, n\}$ when $\#J=n+1$.
One now works in the computational ONB $\{|0\rangle, |1\rangle\}$ for $\bbC^2$, observing these operators preserve parity as before.
\end{rem}

\begin{construction}
\label{construction:ToricCodeInjective}
We now construct an isomorphism of nets of algebras from 
the fusion categorical net $\fF$ for $\cC=\fdHilb(\bbZ/2)$ from 
Example~\ref{ex:TensorCat} with $X = 1 \oplus g$, where $1, g \in \mathbb{Z}_2$,
to either $\fC'$ or $\fD'$ from Remark \ref{rem:C(I)ParityPreservation} above.
This isomorphism is essentially a planar algebra embedding from the $\cC$ planar algebra with generator $X$ to the $\bbC^2$-spin model planar algebra from \cite[Ex.~2.8]{MR4374438}.

The `box space' $\fF(J)=\End_\cC(X^{\#J})$ where $\#J=n$, is spanned by Temperley-Lieb string diagrams with $n$ top boundary points and $n$ bottom boundary points
with three types of strands, $X,1_\cC,g$ subject to the following relations (in addition to $1_\cC$ being the empty strand):
$$
\tikzmath{
\draw (0,-.5) --node[right]{$\scriptstyle X$} (0,.5);
}
=\,\,
\tikzmath{
\draw[thick, dotted] (0,-.5) --node[right]{$\scriptstyle 1_\cC$} (0,.5);
}
\oplus
\tikzmath{
\draw[thick, red] (0,-.5) --node[right]{$\scriptstyle g$} (0,.5);
}
\qquad\qquad\qquad
\tikzmath{
\draw[thick, red] (-.2,-.5) -- (-.2,.5);
\draw[thick, red] (.2,-.5) -- (.2,.5);
}
\,\,=\,\,
\tikzmath{
\draw[thick, red] (-.2,-.5) arc (180:0:.2cm);
\draw[thick, red] (-.2,.5) arc (-180:0:.2cm);
}
\qquad\qquad\qquad
\tikzmath{
\draw[thick, red] (0,0) circle (.3cm);
}
=
1.
$$
The first diagram denotes an orthogonal direct sum, which suppresses distinguished isometries $\iota_1:1_\cC \to X$ and $\iota_g: g\to X$ satisfying $\iota_1\iota_1^\dag + \iota_g\iota_g^\dag = \id_X$.
The rotations of $\iota_1$ and $\iota_g$ are their adjoints.
We may always expand every diagram with $X$ strands, so we see that $\fF(J)$ is spanned by diagrams with only $1$ and $g$ strands; we only work with these string diagrams.
Multiplication is stacking of boxes, where we get zero if the string types of $1_\cC$ and $g$ do not match.

Observe that $X^{n+1} \cong 2^n \cdot 1_\cC \oplus 2^n\cdot g$, so $\End_\cC(X^{n+1}) \cong M_{2^n}(\bbC)\oplus M_{2^n}(\bbC)$.
Let $p_1,p_g\in \End_\cC(X)$ be the orthogonal projections onto the first and second copy of $M_1(\bbC) \cong \bbC$ for $n=0$.
The operator $u:=p_1 - p_g\in \End_\cC(X)$ is a self-adjoint unitary generating $\End_\cC(X)$.
Now consider the morphism
$$
\tikzmath{
\draw (-.2,-.7) -- (-.2,.7);
\draw (.2,-.7) -- (.2,.7);
\roundNbox{fill=white}{(0,0)}{.3}{.2}{.2}{$v$}
}
:=
\tikzmath{
\draw[thick, dotted] (-.3,-.5) arc (180:0:.3cm);
\draw[thick, red] (-.3,.5) arc (-180:0:.3cm);
}
+
\tikzmath{
\draw[thick, dotted] (-.3,.5) arc (-180:0:.3cm);
\draw[thick, red] (-.3,-.5) arc (180:0:.3cm);
}
+
\tikzmath{
\draw[thick, dotted] (-.3,-.5) -- (.3,.5);
\draw[thick, red] (.3,-.5) -- (-.3,.5);
}
+
\tikzmath{
\draw[thick, dotted] (.3,-.5) -- (-.3,.5);
\draw[thick, red] (-.3,-.5) -- (.3,.5);
}
\in \End_\cC(X^2).
$$
Here, the crossings mean mapping between the two orthogonal copies of $g$ in $X^2\cong 2\cdot 1_\cC \oplus 2\cdot g$.
It is easily verified that $v$ is a self-adjoint unitary which anticommutes with $u\otimes \id_X$ and $\id_X\otimes u$.
For the algebra $\fF(J)=\End_\cC(X^{\#J})$ where $\#J=n+1$, we write $u_i$ for the copy of $u$ on the $i$-th strand, and we write $v_j$ for the copy of $v$ on the $j$-th and $(j+1)$-th strands, where $i = 1, \dots, n+1$ and $j = 1, \dots, n$.
Observe that the $u_i, v_j$ give $2n+1$ self-adjoint unitaries satisfying the relations of Lemma \ref{lem:AbstractBoundaryAlgebra}.
This gives an abstract isomorphism $\fF\cong \fC'\cong \fD'$, where $\fC'$ and $\fD'$ are the concrete realizations in terms of Pauli matrices of the abstract presentations $\fC$ and $\fD$, as defined above.
We now give a concrete isomorphism.

The $\bbC^2$-spin model planar algebra is a diagrammatic representation of the 1D spin chain with $\bbC^2$-spins at each site of $\bbZ$, with a local $M_2(\bbC)$ acting at each site.
We represent a distinguished ONB of $\bbC^2$ by an unshaded and a red node 
$\{
\tikzmath{\filldraw[thick,fill=white]circle(.1cm);}
,
\tikzmath{\fill[red]circle(.1cm);}
\}$.
(For the isomorphism $\fF\to \fC'$, 
$\tikzmath{\filldraw[thick,fill=white]circle(.1cm);}=|+\rangle$
and
$\tikzmath{\fill[red]circle(.1cm);}=|-\rangle$,
while for the isomorphism $\fF\to \fD'$,
$\tikzmath{\filldraw[thick,fill=white]circle(.1cm);}=|0\rangle$
and
$\tikzmath{\fill[red]circle(.1cm);}=|1\rangle$
in the notation of Remark~\ref{rem:C(I)ParityPreservation}).
Product tensors in $\bigotimes^n \bbC^2$ are represented by drawing $n$ nodes on a line, for example:
$$
|
\tikzmath{\filldraw[thick,fill=white]circle(.1cm);}\,
\tikzmath{\filldraw[thick,fill=white]circle(.1cm);}\,
\tikzmath{\fill[red]circle(.1cm);}\,
\tikzmath{\fill[red]circle(.1cm);}
\rangle
=
\tikzmath{
\draw (-.5,0) -- (2,0);
\filldraw[thick,fill=white](0,0)circle(.1cm);
\filldraw[thick,fill=white](.5,0)circle(.1cm);
\fill[red](1,0)circle(.1cm);
\fill[red](1.5,0)circle(.1cm);
}
$$
 Matrix units for this computational basis $|\eta\rangle \langle \xi|$ are represented by rectangles where $|\eta\rangle$ is represented by nodes on the top of the rectangle and $|\xi\rangle$ is represented by nodes on the bottom of the rectangle, e.g.:
$$
|
\tikzmath{\filldraw[thick,fill=white]circle(.1cm);}\,
\tikzmath{\filldraw[thick,fill=white]circle(.1cm);}\,
\tikzmath{\fill[red]circle(.1cm);}\,
\tikzmath{\fill[red]circle(.1cm);}
\rangle
\langle
\tikzmath{\fill[red]circle(.1cm);}\,
\tikzmath{\filldraw[thick,fill=white]circle(.1cm);}\,
\tikzmath{\filldraw[thick,fill=white]circle(.1cm);}\,
\tikzmath{\fill[red]circle(.1cm);}
|
=
\tikzmath{
\draw[thick] (-.5,0) rectangle (2,.5);
\filldraw[thick,fill=white](0,.5)circle(.1cm);
\filldraw[thick,fill=white](.5,.5)circle(.1cm);
\fill[red](1,.5)circle(.1cm);
\fill[red](1.5,.5)circle(.1cm);
\fill[red](0,0)circle(.1cm);
\filldraw[thick,fill=white](.5,0)circle(.1cm);
\filldraw[thick,fill=white](1,0)circle(.1cm);
\fill[red](1.5,0)circle(.1cm);
}
$$
Composition is the bilinear extension of stacking boxes, where we get zero unless all nodes match along the middle.

When $\#J=n+1$, we get an injection $\fF(J)\hookrightarrow \bigotimes^{n+1} M_2(\bbC)=M_{2^{n+1}}(\bbC)$ by mapping a string diagram in the $1_\cC$ and $g$ strings to the matrix unit which only remembers the shadings at the end points, e.g.:
$$
\tikzmath{
\draw[thick] (-.5,0) rectangle (3,.6);
\draw[thick, red] (0,0) -- (1,.6);
\draw[thick, red] (1.5,.6) arc (-180:0:.25cm);
\draw[thick, red] (1.5,0) -- (2.5,.6);
\filldraw[thick,fill=white](0,.6)circle(.1cm);
\filldraw[thick,fill=white](.5,.6)circle(.1cm);
\fill[red](1,.6)circle(.1cm);
\fill[red](1.5,.6)circle(.1cm);
\fill[red](2,.6)circle(.1cm);
\fill[red](2.5,.6)circle(.1cm);
\fill[red](0,0)circle(.1cm);
\filldraw[thick,fill=white](.5,0)circle(.1cm);
\filldraw[thick,fill=white](1,0)circle(.1cm);
\fill[red](1.5,0)circle(.1cm);
\filldraw[thick,fill=white](2,0)circle(.1cm);
\filldraw[thick,fill=white](2.5,0)circle(.1cm);
}
\longmapsto
\tikzmath{
\draw[thick] (-.5,0) rectangle (3,.5);
\filldraw[thick,fill=white](0,.5)circle(.1cm);
\filldraw[thick,fill=white](.5,.5)circle(.1cm);
\fill[red](1,.5)circle(.1cm);
\fill[red](1.5,.5)circle(.1cm);
\fill[red](2,.5)circle(.1cm);
\fill[red](2.5,.5)circle(.1cm);
\fill[red](0,0)circle(.1cm);
\filldraw[thick,fill=white](.5,0)circle(.1cm);
\filldraw[thick,fill=white](1,0)circle(.1cm);
\fill[red](1.5,0)circle(.1cm);
\filldraw[thick,fill=white](2,0)circle(.1cm);
\filldraw[thick,fill=white](2.5,0)circle(.1cm);
}\,.
$$
This map is well-defined and injective as all relations in $\cC$ lie in its kernel.
Indeed, observe that recabling the red strands has no effect on the location of the endpoints.
Denote by $J+1$ the interval obtained from $J$ by including one site to the right, and denote by $1+J$ the interval obtained from $J$ by adding one site to the left.
The following squares commute:
$$
\begin{matrix}
x
&\mapsto&
1\otimes x
\\
M_{2^{n+1}}(\bbC)  &\hookrightarrow & M_2(\bbC)\otimes M_{2^{n+1}}(\bbC)
\\
\rotatebox{90}{$\hookrightarrow$}&&\rotatebox{90}{$\hookrightarrow$}
\\
\fF(J)   &\hookrightarrow & \fF(1+J)
\\
\tikzmath{
\draw (-.2,-.4) -- (-.2,.4);
\draw (.2,-.4) -- (.2,.4);
\roundNbox{fill=white}{(0,0)}{.2}{.2}{.2}{\scriptsize{$f$}}
\node at (.03,-.3) {$\scriptstyle\cdots$};
\node at (.03,.3) {$\scriptstyle\cdots$};
}
&\mapsto&
\tikzmath{
\draw (-.5,-.4) -- (-.5,.4);
\draw (-.2,-.4) -- (-.2,.4);
\draw (.2,-.4) -- (.2,.4);
\roundNbox{fill=white}{(0,0)}{.2}{.2}{.2}{\scriptsize{$f$}}
\node at (.03,-.3) {$\scriptstyle\cdots$};
\node at (.03,.3) {$\scriptstyle\cdots$};
}
\end{matrix}
\qquad\qquad
\begin{matrix}
x
&\mapsto&
x\otimes 1
\\
M_{2^{n+1}}(\bbC)  &\hookrightarrow &  M_{2^{n+1}}(\bbC)\otimes M_2(\bbC)
\\
\rotatebox{90}{$\hookrightarrow$}&&\rotatebox{90}{$\hookrightarrow$}
\\
\fF(J)   &\hookrightarrow & \fF(J+1)
\\
\tikzmath{
\draw (-.2,-.4) -- (-.2,.4);
\draw (.2,-.4) -- (.2,.4);
\roundNbox{fill=white}{(0,0)}{.2}{.2}{.2}{\scriptsize{$f$}}
\node at (.03,-.3) {$\scriptstyle\cdots$};
\node at (.03,.3) {$\scriptstyle\cdots$};
}
&\mapsto&
\tikzmath{
\draw (.5,-.4) -- (.5,.4);
\draw (-.2,-.4) -- (-.2,.4);
\draw (.2,-.4) -- (.2,.4);
\roundNbox{fill=white}{(0,0)}{.2}{.2}{.2}{\scriptsize{$f$}}
\node at (.03,-.3) {$\scriptstyle\cdots$};
\node at (.03,.3) {$\scriptstyle\cdots$};
}
\end{matrix}
$$
Moreover, these squares fit into a larger commutative cube with $\fF(1+J+1)$, as adding strings/tensoring on the left and right commute.

We have thus constructed an embedding of nets of algebras from $\fF$ into the 1D spin chain.
It remains to identify the image of $\fF$ under this map.
By inspection, the image is exactly spanned by those diagrams with an even number of $\tikzmath{\fill[red]circle(.1cm);}=|-\rangle$ boundary nodes, which is exactly the subalgebra of $M_{2^{n+1}}(\bbC)$ which preserves the subspaces spanned by product tensors in the $\{
\tikzmath{\filldraw[thick,fill=white]circle(.1cm);}
,
\tikzmath{\fill[red]circle(.1cm);}
\}$
computational ONB which 
have an even or odd number of $\tikzmath{\fill[red]circle(.1cm);}=|-\rangle$ nodes.
Hence if 
$\tikzmath{\filldraw[thick,fill=white]circle(.1cm);}=|+\rangle$
and
$\tikzmath{\fill[red]circle(.1cm);}=|-\rangle$,
the image is exactly $\fC'$,
where the image of $u_i$ is $\sigma^X_i$
and the image of $v_j$ is $\sigma^Z_j\sigma_{j+1}^Z$.
If
$\tikzmath{\filldraw[thick,fill=white]circle(.1cm);}=|0\rangle$
and
$\tikzmath{\fill[red]circle(.1cm);}=|1\rangle$,
the image is exactly $\fD'$, 
where the image of $u_i$ is $\sigma^Z_i$
and the image of $v_j$ is $\sigma_j^X \sigma_{j+1}^X$.
\end{construction}

\begin{prop}
\label{prop:ToricCodeInjective}
Suppose we have rectangles $\Lambda \Subset_2 \Delta$ with $I=\partial \Lambda \cap \partial \Delta \neq \emptyset$ and let $p_\Delta$ be as in Equation~\eqref{eq:tcproj}.
\begin{enumerate}[label=(\arabic*)]
\item 
If $I$ is rough as in the left hand side of \eqref{eq:TC-BoundaryExcitations},
then $x\in \fC(I)$ and $xp_\Delta =0$ implies $x=0$.
\item 
If $I$ is smooth as in the right hand side of \eqref{eq:TC-BoundaryExcitations},
then $x\in \fD(I)$ and $xp_\Delta =0$ implies $x=0$.
\end{enumerate}
\end{prop}
\begin{proof}
We prove the first case, and the second is similar and left to the reader.
Without loss of generality, we may assume that $\partial \Delta$ is rough on all sides of $\Delta$; if this is not the case, we can replace $\Delta$ with a larger region satisfying this property.
Observe that $xp_\Delta$ preserves the space of ground states for $H_\Delta$,
which is isomorphic to $\bbC^{2^{2i+2j-1}}$
where $\Delta$ has dimensions $i\times j$ (with $n+1 \leq i+4$ as $\Lambda \Subset_2\Delta$).
Indeed, the space of ground states for $H_\Delta$
can be identified with the space of states along $\partial\Delta$ spanned by the simple tensors with an even number of $|-\rangle$'s in the $|\pm\rangle$ computational basis.
Comparing with the faithful action of $\fC(I)$ on $\bbC^{2^{n+1}}$ from Remark \ref{rem:C(I)ParityPreservation} above,
we can view the even-parity subspace of $\bbC^{2^{n+1}}$ as a subspace of the $\partial\Delta$ subspace by extending by all $|+\rangle$ outside sites in $I$, and we can view the odd-parity subspace of $\bbC^{2^{n+1}}$ as a subspace of the $\partial\Delta$ subspace by extending by all $|+\rangle$ outside sites in $I$ except for a single site $j\in \partial \Delta\setminus I$ which is always $|-\rangle$.
These two subspaces witness a faithful action of $\fC(I)$ on the ground state subspace, and thus the map $x\mapsto xp_\Delta$ is injective.
\end{proof}

To demonstrate that the axioms~\ref{LTO:Hastings}--\ref{LTO:Injective} hold, we adapt the algorithm presented in \cite[p.~6]{MR4650344}, which in turn is based on work in \cite{MR2345476}.
We thank Shuqi Wei for a simplification in Step 1 below.

\begin{algorithm}
\label{alg:PauliReduction}
Suppose we have rectangles $\Lambda \ll_2 \Delta$ or $\Lambda \Subset_2 \Delta$, and set $J:= \partial\Lambda \cap \partial\Delta$.
We assume either $J=\emptyset$ or $J$ is rough, and the case when $J$ is smooth is entirely similar.
The following algorithm expresses a local operator $a\in \fA(\Lambda)$ which is a monomial in the Pauli operators that commutes with all $A_s,B_p$ for $s,p\subset \Delta$ as a product of the $A_s,B_p$ for $s,p\subset \Lambda$ times an operator in $\fC(J)$.

\begin{enumerate}[label=\underline{Step \arabic*:}]
\item
We only apply this step if $a$ is supported entirely on two adjacent columns or two adjacent rows of sites, e.g.,
\begin{equation}
\label{eq:TwoAdjacentSites}
\tikzmath{
\draw (1.5,.05) -- (1.5,3.7);
\foreach \y in {.75,1.5,2.25,3}{
    \draw (2.25,\y) -- (1.5,\y);
}
}
\qquad\qquad\text{or}\qquad\qquad
\tikzmath{
\draw (.05,1.5,.05) -- (3.7,1.5);
\foreach \x in {.75,1.5,2.25,3}{
    \draw (\x,2.25) -- (\x,1.5);
}
}
\end{equation}
We call these two rows or columns $H$.
If $a$ is supported on a larger region, go to Step 2.

First, if $H\cap J=\emptyset$, we claim $a=1$.
The Pauli operator for $a$ on the outermost edge $\ell$ commutes with the $A_s$ for the outermost vertex $s$, and thus must be either $1_\ell$ or $\sigma^X_\ell$, where $1_\ell$ is the unit of the algebra at edge $\ell$.
But it also commutes with the $B_p$ next to it, so it must be either $1_\ell$ or $\sigma^Z_\ell$. 
We conclude it is $1_\ell$.
Working from the outside in, the result follows.

Otherwise, $H\cap J\neq \emptyset$, and arguing as in the previous paragraph, we may assume $H\subset \widetilde{J}$.
We claim $a\in \fC(J)$.
We assume $\widetilde{J}$ is oriented similar to the left hand side of \eqref{eq:TwoAdjacentSites} as in the left hand side of \eqref{eq:TC-BoundaryExcitations}; the other cases are similar.
First, consider an extremal vertical edge $\ell$ of $\widetilde{J}$ beyond any rough horizontal edges as in \eqref{eq:TC-BoundaryExcitations}.
These extremal edges must always be $1_\ell$, as they commute with the $B_p$ to the left and the $A_s$ above or below.  

Now, consider a non-extremal vertical edge $\ell$ of $\widetilde J$.  
By considering the plaquette term on the left side of $\ell$, we know that $a$ must be $1_\ell$ or $\sigma^Z_\ell$ on this edge, as this plaquette only intersects $\widetilde J$ at $\ell$.  
If $a$ is $\sigma^Z_\ell$ here, we can multiply $a$ by a $D_p\in\fC(J)$ operator to the right of $\ell$ and thus assume that $a$ acts as the identity on this edge.  
Hence, we may assume that the support of $a$ is contained in $J$, the horizontal rough edges.  
Now, considering the star terms to the left of these edges, we know that $a$ must be $1_\ell$ or $\sigma^X_\ell=C_\ell\in\fC(J)$ for each rough edge $\ell$, as the star term to the right of $\ell$ only intersects $J$ at $\ell$.  
Thus $a \in \fC(J)$.

\item 
Now suppose $a$ is supported on a larger region.
We pick a distinguished side of $\partial \Lambda$ which is necessarily either rough or smooth.
If $J=\emptyset$, any side of $\partial \Lambda$ works.
If $J\neq \emptyset$, we pick the distinguished side of $\partial \Lambda$ which is opposite $J$.
If the distinguished side is rough, go to Step 3; if it is smooth, go to Step 4.

\item
Since the distinguished edge is rough, the Pauli operators on the rough edges must commute with $A_s$ terms for vertices $s$ on the outside of the rough edges, which are necessarily in $\Delta\setminus \Lambda$.
This means these Pauli operators must be either $1_\ell$ or $\sigma^X_\ell$.
For each $\sigma^X_\ell$ that appears, multiply by an $A_s$ for the star $s\subset \Lambda$ containing $\ell$ to `cancel' the $\sigma^X_\ell$.
This will work except possibly at the two edges of the distinguished edge, where the cancelling star $s$ may not be contained in $\Lambda$.
In this case, we see the edge in question also commutes with a plaquette operator in $\Delta$, forcing the Pauli operator to be $1_\ell$.

We may now view $a'$ as a monomial supported on a smaller rectangle $\Lambda' \subset \Lambda$ with a smooth edge opposite $J$.
Go back to Step 1 with $a'$ supported on $\Lambda'$.

\item
Since the distinguished edge is smooth, the Pauli operators on the smooth edges must commute with $B_p$ terms for plaquettes $p$ on the outside of the smooth edges, which are necessarily in $\Delta\setminus \Lambda$.
This means these Pauli operators must be either $1_\ell$ or $\sigma^Z_\ell$.
For each $\sigma^Z_\ell$ that appears, multiply by a $B_p$ for the plaquette $p\subset \Lambda$ containing $\ell$ to `cancel' the $\sigma^Z_\ell$.
This will work except possibly at the two edges of the distinguished edge, where the cancelling plaquette $p$ may not be contained in $\Lambda$.
In this case, we see the edge in question also commutes with a star operator in $\Delta$, forcing the Pauli operator to be $1_\ell$.

We may now view $a'$ as a monomial supported on a smaller rectangle $\Lambda' \subset \Lambda$ with a rough edge opposite $J$.
Go back to Step 1 with $a'$ supported on $\Lambda'$.

\end{enumerate}
Observe that the product obtained from this algorithm is independent of the choice of $\Delta$ beyond that $\Lambda \ll_2\Delta$ or $\Lambda \Subset_2 \Delta$ with $J=\partial \Lambda \cap \partial \Delta$.
\end{algorithm}

Construction \ref{construction:ToricCodeInjective}
and
Algorithm \ref{alg:PauliReduction}
immediately imply the following theorem.

\begin{thm}
\label{thm:TC-LTO}
The axioms \ref{LTO:Hastings}--\ref{LTO:Injective} hold for the Toric Code
where $\fB(I)$ is $\fC(I)$ or $\fD(I)$ depending on whether we choose $I$ to be rough or smooth.
\end{thm}

\begin{rem}
\label{rem:DiagonalRegions}
The above analysis did not exactly follow the conventions of \S\ref{sec:LTO}.
Indeed, our Hilbert spaces were placed on edges in our lattice, which is not strictly speaking  $\bbZ^2\subset \bbR^2$.
However, if we draw the edges as points and rotate our heads $45^\circ$, we again see a $\bbZ^2$ lattice, and we can consider rectangles on this lattice.
In this setup, both the star operators $A_s$ and the plaquette operators appear as plaquette operators.
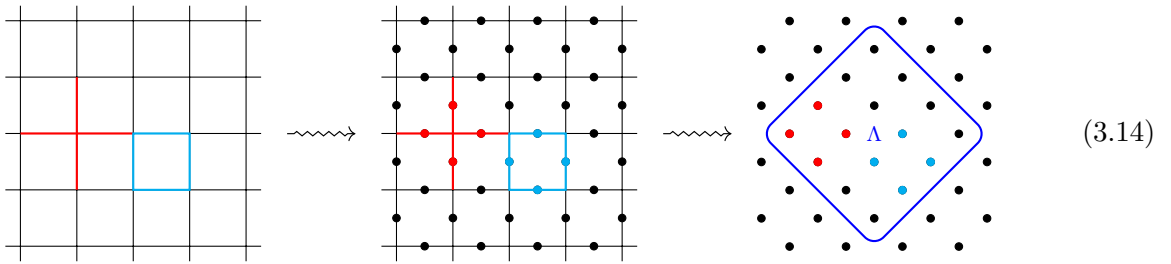
\begin{equation}
\label{eq:DiagonalRectangles}
\tikzmath{
\draw[step = .75] (-.2,-.2) grid (3.2,3.2);
\draw[thick, red] (0,1.5) -- (1.5,1.5);
\draw[thick, red] (.75,.75) -- (.75,2.25);
\draw[thick, cyan] (1.5,.75) rectangle (2.25,1.5);
}
\begin{tikzcd}
\mbox{}
\arrow[r,squiggly]
&
\mbox{}
\end{tikzcd}
\tikzmath{
\draw[step = .75] (-.2,-.2) grid (3.2,3.2);
\foreach \x in {.375,1.125,1.875,2.625}{
\foreach \y in {0,.75,1.5,2.25,3}{
\filldraw (\x,\y) circle (.05cm);
\filldraw (\y,\x) circle (.05cm);
}}
\draw[thick, red] (0,1.5) -- (1.5,1.5);
\draw[thick, red] (.75,.75) -- (.75,2.25);
\filldraw[red] (.375,1.5) circle (.05cm);
\filldraw[red] (1.125,1.5) circle (.05cm);
\filldraw[red] (.75,1.125) circle (.05cm);
\filldraw[red] (.75,1.875) circle (.05cm);
\draw[thick, cyan] (1.5,.75) rectangle (2.25,1.5);
\filldraw[cyan] (1.875,.75) circle (.05cm);
\filldraw[cyan] (1.875,1.5) circle (.05cm);
\filldraw[cyan] (1.5,1.125) circle (.05cm);
\filldraw[cyan] (2.25,1.125) circle (.05cm);
}
\begin{tikzcd}
\mbox{}
\arrow[r,squiggly]
&
\mbox{}
\end{tikzcd}
\tikzmath{
\foreach \x in {.375,1.125,1.875,2.625}{
\foreach \y in {0,.75,1.5,2.25,3}{
\filldraw (\x,\y) circle (.05cm);
\filldraw (\y,\x) circle (.05cm);
}}
\draw[thick, blue, rounded corners=5pt] (2.25,.75) -- (3,1.5) -- (1.5,3) -- (0,1.5) -- (1.5,0) -- (2.25,.75);
\node[blue] at (1.5,1.5) {$\scriptstyle \Lambda$};
\filldraw[red] (.375,1.5) circle (.05cm);
\filldraw[red] (1.125,1.5) circle (.05cm);
\filldraw[red] (.75,1.125) circle (.05cm);
\filldraw[red] (.75,1.875) circle (.05cm);
\filldraw[cyan] (1.875,.75) circle (.05cm);
\filldraw[cyan] (1.875,1.5) circle (.05cm);
\filldraw[cyan] (1.5,1.125) circle (.05cm);
\filldraw[cyan] (2.25,1.125) circle (.05cm);
}
\end{equation}
We can define the $p_\Lambda$ as the product of the local commuting projectors, but we observe that our net of projections $(p_\Lambda)$ is only translation invariant by \emph{even} translations.
One must `coarse grain' in order to obtain a true translation invariant net of projections.

We can now pick a 1D hyperplane $\cK$ in this rotated $\bbZ^2$ lattice and consider rectangles whose boundaries intersect $\cK$.
By an algorithm analogous to Algorithm \ref{alg:PauliReduction}, 
for every sufficiently large interval $I$ in $\cK$ and sufficiently large $\Delta$ with $\partial \Delta\cap \cK = I$, $p_\Delta a p_\Delta = \gamma_a p_\Delta$ for a unique operator $\gamma_a$ (independent of $\Delta$) in an algebra $\fE(I)$ generated by certain monomials of Pauli operators.
Based on the parity of sites in $I$, these generating Pauli monomials can be taken to be of the form
$$
1\otimes \cdots \otimes 1 \otimes \sigma^Z \otimes \sigma^Z \otimes 1\otimes \cdots \otimes 1
\qquad\text{and}\qquad
1\otimes \cdots \otimes 1 \otimes \sigma^X \otimes \sigma^X \otimes 1\otimes \cdots \otimes 1,
$$
where the $\sigma^Z$ always occur on, say, sites $2i$ and $2i+1$ and the $\sigma^X$ always occur on sites $2i+1$ and $2i+2$, and we again write $1$ for the unit operator of the local algebra at an edge.
For example, when $I$ is the northeast edge of $\Lambda$ in \eqref{eq:DiagonalRectangles}, $\fE(I)$ is generated by
$$
\sigma^X\otimes \sigma^X \otimes 1 \otimes 1,
\qquad
1\otimes \sigma^Z \otimes \sigma^Z \otimes 1,
\qquad
\text{and}
\qquad
1 \otimes 1\otimes \sigma^X\otimes \sigma^X.
$$
It is clear that these operators satisfy the relations of Lemma \ref{lem:AbstractBoundaryAlgebra}, but observe that the algebras $\fE(I)$ grow at roughly half the rate of $\fC(I)$ or $\fD(I)$.
Thus the net of algebras $\fE$ is `coarse grained' bounded spread isomorphic to $\fC$ and $\fD$.
While this equivalence relation can be made rigorous, we leave it to a future paper as it would take us too far afield. We also note that boundary algebras of Toric Code have recently appeared in a slightly different form in \cite{2304.01277} as a host for measuremnt-based quantum cellular auotmata. We plan to expand on this connection in future work.
\end{rem}

\section{Example: Levin-Wen string nets}
\label{sec:LWStringNet}

In this section, we prove that the Levin-Wen string net model \cite{PhysRevB.71.045110,PhysRevB.103.195155} for a unitary fusion category (UFC) $\cC$
has a net of projections $\Lambda \mapsto p_\Lambda$
satisfying \ref{LTO:Hastings}--\ref{LTO:Injective}.
We first recall the definition of the model following 
\cite[\S2]{2305.14068} which was adapted from \cite{MR2942952,MR3204497} before we define the projections $p_\Lambda$.

Let $\cC$ denote a UFC, and denote its  quantum double (Drinfeld center) by $Z(\cC)$ \cite[\S7.13]{MR3242743}.
We write $\cC(a\to b)$ to denote the space of morphisms $a\to b$ in $\cC$.
For simplicity we will only consider the model on a square lattice in two dimensions.
Schematically, the Hilbert space can be visualized as follows, where the black edges carry labels from $\Irr(\cC)$.
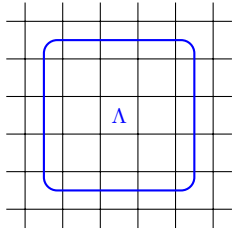
$$
\tikzmath{
\draw[step=.5,black,thin] (0.25,0.25) grid (3.25,3.25);
\draw[thick, blue, rounded corners=5pt] (.75,.75) rectangle (2.75,2.75);
\node[blue] at (1.75,1.75) {$\scriptstyle \Lambda$};
}
$$
Here, we read from bottom left to top right.
The total Hilbert space is the tensor product of local Hilbert spaces over all sites:
$$
\tikzmath{
\draw[thick] (-.5,0) node[left]{$\scriptstyle a$} -- (.5,0) node[right]{$\scriptstyle d$};
\draw[thick] (0,-.5) node[below]{$\scriptstyle b$} -- (0,.5) node[above]{$\scriptstyle c$};
\filldraw (0,0) circle (.05cm);
\node at (-.2,-.2) {$\scriptstyle v$};
\draw[blue!50, very thin] (-.5,.5) -- (.5,-.5);
{\draw[blue!50, very thin, -stealth] ($ (-.3,.3) - (.1,.1)$) to ($ (-.3,.3) + (.1,.1)$);}
{\draw[blue!50, very thin, -stealth] ($ (.3,-.3) - (.1,.1)$) to ($ (.3,-.3) + (.1,.1)$);}
}
\qquad
\longleftrightarrow
\qquad
\cH_v 
:= 
\bigoplus_{
a,b,c,d\in \Irr(\cC)
}
\cC(a\otimes b \to c\otimes d)
$$
where the direct sum is orthogonal.
The space $\cH_v$ is equipped with the `skein-module' inner product
$$
\left\langle
\tikzmath{
\draw (-.5,0) node[left]{$\scriptstyle a$} -- (.5,0) node[right]{$\scriptstyle d$};
\draw (0,-.5) node[below]{$\scriptstyle b$} -- (0,.5) node[above]{$\scriptstyle c$};
\node at (.2,-.2) {$\scriptstyle \xi$};
}
\middle|
\tikzmath{
\draw (-.5,0) node[left]{$\scriptstyle a'$} -- (.5,0) node[right]{$\scriptstyle d'$};
\draw (0,-.5) node[below]{$\scriptstyle b'$} -- (0,.5) node[above]{$\scriptstyle c'$};
\node at (.2,-.2) {$\scriptstyle \xi'$};
}
\right\rangle
=
\delta_{a=a'}
\delta_{b=b'}
\delta_{c=c'}
\delta_{d=d'}
\frac{1}{\sqrt{d_ad_bd_cd_d}}
\cdot
\tr_\cC(\xi^\dag\circ \xi').
$$
Here, $\dag$ is the dagger structure on $\cC$ and $\tr_\cC$ is the categorical trace using the unique unitary spherical structure \cite{MR2091457,MR4133163}.

Consider now a rectangle $\Lambda$ in our lattice $\cL$.
We consider the canonical spin system from this setup as in Example \ref{ex:SpinSystem}, i.e.,
$\fA(\Lambda):= \bigotimes_{v\in \Lambda} B(\cH_v)$.
We set $\fA:= \varinjlim \fA(\Lambda) = \bigotimes_v B(\cH_v)$.

For a rectangle $\Lambda \subset \cL$, we say
\begin{itemize}
\item 
an edge/link $\ell\subset \Lambda$ if  the two vertices at the endpoints of $\ell$ are contained in $\Lambda$, and
\item
a face/plaquette $p\subset \Lambda$ if the four vertices at the corners of $\ell$ are contained in $\Lambda$.
\end{itemize}

For each edge $\ell\subset \Lambda$, we have an orthogonal projector $A_\ell\in \fA(\Lambda)$ which enforces that the edge labels on $\ell$ match from either side:
$$
A_\ell
\left(
\tikzmath{
\draw (-.5,0) node[left]{$\scriptstyle a$} -- (.5,0) node[right]{$\scriptstyle d$};
\draw (0,-.5) node[below]{$\scriptstyle b$} -- (0,.5) node[above]{$\scriptstyle c$};
\node at (.2,-.2) {$\scriptstyle \eta$};
}
\otimes
\tikzmath{
\draw (-.5,0) node[left]{$\scriptstyle e$} -- (.5,0) node[right]{$\scriptstyle h$};
\draw (0,-.5) node[below]{$\scriptstyle f$} -- (0,.5) node[above]{$\scriptstyle g$};
\node at (.2,-.2) {$\scriptstyle \xi$};
}
\right)
:=
\delta_{d=e}
\tikzmath{
\draw (-.5,0) node[left]{$\scriptstyle a$} -- (.5,0) node[right]{$\scriptstyle d$};
\draw (0,-.5) node[below]{$\scriptstyle b$} -- (0,.5) node[above]{$\scriptstyle c$};
\node at (.2,-.2) {$\scriptstyle \eta$};
}
\otimes
\tikzmath{
\draw (-.5,0) node[left]{$\scriptstyle d$} -- (.5,0) node[right]{$\scriptstyle h$};
\draw (0,-.5) node[below]{$\scriptstyle f$} -- (0,.5) node[above]{$\scriptstyle g$};
\node at (.2,-.2) {$\scriptstyle \xi$};
}
$$
We define $p_\Lambda^A := \prod_{\ell\in \Lambda} A_\ell$.

For each plaquette $p\subset \Lambda$, we have an orthogonal projector $B_p \in p^A_\Lambda\fA(\Lambda)p^A_\Lambda$ using the usual definition from the Levin-Wen local Hamiltonian \cite{PhysRevB.71.045110,MR3204497,PhysRevB.103.195155}:
\begin{align*}
\frac{1}{D_\cC}\sum_{s\in \Irr(\cC)}
d_s\cdot
\tikzmath{
\draw[step=1.0,black,thin] (0.5,0.5) grid (2.5,2.5);
\node at (2.3,.8) {$\scriptstyle \xi_{2,1}$};
\node at (.7,2.2) {$\scriptstyle \xi_{1,2}$};
\node at (.7,.8) {$\scriptstyle \xi_{1,1}$};
\node at (2.3,2.2) {$\scriptstyle \xi_{2,2}$};
\filldraw[knot, thick, blue, rounded corners=5pt, fill=gray!30] (1.15,1.15) rectangle (1.85,1.85);
\node[blue] at (1.3,1.5) {$\scriptstyle s$};
\node at (1.5,.85) {$\scriptstyle i$};
\node at (2.15,1.5) {$\scriptstyle j$};
\node at (1.5,2.15) {$\scriptstyle k$};
\node at (.85,1.5) {$\scriptstyle \ell$};
\node at (.3,1) {$\scriptstyle a$};
\node at (1,.3) {$\scriptstyle b$};
\node at (2,.3) {$\scriptstyle c$};
\node at (2.7,1) {$\scriptstyle d$};
\node at (2.7,2) {$\scriptstyle e$};
\node at (2,2.7) {$\scriptstyle f$};
\node at (1,2.7) {$\scriptstyle g$};
\node at (.3,2) {$\scriptstyle h$};
}
&=
\frac{1}{D_\cC}\sum_{m,n,q,r,s \in \Irr(\cX)}
\frac{\sqrt{d_md_nd_qd_r}}{d_s\sqrt{d_id_jd_kd_\ell}}
\tikzmath{
\draw[step=1.0,black,thin] (0.5,0.5) grid (2.5,2.5);
\draw[thick, blue] (1.3,1) -- (1,1.3);
\draw[thick, blue] (1.7,1) -- (2,1.3);
\draw[thick, blue] (1.3,2) -- (1,1.7);
\draw[thick, blue] (1.7,2) -- (2,1.7);
\fill[rounded corners=5pt, gray!30] (1.15,1.15) rectangle (1.85,1.85);
\fill[fill=green] (1.3,1) circle (.05cm);
\fill[fill=green] (1.7,1) circle (.05cm);
\fill[fill=red] (2,1.3) circle (.05cm);
\fill[fill=red] (2,1.7) circle (.05cm);
\fill[fill=yellow] (1.3,2) circle (.05cm);
\fill[fill=yellow] (1.7,2) circle (.05cm);
\fill[fill=orange] (1,1.3) circle (.05cm);
\fill[fill=orange] (1,1.7) circle (.05cm);
\node at (2.3,.8) {$\scriptstyle \xi_{2,1}$};
\node at (.7,2.2) {$\scriptstyle \xi_{1,2}$};
\node at (.7,.8) {$\scriptstyle \xi_{1,1}$};
\node at (2.3,2.2) {$\scriptstyle \xi_{2,2}$};
\node at (1.15,.85) {$\scriptstyle i$};
\node at (1.85,.85) {$\scriptstyle i$};
\node at (2.15,1.15) {$\scriptstyle j$};
\node at (2.15,1.85) {$\scriptstyle j$};
\node at (1.15,2.15) {$\scriptstyle k$};
\node at (1.85,2.15) {$\scriptstyle k$};
\node at (.85,1.15) {$\scriptstyle \ell$};
\node at (.85,1.85) {$\scriptstyle \ell$};
\node at (1.5,.85) {$\scriptstyle m$};
\node at (2.15,1.5) {$\scriptstyle n$};
\node at (1.5,2.15) {$\scriptstyle q$};
\node at (.85,1.5) {$\scriptstyle r$};
\node at (.3,1) {$\scriptstyle a$};
\node at (1,.3) {$\scriptstyle b$};
\node at (2,.3) {$\scriptstyle c$};
\node at (2.7,1) {$\scriptstyle d$};
\node at (2.7,2) {$\scriptstyle e$};
\node at (2,2.7) {$\scriptstyle f$};
\node at (1,2.7) {$\scriptstyle g$};
\node at (.3,2) {$\scriptstyle h$};
}
\\&=
\sum_{\eta}
C(\xi,\eta)
\tikzmath{
\draw[step=1.0,black,thin] (0.5,0.5) grid (2.5,2.5);
\fill[rounded corners=5pt, gray!30] (1.15,1.15) rectangle (1.85,1.85);
\node at (2.3,.8) {$\scriptstyle \eta_{2,1}$};
\node at (.7,2.2) {$\scriptstyle \eta_{1,2}$};
\node at (.7,.8) {$\scriptstyle \eta_{1,1}$};
\node at (2.3,2.2) {$\scriptstyle \eta_{2,2}$};
\node at (1.5,.85) {$\scriptstyle m$};
\node at (2.15,1.5) {$\scriptstyle n$};
\node at (1.5,2.15) {$\scriptstyle q$};
\node at (.85,1.5) {$\scriptstyle r$};
\node at (.3,1) {$\scriptstyle a$};
\node at (1,.3) {$\scriptstyle b$};
\node at (2,.3) {$\scriptstyle c$};
\node at (2.7,1) {$\scriptstyle d$};
\node at (2.7,2) {$\scriptstyle e$};
\node at (2,2.7) {$\scriptstyle f$};
\node at (1,2.7) {$\scriptstyle g$};
\node at (.3,2) {$\scriptstyle h$};
}\,.
\end{align*}
Here, $D_\cC:= \sum_{c\in \Irr(\cC)}d_c^2$ is the \emph{global dimension} of $\cC$, and we use the convention from \cite{MR3663592} 
writing a pair of shaded vertices to denote summing over an orthonormal basis for the trivalent skein module (see \eqref{eq:SkeinModuleInnerProduct} below) and its dual.
Using these conventions, the fusion relation in $\cC$ is given by
\begin{equation}
\label{eq:Fusion}
\sum_{c\in \Irr(\cX)}
\sqrt{d_c}
\begin{tikzpicture}[baseline=-.1cm]
	\draw (.2,-.6) -- (0,-.3) -- (-.2,-.6);
	\draw (.2,.6) -- (0,.3) -- (-.2,.6);
	\draw (0,-.3) -- (0,.3);
	\fill[fill=blue] (0,-.3) circle (.05cm);
	\fill[fill=blue] (0,.3) circle (.05cm);	
	\node at (-.2,-.8) {\scriptsize{$a$}};
	\node at (.2,-.8) {\scriptsize{$b$}};
	\node at (-.2,.8) {\scriptsize{$a$}};
	\node at (.2,.8) {\scriptsize{$b$}};
	\node at (.2,0) {\scriptsize{$c$}};
\end{tikzpicture}
\,=\,\sqrt{d_ad_b}\cdot
\begin{tikzpicture}[baseline=-.1cm]
	\draw (.2,-.6) -- (.2,.6);
	\draw (-.2,-.6) -- (-.2,.6);
	\node at (-.2,-.8) {\scriptsize{$a$}};
	\node at (.2,-.8) {\scriptsize{$b$}};
\end{tikzpicture}
\end{equation}
These conventions have also been used in other descriptions of the Levin-Wen model \cite{1502.06845,MR4640433,MR4642306}.

\begin{lem}[{\cite[Lem.~2.8]{2305.14068}, see also \cite{0907.2204}}]
For
$$
\xi=
\tikzmath{
\draw[step=1.0,black,thin] (0.5,0.5) grid (2.5,2.5);
\fill[rounded corners=5pt, gray!30] (1.15,1.15) rectangle (1.85,1.85);
\node at (2.3,.8) {$\scriptstyle \xi_{2,1}$};
\node at (.7,1.8) {$\scriptstyle \xi_{1,2}$};
\node at (.7,.8) {$\scriptstyle \xi_{1,1}$};
\node at (2.3,1.8) {$\scriptstyle \xi_{2,2}$};
\node at (.3,1) {$\scriptstyle a$};
\node at (1,.3) {$\scriptstyle b$};
\node at (2,.3) {$\scriptstyle c$};
\node at (2.7,1) {$\scriptstyle d$};
\node at (2.7,2) {$\scriptstyle e$};
\node at (2,2.7) {$\scriptstyle f$};
\node at (1,2.7) {$\scriptstyle g$};
\node at (.3,2) {$\scriptstyle h$};
}
\qquad\qquad
\xi'=
\tikzmath{
\draw[step=1.0,black,thin] (0.5,0.5) grid (2.5,2.5);
\fill[rounded corners=5pt, gray!30] (1.15,1.15) rectangle (1.85,1.85);
\node at (2.3,.8) {$\scriptstyle \xi_{2,1}'$};
\node at (.7,1.8) {$\scriptstyle \xi_{1,2}'$};
\node at (.7,.8) {$\scriptstyle \xi_{1,1}'$};
\node at (2.3,1.8) {$\scriptstyle \xi_{2,2}'$};
\node at (.3,1) {$\scriptstyle a'$};
\node at (1,.3) {$\scriptstyle b'$};
\node at (2,.3) {$\scriptstyle c'$};
\node at (2.7,1) {$\scriptstyle d'$};
\node at (2.7,2) {$\scriptstyle e'$};
\node at (2,2.7) {$\scriptstyle f'$};
\node at (1,2.7) {$\scriptstyle g'$};
\node at (.3,2) {$\scriptstyle h'$};
}
\,,
$$
the constant $C(\xi',\xi)$ above is given by
$$
C(\xi',\xi)
=
\delta_{a=a'}\cdots\delta_{h=h'}
\frac{1}{D_\cC\sqrt{d_a\cdots d_h}}
\cdot
\tikzmath{
\clip (-1.8,-1.8) rectangle (4.8,4.8);
\draw (0,0) rectangle (1,1);
\draw (2,2) rectangle (3,3);
\draw (1,1) to[out=0, in=-90] (2,2);
\draw (1,1) to[out=90, in=180] (2,2);
\draw (1,0) to[out=0, in=-90] (3,2);
\draw (1,0) .. controls ++(-90:2cm) and ++(0:2cm) .. (3,2);
\draw (0,1) to[out=90, in=180] (2,3);
\draw (0,1) .. controls ++(180:2cm) and ++(90:2cm) .. (2,3);
\draw (0,0) .. controls ++(180:5cm) and ++(90:5cm) .. (3,3);
\draw (0,0) .. controls ++(-90:5cm) and ++(0:5cm) .. (3,3);
\node at (.7,-.2) {$\scriptstyle \xi_{2,1}'$};
\node at (.3,.8) {$\scriptstyle \xi_{1,2}'$};
\node at (-.3,-.2) {$\scriptstyle \xi_{1,1}'$};
\node at (1.3,.8) {$\scriptstyle \xi_{2,2}'$};
\node at (3.3,2.2) {$\scriptstyle \xi_{2,1}^\dag$};
\node at (2.3,2.7) {$\scriptstyle \xi_{1,2}^\dag$};
\node at (3.3,2.7) {$\scriptstyle \xi_{1,1}^\dag$};
\node at (2.3,1.7) {$\scriptstyle \xi_{2,2}^\dag$};
\node at (-1,.3) {$\scriptstyle a$};
\node at (0,-1) {$\scriptstyle b$};
\node at (.9,-.5) {$\scriptstyle c$};
\node at (1.2,.2) {$\scriptstyle d$};
\node at (1.8,1.1) {$\scriptstyle e$};
\node at (1.3,1.3) {$\scriptstyle f$};
\node at (.2,1.2) {$\scriptstyle g$};
\node at (-.2,.8) {$\scriptstyle h$};
}\,.
$$
Hence the $B_p$ are orthogonal projections and $[B_p,B_q]=0$ for $p\neq q$.
\end{lem}

For a rectangle $\Lambda$, we define $p_\Lambda^B := \prod_{p\subset \Lambda} B_p$.
We now define our net of projections.

\begin{defn}[Net of projections for the Levin-Wen string net]
For each rectangle $\Lambda$ in $\cL$, define 
$$
p_\Lambda := p_\Lambda^A p_\Lambda^B = \prod_{p\subset \Lambda} B_p \prod_{\ell\subset\Lambda} A_\ell.
$$
Clearly $\Lambda \subset \Delta$ implies $p_\Delta \leq p_\Lambda$.
\end{defn}

There is a nice description of the image of $p^B_\Lambda$ on $p^A_\Lambda\bigotimes_{v\in \Lambda} \cH_v$ in terms of the \emph{skein module} $\cS_\cC(\#\partial\Lambda)$ for $\cC$ with $n:=\#\partial \Lambda$ boundary points, defined as the orthogonal direct sum
$$
\bigoplus_{x_1,\dots, x_n} \cC(1\to x_1\otimes\cdots \otimes x_n)
$$
with inner product given by
\begin{equation}
\label{eq:SkeinModuleInnerProduct}
\left\langle
\tikzmath{
\draw (-.4,.3) --node[left]{$\scriptstyle c_1$} (-.4,.7);
\node at (0,.5) {$\cdots$};
\draw (.4,.3) --node[right]{$\scriptstyle c_n$} (.4,.7);
\roundNbox{}{(0,0)}{.3}{.3}{.3}{$\xi$}
}
\middle|
\tikzmath{
\draw (-.4,.3) --node[left]{$\scriptstyle c_1'$} (-.4,.7);
\node at (0,.5) {$\cdots$};
\draw (.4,.3) --node[right]{$\scriptstyle c_n'$} (.4,.7);
\roundNbox{}{(0,0)}{.3}{.3}{.3}{$\xi'$}
}
\right\rangle
:=
\delta_{c_1=c_1'}\cdots \delta_{c_n=c_n'}
\frac{1}{\sqrt{d_{c_1}\cdots d_{c_n}}}
\xi^\dag \circ \xi'.
\end{equation}
There is an obvious linear map
$\eval : p^A_\Lambda\bigotimes_{v\in \Lambda} \cH_v \to \cS_\cC(\#\partial\Lambda)$ called \emph{evaluation} which writes a vector in $p^A_\Lambda\bigotimes_{v\in \Lambda} \cH_v$ as an element of $\cS_\cC(\#\partial \Lambda)$.

The proof of the following lemma from \cite{MR3204497} appears in \cite[Thm.~2.9]{2305.14068}.

\begin{lem}[\cite{MR3204497}]
\label{lem:KongSkeinModule}
For a rectangle $\Lambda$, 
on $p^A_\Lambda\bigotimes_{v\in \Lambda} \cH_v$,
$$
p^B_\Lambda = D^{-\#p\subset \Lambda}_{\cC} \eval^\dag\circ \eval.
$$
where $\#p \subset \Lambda$ is the number of plaquettes internal to $\Lambda$.
Hence $p_\Lambda\bigotimes_{v\in \Lambda} \cH_v$ is unitarily isomorphic to $\cS_\cC(\#\partial \Lambda)$ via the map $D^{-\#p/2}_{\cC}\eval $.
\end{lem}

The article \cite{MR4814692} studied the 1D net of algebras on $\bbZ$ given by $\fF(I):= \End_\cC(X^{\# I})$ where $X:= \bigoplus_{c\in \Irr(\cC)} c$.

\begin{defn}[Boundary action of $\fF(I)=\End_\cC(X^{\#I})$ on $\fA(\Lambda)$]
\label{defn:BoundaryAction}
Suppose $\Lambda \subset \bbZ^2$ is a rectangle, $I\subset \partial \Lambda$ is an interval, and $\varphi \in \End_\cC(X^{\# I})$.
Then $\varphi$ defines an operator $\Gamma_\varphi\in\fA(\Lambda)$ 
given by 
first applying $p_\Lambda$ and then gluing $\varphi$ into the $I$-boundary on the outer edge of $\Lambda$.
If $I$ is on the top or right side of $\Lambda$, then the map $\Gamma: \End_\cC(X^{\#I}) \to p_\Lambda\fA(\Lambda)p_\Lambda$ is a $*$-algebra map; if $I$ is on the bottom or left side of $\Lambda$, then $\Gamma$ is a $*$-anti-algebra map.
We must also multiply by a ratio of 4th roots of quantum dimensions in order to make this map a $*$-(anti-)algebra map.

We give a graphical example below with $\#I=3$, where we assume $I$ is at the top of $\Lambda$ and we suppress $\Lambda\setminus I$
from the picture. 
If $\varphi \in \cC(r\otimes s\otimes t\to x\otimes y\otimes z)\subset \End(X^3)$, then $\Gamma_\varphi$ is the map
\begin{align*}
\tikzmath{
\draw (-.5,0) node[left]{$\scriptstyle a$} -- (.5,0) node[right]{$\scriptstyle d$};
\draw (0,-.5) node[below]{$\scriptstyle b$} -- (0,.5) node[above]{$\scriptstyle c$};
\node at (.2,-.2) {$\scriptstyle \eta$};
}
\otimes
\tikzmath{
\draw (-.5,0) node[left]{$\scriptstyle e$} -- (.5,0) node[right]{$\scriptstyle h$};
\draw (0,-.5) node[below]{$\scriptstyle f$} -- (0,.5) node[above]{$\scriptstyle g$};
\node at (.2,-.2) {$\scriptstyle \xi$};
}
\otimes
\tikzmath{
\draw (-.5,0) node[left]{$\scriptstyle i$} -- (.5,0) node[right]{$\scriptstyle \ell$};
\draw (0,-.5) node[below]{$\scriptstyle j$} -- (0,.5) node[above]{$\scriptstyle k$};
\node at (.2,-.2) {$\scriptstyle \zeta$};
}
&\overset{p_\Lambda}{\longmapsto}
\delta_{d=e}\delta_{h=i}
\tikzmath{
\draw (-.5,0) node[left]{$\scriptstyle a$} -- (0,0) --node[below]{$\scriptstyle d$} (1,0) --node[below]{$\scriptstyle h$} (2,0) -- (2.5,0) node[right]{$\scriptstyle \ell$};
\draw (0,-.5) node[below]{$\scriptstyle b$} -- (0,0) --node[left]{$\scriptstyle c$} (0,.5);
\draw (1,-.5) node[below]{$\scriptstyle f$} -- (1,0) --node[left]{$\scriptstyle g$} (1,.5);
\draw (2,-.5) node[below]{$\scriptstyle j$} -- (2,0) --node[left]{$\scriptstyle k$} (2,.5);
\node at (.2,-.2) {$\scriptstyle \eta$};
\node at (1.2,-.2) {$\scriptstyle \xi$};
\node at (2.2,-.2) {$\scriptstyle \zeta$};
}
\\&\longmapsto
\delta_{d=e}\delta_{h=i}
\delta_{c=r}\delta_{g=s}\delta_{k=t}
\left(
\frac{d_xd_yd_z}{d_rd_sd_t}
\right)^{1/4}
\tikzmath{
\draw (-.5,0) node[left]{$\scriptstyle a$} -- (0,0) --node[below]{$\scriptstyle d$} (1,0) --node[below]{$\scriptstyle h$} (2,0) -- (2.5,0) node[right]{$\scriptstyle \ell$};
\draw (0,-.5) node[below]{$\scriptstyle b$} -- (0,0) --node[left]{$\scriptstyle c$} (0,.5) -- (0,1.5) node[above]{$\scriptstyle x$};
\draw (1,-.5) node[below]{$\scriptstyle f$} -- (1,0) --node[left]{$\scriptstyle g$} (1,.5) -- (1,1.5) node[above]{$\scriptstyle y$};
\draw (2,-.5) node[below]{$\scriptstyle j$} -- (2,0) --node[left]{$\scriptstyle k$} (2,.5) -- (2,1.5) node[above]{$\scriptstyle z$};
\node at (.2,-.2) {$\scriptstyle \eta$};
\node at (1.2,-.2) {$\scriptstyle \xi$};
\node at (2.2,-.2) {$\scriptstyle \zeta$};
\roundNbox{fill=white}{(1,.8)}{.3}{1}{1}{$\varphi$}
}
\end{align*}
This final picture must be interpreted as a vector in $\bigotimes_{v\in I}\cH_v$ by decomposing into simples in the usual way.
\end{defn}

\begin{lem}
\label{lem:GammaInjective}
Suppose $\Lambda$ is a rectangle and $I\subset \partial \Lambda$ is an interval which is on the top or right side of $\Lambda$.
Whenever $\Lambda \subseteq \Delta$ with $I\subset \partial \Lambda \cap \partial \Delta$ (compare with Remark \ref{rem:StrategyForExamples}),
$[\Gamma_\varphi,p_\Delta]=0$ and $\Gamma_\varphi p_\Delta=0$ implies $\varphi=0$.
In particular, the map $\Gamma: \fF(I)\to p_\Lambda \fA(\Lambda)p_\Lambda$ given by $\varphi\mapsto \Gamma_\varphi$ is injective.
\end{lem}
\begin{proof}
By Lemma \ref{lem:KongSkeinModule},
$p_\Delta\bigotimes_{v\in \Delta} \cH_v$ is unitarily isomorphic to $\cS_\cC(\#\partial \Delta)$, and $\Gamma_\varphi$ acts by gluing $\varphi$ onto the sites in $I\subset \partial \Delta$.
Clearly $\Gamma_\varphi$ preserves $p_\Delta\bigotimes_{v\in \Delta} \cH_v$, so $[\Gamma_\varphi, p_\Delta]=0$.
That $\Gamma_\varphi p_\Delta=0$ implies $\varphi=0$ is easily verified using the positive definite skein module inner product.
Finally, if $\Gamma_\varphi p_\Lambda=0$, then $\Gamma_\varphi p_\Delta=0$, so $\varphi=0$.
\end{proof}

The first part of the next theorem shows that the Levin-Wen string net model satisfies \ref{LTO:Hastings}, which implies (TQO1) and (TQO2) of \cite{MR2742836} by Proposition \ref{prop:LTQO}.
We supply a short conceptual proof using (an algebra Morita equivalent to) the tube algebra \cite{MR1782145,MR1832764,MR1966525}.
The axioms (TQO1) and (TQO2) for the Levin-Wen model were originally proven in \cite{QIU2020168318}.\footnote{The article  \cite{QIU2020168318} proves the modified (TQO2) axiom from \cite{MR2842961}, which is implied by the (TQO2) condition from \cite{MR2742836} by \cite[Cor.~2.1]{MR2742836}.}

\begin{thm}
\label{thm:LTQO}
The Levin-Wen string net model satisfies \ref{LTO:Hastings}--\ref{LTO:Injective} with $s=1$. 
\end{thm}
\begin{proof}
\item[\underline{\ref{LTO:Hastings}:}]
Suppose $\Lambda \ll_1 \Delta$.
We define the $\partial \Delta$-\emph{tube algebra}  $\operatorname{Tube}_\cC(\partial \Delta)$ with internal and external boundaries given by $\partial \Delta$:
$$
\tikzmath[scale=.5]{
\draw[step=1.0,black,thin] (0.5,0.5) grid (6.5,6.5);
\draw[thick, blue, rounded corners=5pt] (1.5,3.5) -- (1.5,5.5) -- (5.5,5.5) -- (5.5,1.5) -- (1.5,1.5) -- (1.5,3.5);
\node[blue] at (3.5,3.5) {$\scriptstyle \Delta$};
}
\qquad\qquad
\operatorname{Tube}_\cC(\partial \Delta)
=
\left\{\,
\tikzmath[scale=.5]{
\filldraw[thick, blue, fill=white, rounded corners=5pt] (1.5,3.5) -- (1.5,5.5) -- (5.5,5.5) -- (5.5,1.5) -- (1.5,1.5) -- (1.5,3.5);
\node[blue] at (3.5,3.5) {$\scriptstyle \Delta$};
\draw[thick, blue, rounded corners=5pt] (-.5,2.5) -- (-.5,7.5) -- (7.5,7.5) -- (7.5,-.5) -- (-.5, -.5) -- (-.5,2.5);
\draw[rounded corners=5pt] (.5,2.5) -- (.5,6.5) --  (6.5,6.5) -- (6.5,.5) -- (.5,.5) -- (.5,2.5);
\draw (1.5,2) -- (-.5,2);
\draw (1.5,3) -- (-.5,3);
\draw (1.5,4) -- (-.5,4);
\draw (1.5, 5) -- (-.5, 5);
\draw (2,5.5) -- (2,7.5);
\draw (3,5.5) -- (3,7.5);
\draw (4,5.5) -- (4,7.5);
\draw (5,5.5) -- (5,7.5);
\draw (5.5,2) -- (7.5,2);
\draw (5.5,3) -- (7.5,3);
\draw (5.5,4) -- (7.5,4);
\draw (5.5, 5) -- (7.5, 5);
\draw (2,1.5) -- (2,-.5);
\draw (3,1.5) -- (3,-.5);
\draw (4,1.5) -- (4,-.5);
\draw (5,1.5) -- (5,-.5);
} 
\,\right\}.
$$
Observe that $p^A_\Delta\bigotimes_{v\in \Delta} \cH_v$ carries a $\operatorname{Tube}_\cC(\partial \Delta)$-action where we resolve the $\cC$-morphism from the annulus into the outer-most vertex spaces.
(This action is similar to \cite[Eq.~(10)]{2305.14068}.)
Since $\Lambda \ll \Delta$, 
every $p_\Delta^A x p_\Delta^A$ commutes with this $\operatorname{Tube}_\cC(\partial \Delta)$-action for $x\in \fA(\Lambda)$.

Now consider the evaluation map into the skein module 
$\eval: p^A_\Delta\bigotimes_{v\in \Delta}\cH_v\to \cS(\#\partial \Delta)$.
Observe that $\cS(\#\partial \Delta)$ also carries a $\operatorname{Tube}_\cC(\partial \Delta)$-action, and the evaluation map clearly intertwines the actions.
We also observe that $\cS(\#\partial \Delta)$ is an irreducible $\operatorname{Tube}_\cC(\partial \Delta)$-module, as all rank-one operators on $\cS(\#\partial \Delta)$ can be realized by operators in an annulus with no $\cC$-strings going around the annulus.
By Lemma \ref{lem:KongSkeinModule}, for any $x\in \fA(\Lambda)$,
$$
p_\Delta
x
p_\Delta
=
p^A_\Delta p^B_\Delta  p^A_\Delta x p^A_\Delta p^B_\Delta p^A_\Delta
=
D_\cC^{-2\#p\in \Delta}
p^A_\Delta
\eval^\dag 
\underbrace{
\eval
(p^A_\Delta x p^A_\Delta)
\eval^\dag 
}_{\in \End_{\operatorname{Tube}_\cC(\partial \Delta)}(\cS(\#\partial \Delta))\cong \bbC}
\eval
p^A_\Delta.
$$
This means that the compression of $x$ above lies in
$$
\bbC 
(\eval^\dag \eval)
\underset{\text{(Lem.~\ref{lem:KongSkeinModule})}}{=}
\bbC
p^B_\Delta
$$
acting on $p^A_\Delta \bigotimes_{v\in \Delta} \cH_v$.
Thus $p_\Delta \fA(\Lambda)p_\Delta = \bbC p_\Delta$.

\item[\underline{\ref{LTO:Boundary}:}]
Suppose $\Lambda \Subset_1 \Delta$ with $\partial \Lambda \cap \partial \Delta =I\neq \emptyset$.
Without loss of generality, we may assume $I$ is at the top or right of $\Lambda$.
We claim that $p_\Delta \fA(\Lambda)p_\Delta =\fF(I)p_\Delta$.
This will prove \ref{LTO:Boundary} as
$$
\fF(I)p_\Delta 
\underset{\text{(Lem.~\ref{lem:GammaInjective})}}{\subseteq}
\fB(\Lambda\Subset_1 \Delta)
\subseteq
p_\Delta \fA(\Lambda)p_\Delta
=
\fF(I)p_\Delta.
$$

It remains to prove $p_\Delta \fA(\Lambda)p_\Delta \subseteq \fF(I)p_\Delta$.
By Lemma \ref{lem:KongSkeinModule},
$p_\Delta \bigotimes_{v\in \Delta} \cH_v \cong \cS(\#\partial \Delta)$, which can be identified with
$$
\bigoplus_{\vec{c}_{\partial\Delta\setminus I},\vec{c}_{ I}}\cC(\vec{c}_{\partial\Delta\setminus I} \to \vec{c}_I)
\cong
\cC(X^{\#\partial\Delta\setminus I} \to X^{\# I})
$$
with the skein module inner product,
where 
$\vec{c}_{\partial\Delta\setminus I}= c_{i_1}\cdots c_{i_m}$
is a tensor product of simples in $\cC$ over the sites of $\partial\Delta\setminus I$, and
$\vec{c}_I = c_{j_1}\cdots c_{j_n}$ 
is a tensor product of simples over the sites of $I$.
Observe that this space is canonically an invertible 
$\End_\cC(X^{\# I}) - \End_\cC(X^{\# \partial\Delta\setminus I})$
bimodule, i.e., these algebras are each other's commutants by the Yoneda Lemma.

Now given $x\in \fA(\Lambda)$,
$p_\Delta xp_\Delta$ acts as the identity on all tensorands $\cH_v$ with $v\in \partial \Delta \setminus I$.
Thus the action of $p_\Delta ap_\Delta$ commutes with the right action of $\End_\cC(X^{\# \partial\Delta\setminus I})$ transported back to $p_\Delta \bigotimes_{v\in \Delta} \cH_v$ from the skein module.
We conclude that there is a $\varphi \in \End_\cC(X^{\# I})$ such that $p_\Delta xp_\Delta = \Gamma_{\varphi}p_\Delta$.

\item[\underline{\ref{LTO:Surjective}:}]
Suppose $\Lambda_1 \subset \Lambda_2 \Subset_1 \Delta$ with $\partial \Lambda_1 \cap \partial \Delta = \partial \Lambda_2
\cap \partial \Delta =I\neq \emptyset$.
Without loss of generality, we may assume $I$ is at the top or right of $\Lambda_i$ for $i=1,2$.
By the proof of \ref{LTO:Boundary} above, 
$$
\fB(\Lambda_1 \Subset_1 \Delta) 
= 
\fF(I)p_{\Delta}
=
\fB(\Lambda_2 \Subset_1 \Delta).
$$

\item[\underline{\ref{LTO:Injective}:}]
Suppose $\Lambda \Subset_1 \Delta_1\subset \Delta_2$ with $\partial \Lambda \cap \partial \Delta_1 = \partial \Lambda
\cap \partial \Delta_2 =I\neq \emptyset$.
Without loss of generality, we may assume $I$ is at the top or right of $\Lambda$.
We saw in the proof of \ref{LTO:Boundary} above that
$\fB(\Lambda \Subset_1 \Delta_1) = 
\fF(I)p_{\Delta_1}$.
By Lemma \ref{lem:GammaInjective} applied to $\Lambda \Subset_1 \Delta_2$, 
$\Gamma_\varphi p_{\Delta_2}=0$ implies $\varphi=0$, so $\Gamma_\varphi p_{\Delta_1}=0$.
\end{proof}

\begin{rem}
\label{rem:IdentifyBoundaryLW}
It follows immediately from the proof of \ref{LTO:Boundary} above that
$\fB(I)$ is isomorphic to $\fF(I)=\End_\cC(X^{\# I})$ depending on if we chose $\bbH$ to be on the left or bottom of the boundary $\cK$; otherwise, we would obtain $\fF(I)^{\op}=\End_{\cC^{\op}}(X^{\# I})$.
(Here, $\bbH$ and $\cK$ are chosen as in \S\ref{sec:BoundaryNet}.)
\end{rem}

\begin{rem}
\label{rem:FiniteDepthQuantumCircuit}
The most common operational definition of an equivalence between topologically ordered states is a finite depth quantum circuit which takes one state to the other \cite{PhysRevB.82.155138}. 
In particular, the articles \cite{PhysRevB.80.155136,0907.3724} mapped between Kitaev's quantum double model for a finite group $G$ \cite{MR1951039}
and the Levin-Wen string net model for $\Rep(G)$.
The recent article \cite{PhysRevB.105.085130} constructs a finite depth quantum circuit mapping between the Levin-Wen string net models for arbitrary Morita equivalent UFCs.

However, it is important to note that even though finite depth quantum circuits should preserve topological order of the bulk (as characterized by, say, topological string operators), they will not necessarily preserve the boundary algebra $\fB$.
This is because finite depth circuits will generally not intertwine the local ground state spaces, and thus will not naturally map boundary algebras to boundary algebras without explicitly requiring this.
Indeed, we now give an explicit example of two Levin-Wen string nets with the same topological order but non-isomorphic boundary quasi-local algebras. 

The boundary quasi-local algebra for the Levin-Wen model constructed from $\fdHilb(S_{3})$ has a Bratteli diagram with $6$ vertices at each level, 
and in the first level, each vertex corresponds to a copy of $\bbC$.
Going from one level to the next corresponds to alternately tensoring on the left and right by the direct sum of all simple objects, which yields a complete graph between each consecutive layer. 
The resulting AF-algebra is isomorphic to the UHF algebra $M_{6^{\infty}}$, the infinite tensor product of $M_6(\bbC)$. 
The pairing with the unique trace gives an order isomorphic from the $K_{0}$ group to the subgroup of $\bbR$ of `$6$-adic' rationals, namely the additive group of the ring $\bbZ[\frac{1}{6}]$. 
In particular, the pairing of the unique trace with $K_{0}$ is injective.

The boundary quasi-local algebra for the Levin-Wen model constructed from $\Rep(S_{3})$ has a Bratteli diagram with $3$ vertices at each level (corresponding to the isomorphism classes) of irreducible representations. 
Let $\rho$ denote the 2-dimensional irrep. 
Note that $\rho^{2}$ is the sum of all simples, each with multiplicity $1$. 
Therefore, the Bratelli diagram is given by tensoring on the left and right (alternatively) by $\rho^{2}$, which is isomorphic to just repeatedly tensoring  on the right with $\rho^{2}$ since the category is symmetric. 
The resulting diagram is just a coarse-graining of the AF-algebra constructed by taking tensor powers of the simple object $\rho$ itself. 
By \cite[Ex.~4.1]{MR4143013}, this AF-algebra has $K_{0}$ group $\bbZ[t]/\langle 1- t-2t^{2} \rangle$ with positive cone given by $\set{[p(t)]}{ p(\frac{1}{2})>0}\cup \{0\}$. 
Since the Bratteli diagram is connected and stationary, there is a unique trace on $K_{0}$, which pairs with $K$-theory by evaluating a class of polynomial at $\frac{1}{2}$ (which is well-defined since $\frac{1}{2}$ is a root of $1- t-2t^{2}$). 
This map is not injective on $K_{0}$ (or in other words, this group has infinitesimals), since $[2t-1]$ is in the kernel of this map, but it is non-zero on the ring since $1-t-2t^{2}$ cannot divide $2t-1$. 
In particular, this implies this AF-algebra cannot be UHF, hence not isomorphic to the boundary quasi-local algebra for $\fdHilb(S_{3})$.
\end{rem}

\section{Boundary states and applications to cone algebras}
\label{sec:BoundaryStates}

Suppose $(\fA,p)$ is a translation-invariant net of algebras and net of projections satisfying \ref{LTO:Hastings}--\ref{LTO:Injective}.
In this section, we study states on the boundary net $\fB$ for our examples.
The states necessarily extend to states on the half-plane algebra $\fA_\bbH$ which look like the ground state $\psi$ locally on the bulk of $\bbH$ away from the boundary.
We begin by analyzing the canonical boundary state from Example \ref{ex:BoundaryState} for our examples.
We then show it is a KMS state for Levin-Wen models, and we use it to study the \emph{cone algebras} in the case that $\fA$ is a quantum spin system, making connections to \cite{MR2804555, MR3426207, MR4362722}.
Finally, we study other canonical boundary states on Levin-Wen models associated to Q-systems.

\subsection{Toric Code boundary states}
\label{sec:ToricCodeBoundaryStates}

Here, we give three states on the boundary net of the Toric Code.

\begin{ex}
It is straightforward to compute the canonical state $\psi_\fB$ on the boundary net $\fB$, which we recall is isomorphic to the nets $\fC$ and $\fD$.
By Lemma \ref{lem:AbstractBoundaryAlgebra}, we see that each $\fC(I)$ is linearly spanned by monomials in the Pauli operators supported on $\widetilde{I}$, and similarly for $\fD(J)$.
Observe further that every such monomial in this canonical basis which is not the identity monomial anti-commutes with some star $A_s$ or plaquette operator in a large enough region $\Delta$ containing $\widetilde{I}$.
Since $\psi(x)p_\Delta = p_\Delta xp_\Delta$ for  $x\in \fC(I)$ and $\Delta$ sufficiently large, arguing as in Algorithm \ref{alg:PauliReduction}, we see that $\psi_\fB(x)$ is exactly the coefficient of the identity monomial.
We have not only determined $\psi_\fB$, but we have also shown it is a normalized \emph{trace} on $\fB$ satisfying $\psi_\fB(ab)=\psi_\fB(ba)$ for all $a,b\in \fB(I)$.
Under the isomorphism $\fB\cong\fF$ to the fusion categorical net for $\fdHilb(\bbZ/2)$ from Example \ref{ex:TensorCat} with $X=1\oplus g$,
$\psi_\fB$ corresponds to the unique \emph{Markov trace} \cite[\S2.7]{MR999799} on $\fF$ given by
$\tr_\fF(x) = d_X^{-n} \tr_\cC(\varphi)$ for $\varphi\in \End_\cC(X^n)$.
\end{ex}

In the example below, we discuss two other canonical states on the Toric Code boundary net $\fB$.
These states translate directly to canonical boundary states on fusion categorical nets discussed in \S\ref{sec:AdditionalBoundaryStates} below.

\begin{ex}
We use the diagrammatic description of the boundary net $\fB\cong \fD'$ from Remark \ref{rem:C(I)ParityPreservation} and Construction \ref{construction:ToricCodeInjective}.
When $\#I=n+1$, the $*$-algebra $\fB(I)\cong \fD'(I)$ is generated by the operators $\sigma^Z_i$ acting at site $i\in \{1,\dots, n{+}1\}$ and $\sigma^X_j \sigma^X_{j+1}$ acting at sites $j$ and $j+1$ for $j\in \{1,\dots, n\}$.
We get a pure state $\phi_Z^I=\langle \xi_Z^I |\,\cdot\, |\xi_Z^I\rangle$ on $\fD'(I)$ by choosing the product state vector $|\xi^I_Z\rangle=\bigotimes_{i=1}^{n+1}|0\rangle \in \bigotimes_{i=1}^{n+1} \bbC^2$, which lives in the $+1$ eigenspace for each $\sigma^Z_i$.
When $I\subset J$, $\fD'(I)$ includes into $\fD'(J)$ by tensoring with $I$, so $\phi^J_Z|_{\fD'(I)}=\phi^I_Z$.
We define $\phi_Z:=\varinjlim \phi^I_Z$.

We will now show that the $e$ particle condenses at this boundary, so it corresponds to the `rough boundary' in the sense of \cite{MR2942952}.  
Note that $e$ particles correspond to violations of $A_s$ terms in the Hamiltonian.
Thus two $e$ particles are created by applying a single $\sigma^Z$ operator.
Since $|\xi_Z^I \rangle$ is a $+1$ eigenvector for all $\sigma^Z$ operators, this state absorbs $\sigma^Z$, so it condenses the $e$ particle.  

Similarly, we can define an inductive limit state $\phi_X:=\varinjlim \phi^I_X$ where $\phi_X^I=\langle \xi_X^I |\,\cdot\, |\xi_X^I\rangle$
and
$|\xi^I_X\rangle=\bigotimes_{i=1}^{n+1}|+\rangle \in \bigotimes_{i=1}^{n+1} \bbC^2$.
By similar analysis, $\phi_X$ condenses the $m$ particle, so it corresponds to the `smooth boundary' in the sense of \cite{MR2942952}.
\end{ex}

\begin{rem}
Under the isomorphism $\fF\cong \fD'$ from Construction \ref{construction:ToricCodeInjective}, $\phi_Z$ on $\fD'$ corresponds to $\phi_{1_\cC}$ on $\fF$ from Example \ref{ex:UnitState} below, as $\phi_Z$ takes the coefficient of the empty diagram in $\fF(I)=\End_\cC(X^{\#I})$ because $\bigotimes^{n+1}_{i=1}|0\rangle$ is the first standard basis vector in $\bbC^{2^{n+1}}$.
On the other hand, $\phi_X$ corresponds to $\phi_{Q}$ for $Q=\bbC[\bbZ/2]$ with normalized multiplication on $\fF$ 
from Example \ref{ex:QSystemBoundary} below, as $\phi_X$ maps every diagrammatic basis element in $\fF(I)$ to $2^{-\#I}$ because
$\bigotimes^{n+1}_{i=1}|+\rangle$ is  $2^{-(n+1)/2}$ times the all $1$s vector in $\bbC^{2^{n+1}}$.
\end{rem}

\subsection{The canonical Levin-Wen string net boundary state}
\label{sec:LWBoundaryState}

We now consider the Levin-Wen string net model for the UFC $\cC$ discussed in \S\ref{sec:LWStringNet}.

The \emph{interactions} $\Phi$ of the Levin-Wen string-net system are given by a map from the finite subsets $\Lambda$ of our infinite square lattice to self-adjoint operators of $\fA$ such that $\Phi(\Lambda) \in \fA(\Lambda)$.
In our setting, the interactions are easily defined:
$\Phi(\ell) = 1_\fA -A_\ell$ for each edge $\ell$ 
and
$\Phi(p) = 1_\fA-B_p$ for a plaquette $p$,
and $\Phi(F)=0$ for all other finite subsets.
On $\bigotimes_{v\in \Lambda} \cH_v$, we define the local Hamiltonian
$$
H_\Lambda 
:= 
\sum_{F\subset \Lambda} \Phi(F) 
=
\sum_{\ell \subset \Lambda} (1_\fA -A_\ell)
+
\sum_{p \subset \Lambda} (1_\fA -B_p),
$$
which is clearly a commuting projector local Hamiltonian.

\begin{rem}
We can construct the canonical state $\psi$ from a net of state vectors $|\Omega_\Lambda\rangle$ on our local algebras $\fA(\Lambda)$.
On $\bigotimes_{v\in \Lambda} \cH_v$, we normalize the empty state vector $|\emptyset\rangle \in p^A_\Lambda \bigotimes_{v\in \Lambda} \cH_v$ after applying $p^B_\Lambda$ to get the state vector
$$
|\Omega_\Lambda\rangle := D_\cC^{(\#p\subset\Lambda)/2} p^B_\Lambda |\emptyset\rangle \in \bigotimes_{v\in \Lambda} \cH_v.
$$
We set $\omega_\Lambda:= \langle \Omega_\Lambda | \,\cdot\,|\Omega_\Lambda\rangle:\fA(\Lambda)\to \bbC$ to be the corresponding vector state.
Clearly $\omega_\Lambda(p_\Lambda)=1$,
so by Lemma \ref{lem:AFHCauchySchwarz},
$\omega_\Lambda(a) = \omega_\Lambda(p_\Lambda a) = \omega_\Lambda(ap_\Lambda)$ for all $a\in \fA(\Lambda)$.
Moreover, if $a\in \fA(\Lambda)$ and $\Lambda \ll \Delta$,
then
$$
\omega_\Delta(a)
\underset{\text{(Lem.~\ref{lem:AFHCauchySchwarz})}}{=}
\omega_\Delta(p_\Delta ap_\Delta)
\underset{\text{(Thm.~\ref{thm:LTQO})}}{=}
\psi(a)\cdot\omega_\Delta(p_\Delta)
=
\psi(a).
$$
\end{rem}

We now compute the canonical state $\psi_\fB$ from Example \ref{ex:BoundaryState} on the boundary algebra $\fB$ of the Levin-Wen model.
By Theorem \ref{thm:LTQO}, 
operators in $\fB(I)$ where $\partial\Lambda\cap \cK=\cI$ are products of the form $\Gamma_\varphi p_\Delta$ where $\varphi \in \End_\cC(X^{\#I})$ and $\Lambda \Subset \Delta$ with $\partial \Lambda \cap \partial \Delta = I$.
(Without loss of generality, we have assumed $I$ is at the top or right of $\Lambda$.)
Thus to compute $\psi_\fB$ on $\fB$, it suffices to calculate $\psi(\Gamma_\varphi)=\psi(\Gamma_\varphi p_\Delta)$ for $\varphi\in \End(X^{\#I})$ along some boundary interval $I$.
This will also give us a formula for the canonical boundary state transported to the fusion categorical net $\fF$ for $\cC$ from Example \ref{ex:TensorCat}.

Fix $\varphi \in \cC(a_1\otimes\cdots\otimes a_n \to b_1\otimes\cdots\otimes b_n)$, an interval $I$ along $\cK$ in $\bbZ^2$ with $\#I=n$, and a large square $\Delta$ surrounding $I$ such that $\partial \Delta \cap I = \emptyset$.
We compute $\textcolor{blue}{p_\Delta}\Gamma_\varphi \textcolor{red}{p_\Delta} \xi$ where $\xi$ is a simple tensor in $p^A_\Delta\bigotimes_{v\in \Delta} \cH_v$.
Here, we have colored the first-acting $\textcolor{red}{p_\Delta}$ red and the second acting $\textcolor{blue}{p_\Delta}$ blue, which will be reflected in the color of the strings used in the plaquette operators in the diagrammatic proof below.
We will further make the simplification that $n=3$, and it will be clear what the formula will be when $n$ is arbitrary.

First, starting with $\xi$, we apply $\textcolor{red}{p_\Delta}$, screening all punctures corresponding to the plaquettes.
In order to apply the gluing operator $\Gamma_\varphi$, we must first use \eqref{eq:Fusion} to resolve all  plaquette operators which act on the sites in $I$.
In the diagrams below, to ease the notation, we suppress all unnecessary sums over simples and scalars, only keeping track of labels, sums, and scalars for sites in $I$, which are marked in cyan in the diagram on the left.
$$
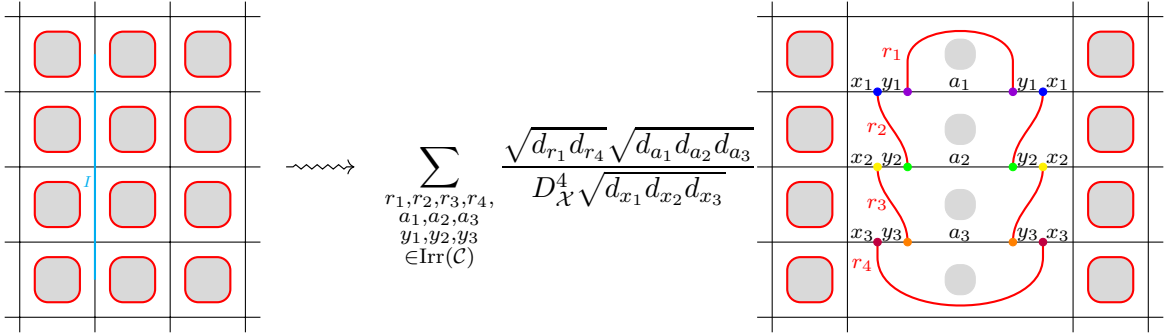
\tikzmath{
\draw (-.2,-.2) grid (3.2,4.2);
\foreach \x in {0,1,2}{
\foreach \y in {0,1,2,3}{
\filldraw[red, thick, fill=gray!30, rounded corners=5pt] ($ (\x,\y) + (.2,.2) $) rectangle ($ (\x,\y) + (.8,.8) $);
}}
\draw[thick, cyan] (1,.5) -- (1,3.5);
\node[cyan, scale=.75] at (.9,1.8) {$\scriptstyle I$};
}
\begin{tikzcd}
\mbox{}
\arrow[r,squiggly]
&
\mbox{}
\end{tikzcd}
\sum_{\substack{
r_1,r_2,r_3,r_4,
\\
a_1,a_2,a_3
\\
y_1,y_2,y_3 
\\
\in \Irr(\cC)
}}
\frac{\sqrt{d_{r_1}d_{r_4}}\sqrt{d_{a_1}d_{a_2}d_{a_3}}}{D_\cX^{4}\sqrt{d_{x_1}d_{x_2}d_{x_3}}}
\tikzmath{
\draw (-.2,4) -- (5.2,4);
\foreach \x in {0,1,4,5}{
\draw (\x,-.2) -- (\x,4.2);
}
\foreach \y in {0,1,2,3}{
\draw (-.2,\y) -- (5.2,\y);
\filldraw[gray!30, rounded corners=5pt] (2.3,\y+.3) rectangle (2.7,\y+.7) ;
\foreach \x in {0,4}{
\filldraw[red, thick, fill=gray!30, rounded corners=5pt] ($ (\x,\y) + (.2,.2) $) rectangle ($ (\x,\y) + (.8,.8) $);
}
}
\foreach \y in {1,2,3}{
\node at ($ (1.2,4-\y) + (0,.1) $) {$\scriptstyle x_{\y}$};
\node at ($ (1.6,4-\y) + (0,.1) $) {$\scriptstyle y_{\y}$};
\node at ($ (2.5,4-\y) + (0,.1) $) {$\scriptstyle a_{\y}$};
\node at ($ (3.4,4-\y) + (0,.1) $) {$\scriptstyle y_{\y}$};
\node at ($ (3.8,4-\y) + (0,.1) $) {$\scriptstyle x_{\y}$};
}
\draw[red, thick] (1.4,1)  -- (1.4,.8) to[out=-90,in=-90] (3.6,.8) -- (3.6,1);
\draw[red, thick] (1.8,3)  -- (1.8,3.4) to[out=90,in=90] (3.2,3.4) -- (3.2,3);
\draw[thick, red] (1.4,3) to[out=-90,in=90] (1.8,2);
\draw[thick, red] (1.4,2) to[out=-90,in=90] (1.8,1);
\draw[thick, red] (3.6,2) to[out=-90,in=90] (3.2,1);
\draw[thick, red] (3.6,3) to[out=-90,in=90] (3.2,2);
\node[red] at (1.6,3.5) {$\scriptstyle r_1$};
\node[red] at (1.4,2.5) {$\scriptstyle r_2$};
\node[red] at (1.4,1.5) {$\scriptstyle r_3$};
\node[red] at (1.2,.7) {$\scriptstyle r_4$};
\filldraw[purple] (1.4,1) circle (.05cm);
\filldraw[purple] (3.6,1) circle (.05cm);
\filldraw[orange] (1.8,1) circle (.05cm);
\filldraw[orange] (3.2,1) circle (.05cm);
\filldraw[yellow] (1.4,2) circle (.05cm);
\filldraw[yellow] (3.6,2) circle (.05cm);
\filldraw[green] (1.8,2) circle (.05cm);
\filldraw[green] (3.2,2) circle (.05cm);
\filldraw[blue] (1.4,3) circle (.05cm);
\filldraw[blue] (3.6,3) circle (.05cm);
\filldraw[violet] (1.8,3) circle (.05cm);
\filldraw[violet] (3.2,3) circle (.05cm);
}
$$
We now apply $\Gamma_\varphi$ and then $\textcolor{blue}{p_\Delta}$ to obtain
$$
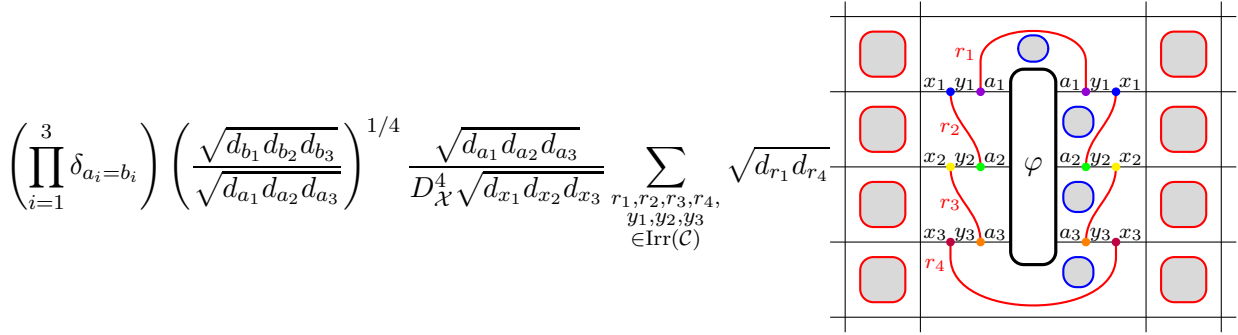
\left(
\prod_{i=1}^3
\delta_{a_i=b_i}
\right)
\left(
\frac{\sqrt{d_{b_1}d_{b_2}d_{b_3}}}{\sqrt{d_{a_1}d_{a_2}d_{a_3}}}
\right)^{1/4}
\frac{\sqrt{d_{a_1}d_{a_2}d_{a_3}}}{D_\cX^{4}\sqrt{d_{x_1}d_{x_2}d_{x_3}}}
\sum_{\substack{
r_1,r_2,r_3,r_4,
\\
y_1,y_2,y_3 
\\
\in \Irr(\cC)
}}
\sqrt{d_{r_1}d_{r_4}}
\tikzmath{
\draw (-.2,4) -- (5.2,4);
\foreach \x in {0,1,4,5}{
\draw (\x,-.2) -- (\x,4.2);
}
\foreach \y in {0,1,2,3}{
\draw (-.2,\y) -- (5.2,\y);
\foreach \x in {0,4}{
\filldraw[red, thick, fill=gray!30, rounded corners=5pt] ($ (\x,\y) + (.2,.2) $) rectangle ($ (\x,\y) + (.8,.8) $);
}
}
\foreach \y in {1,2,3}{
\node at ($ (1.2,4-\y) + (0,.1) $) {$\scriptstyle x_{\y}$};
\node at ($ (1.6,4-\y) + (0,.1) $) {$\scriptstyle y_{\y}$};
\node at ($ (2,4-\y) + (0,.1) $) {$\scriptstyle a_{\y}$};
\node at ($ (3,4-\y) + (0,.1) $) {$\scriptstyle a_{\y}$};
\node at ($ (3.4,4-\y) + (0,.1) $) {$\scriptstyle y_{\y}$};
\node at ($ (3.8,4-\y) + (0,.1) $) {$\scriptstyle x_{\y}$};
}
\draw[red, thick] (1.4,1)  -- (1.4,.8) to[out=-90,in=-90] (3.6,.8) -- (3.6,1);
\draw[red, thick] (1.8,3)  -- (1.8,3.4) to[out=90,in=90] (3.2,3.4) -- (3.2,3);
\draw[thick, red] (1.4,3) to[out=-90,in=90] (1.8,2);
\draw[thick, red] (1.4,2) to[out=-90,in=90] (1.8,1);
\draw[thick, red] (3.6,2) to[out=-90,in=90] (3.2,1);
\draw[thick, red] (3.6,3) to[out=-90,in=90] (3.2,2);
\filldraw[fill=white, rounded corners=5pt, very thick] (2.2,.7) rectangle (2.8,3.3);
\node at (2.5,2) {$\varphi$};
\filldraw[thick, blue, fill=gray!30, rounded corners=5pt] (2.9,.4) rectangle (3.3,.8);
\filldraw[thick, blue, fill=gray!30, rounded corners=5pt] (2.9,1.4) rectangle (3.3,1.8);
\filldraw[thick, blue, fill=gray!30, rounded corners=5pt] (2.9,2.4) rectangle (3.3,2.8);
\filldraw[thick, blue, fill=gray!30, rounded corners=5pt] (2.3,3.4) rectangle (2.7,3.75);
\node[red] at (1.6,3.5) {$\scriptstyle r_1$};
\node[red] at (1.4,2.5) {$\scriptstyle r_2$};
\node[red] at (1.4,1.5) {$\scriptstyle r_3$};
\node[red] at (1.2,.7) {$\scriptstyle r_4$};
\filldraw[purple] (1.4,1) circle (.05cm);
\filldraw[purple] (3.6,1) circle (.05cm);
\filldraw[orange] (1.8,1) circle (.05cm);
\filldraw[orange] (3.2,1) circle (.05cm);
\filldraw[yellow] (1.4,2) circle (.05cm);
\filldraw[yellow] (3.6,2) circle (.05cm);
\filldraw[green] (1.8,2) circle (.05cm);
\filldraw[green] (3.2,2) circle (.05cm);
\filldraw[blue] (1.4,3) circle (.05cm);
\filldraw[blue] (3.6,3) circle (.05cm);
\filldraw[violet] (1.8,3) circle (.05cm);
\filldraw[violet] (3.2,3) circle (.05cm);
}
$$
where we have used that $B_p^2 = B_p$ for all un-resolved plaquette operators.
We now assume $a_i=b_i$ for all $i$, and we 
use the screening property of the blue $\textcolor{blue}{B_p}$ operators with respect to the resolved red $\textcolor{red}{B_p}$ operators to obtain 
$$
\frac{\sqrt{d_{a_1}d_{a_2}d_{a_3}}}{D_\cX^{4}\sqrt{d_{x_1}d_{x_2}d_{x_3}}}
\sum_{\substack{
r_1,r_2,r_3,r_4,
\\
y_1,y_2,y_3 
\\
\in \Irr(\cC)
}}
\sqrt{d_{r_1}d_{r_4}}
\tikzmath{
\draw (-.2,4) -- (5.2,4);
\foreach \x in {0,1,4,5}{
\draw (\x,-.2) -- (\x,4.2);
}
\foreach \y in {0,1,2,3}{
\draw (-.2,\y) -- (5.2,\y);
\foreach \x in {0,4}{
\filldraw[red, thick, fill=gray!30, rounded corners=5pt] ($ (\x,\y) + (.2,.2) $) rectangle ($ (\x,\y) + (.8,.8) $);
}
}
\foreach \y in {1,2,3}{
\node at ($ (1.2,4-\y) + (0,.1) $) {$\scriptstyle x_{\y}$};
\node at ($ (1.6,4-\y) + (0,.1) $) {$\scriptstyle y_{\y}$};
\node at ($ (2,4-\y) + (0,.1) $) {$\scriptstyle a_{\y}$};
\node at ($ (3,4-\y) + (0,.1) $) {$\scriptstyle a_{\y}$};
\node at ($ (3.4,4-\y) + (0,.1) $) {$\scriptstyle y_{\y}$};
\node at ($ (3.8,4-\y) + (0,.1) $) {$\scriptstyle x_{\y}$};
}
\draw[red, thick] (1.4,1)  -- (1.4,.8) to[out=-90,in=-90] (3.6,.8) -- (3.6,1);
\draw[red, thick] (1.8,3)  -- (1.8,3.4) to[out=90,in=90] (3.2,3.4) -- (3.2,3);
\draw[thick, red] (1.4,3) to[out=-90,in=90] (1.8,2);
\draw[thick, red] (1.4,2) to[out=-90,in=90] (1.8,1);
\draw[thick, red] (3.6,2) to[out=-90,in=90] (3.2,1);
\draw[thick, red] (3.6,3) to[out=-90,in=90] (3.2,2);
\filldraw[fill=white, rounded corners=5pt, very thick] (2.2,.7) rectangle (2.8,3.3);
\node at (2.5,2) {$\varphi$};
\filldraw[thick, blue, fill=gray!30, rounded corners=5pt] (3.5,.1) rectangle (3.9,.5);
\filldraw[thick, blue, fill=gray!30, rounded corners=5pt] (3.5,1.25) rectangle (3.9,1.65);
\filldraw[thick, blue, fill=gray!30, rounded corners=5pt] (3.5,2.25) rectangle (3.9,2.65);
\filldraw[thick, blue, fill=gray!30, rounded corners=5pt] (3.5,3.3) rectangle (3.9,3.7);
\node[red] at (1.6,3.5) {$\scriptstyle r_1$};
\node[red] at (1.4,2.5) {$\scriptstyle r_2$};
\node[red] at (1.4,1.5) {$\scriptstyle r_3$};
\node[red] at (1.2,.7) {$\scriptstyle r_4$};
\filldraw[purple] (1.4,1) circle (.05cm);
\filldraw[purple] (3.6,1) circle (.05cm);
\filldraw[orange] (1.8,1) circle (.05cm);
\filldraw[orange] (3.2,1) circle (.05cm);
\filldraw[yellow] (1.4,2) circle (.05cm);
\filldraw[yellow] (3.6,2) circle (.05cm);
\filldraw[green] (1.8,2) circle (.05cm);
\filldraw[green] (3.2,2) circle (.05cm);
\filldraw[blue] (1.4,3) circle (.05cm);
\filldraw[blue] (3.6,3) circle (.05cm);
\filldraw[violet] (1.8,3) circle (.05cm);
\filldraw[violet] (3.2,3) circle (.05cm);
}
$$

We can thus reduce to the analysis of the sub-diagram
$$
\frac{\sqrt{d_{a_1}d_{a_2}d_{a_3}}}{D_\cX^{4}\sqrt{d_{x_1}d_{x_2}d_{x_3}}}
\sum_{\substack{
r_1,r_2,r_3,r_4,
\\
y_1,y_2,y_3 
\\
\in \Irr(\cC)
}}
\sqrt{d_{r_1}d_{r_4}}
\tikzmath{
\foreach \y in {1,2,3}{
\draw (.8,\y) -- (4.2,\y);
\node at ($ (1.2,4-\y) + (0,.1) $) {$\scriptstyle x_{\y}$};
\node at ($ (1.6,4-\y) + (0,.1) $) {$\scriptstyle y_{\y}$};
\node at ($ (2,4-\y) + (0,.1) $) {$\scriptstyle a_{\y}$};
\node at ($ (3,4-\y) + (0,.1) $) {$\scriptstyle a_{\y}$};
\node at ($ (3.4,4-\y) + (0,.1) $) {$\scriptstyle y_{\y}$};
\node at ($ (3.8,4-\y) + (0,.1) $) {$\scriptstyle x_{\y}$};
}
\draw[red, thick] (1.4,1)  -- (1.4,.8) to[out=-90,in=-90] (3.6,.8) -- (3.6,1);
\draw[red, thick] (1.8,3)  -- (1.8,3.4) to[out=90,in=90] (3.2,3.4) -- (3.2,3);
\draw[thick, red] (1.4,3) to[out=-90,in=90] (1.8,2);
\draw[thick, red] (1.4,2) to[out=-90,in=90] (1.8,1);
\draw[thick, red] (3.6,2) to[out=-90,in=90] (3.2,1);
\draw[thick, red] (3.6,3) to[out=-90,in=90] (3.2,2);
\filldraw[fill=white, rounded corners=5pt, very thick] (2.2,.7) rectangle (2.8,3.3);
\node at (2.5,2) {$\varphi$};
\node[red] at (1.6,3.5) {$\scriptstyle r_1$};
\node[red] at (1.4,2.5) {$\scriptstyle r_2$};
\node[red] at (1.4,1.5) {$\scriptstyle r_3$};
\node[red] at (1.2,.7) {$\scriptstyle r_4$};
\filldraw[purple] (1.4,1) circle (.05cm);
\filldraw[purple] (3.6,1) circle (.05cm);
\filldraw[orange] (1.8,1) circle (.05cm);
\filldraw[orange] (3.2,1) circle (.05cm);
\filldraw[yellow] (1.4,2) circle (.05cm);
\filldraw[yellow] (3.6,2) circle (.05cm);
\filldraw[green] (1.8,2) circle (.05cm);
\filldraw[green] (3.2,2) circle (.05cm);
\filldraw[blue] (1.4,3) circle (.05cm);
\filldraw[blue] (3.6,3) circle (.05cm);
\filldraw[violet] (1.8,3) circle (.05cm);
\filldraw[violet] (3.2,3) circle (.05cm);
}
=
\psi(\Gamma_\varphi)
\cdot
\tikzmath{
\foreach \y in {1,2,3}{
\draw (1,\y) -- (3,\y);
\node at ($ (2,4-\y) + (0,.1) $) {$\scriptstyle x_{\y}$};
}
}\,.
$$
Above, we claim that the morphism on the left is equal to a scalar times the identity $\id_{x_1x_2x_3}$; this scalar will necessarily be equal to $\psi(\Gamma_\varphi)$.
First, we use the fusion relation \eqref{eq:Fusion} to contract along the $r_1$ string to obtain
$$
d_{a_1}
\frac{\sqrt{d_{a_2}d_{a_3}}}{D_\cX^{4}\sqrt{d_{x_1}d_{x_2}d_{x_3}}}
\sum_{\substack{
r_2,r_3,r_4,
\\
y_1,y_2,y_3 
\\
\in \Irr(\cC)
}}
\sqrt{d_{y_1}d_{r_4}}
\tikzmath{
\draw (.8,3) --node[above,yshift=-.1cm]{$\scriptstyle x_1$} (1.4,3) to[out=0,in=180] (2.1,3.8) --node[above,yshift=-.1cm]{$\scriptstyle y_1$} (2.9,3.8) to[out=0,in=180] (3.6,3) --node[above,yshift=-.1cm]{$\scriptstyle x_1$} (4.2,3);
\draw (2.2,3) node[left,yshift=-.1cm]{$\scriptstyle a_1$} arc (270:90:.3cm) -- (2.8,3.6) arc (90:-90:.3cm) node[right,yshift=-.1cm]{$\scriptstyle a_1$};
\foreach \y in {1,2}{
\draw (.8,\y) -- (4.2,\y);
}
\foreach \y in {2,3}{
\node at ($ (1.2,4-\y) + (0,.1) $) {$\scriptstyle x_{\y}$};
\node at ($ (1.6,4-\y) + (0,.1) $) {$\scriptstyle y_{\y}$};
\node at ($ (2,4-\y) + (0,.1) $) {$\scriptstyle a_{\y}$};
\node at ($ (3,4-\y) + (0,.1) $) {$\scriptstyle a_{\y}$};
\node at ($ (3.4,4-\y) + (0,.1) $) {$\scriptstyle y_{\y}$};
\node at ($ (3.8,4-\y) + (0,.1) $) {$\scriptstyle x_{\y}$};
}
\draw[red, thick] (1.4,1)  -- (1.4,.8) to[out=-90,in=-90] (3.6,.8) -- (3.6,1);
\draw[thick, red] (1.4,3) to[out=-90,in=90] (1.8,2);
\draw[thick, red] (1.4,2) to[out=-90,in=90] (1.8,1);
\draw[thick, red] (3.6,2) to[out=-90,in=90] (3.2,1);
\draw[thick, red] (3.6,3) to[out=-90,in=90] (3.2,2);
\filldraw[fill=white, rounded corners=5pt, very thick] (2.2,.7) rectangle (2.8,3.3);
\node at (2.5,2) {$\varphi$};
\node[red] at (1.4,2.5) {$\scriptstyle r_2$};
\node[red] at (1.4,1.5) {$\scriptstyle r_3$};
\node[red] at (1.2,.7) {$\scriptstyle r_4$};
\filldraw[purple] (1.4,1) circle (.05cm);
\filldraw[purple] (3.6,1) circle (.05cm);
\filldraw[orange] (1.8,1) circle (.05cm);
\filldraw[orange] (3.2,1) circle (.05cm);
\filldraw[yellow] (1.4,2) circle (.05cm);
\filldraw[yellow] (3.6,2) circle (.05cm);
\filldraw[green] (1.8,2) circle (.05cm);
\filldraw[green] (3.2,2) circle (.05cm);
\filldraw[blue] (1.4,3) circle (.05cm);
\filldraw[blue] (3.6,3) circle (.05cm);
}\,.
$$
Next, we use the fusion relation \eqref{eq:Fusion} to contract along the $y_1$ string to obtain
$$
d_{a_1}
\frac{\sqrt{d_{a_2}d_{a_3}}}{D_\cX^{4}\sqrt{d_{x_2}d_{x_3}}}
\sum_{\substack{
r_2,r_3,r_4,
\\
y_2,y_3 
\\
\in \Irr(\cC)
}}
\sqrt{d_{r_2}d_{r_4}}
\tikzmath{
\draw (.8,3) -- (1,3) to[out=0,in=180] (1.7,4) --node[above,yshift=-.1cm]{$\scriptstyle x_1$} (3.3,4) to[out=0,in=180] (4,3) -- (4.2,3);
\draw (2.2,3) node[left,yshift=-.1cm]{$\scriptstyle a_1$} arc (270:90:.3cm) -- (2.8,3.6) arc (90:-90:.3cm) node[right,yshift=-.1cm]{$\scriptstyle a_1$};
\foreach \y in {1,2}{
\draw (.8,\y) -- (4.2,\y);
}
\foreach \y in {2,3}{
\node at ($ (1.2,4-\y) + (0,.1) $) {$\scriptstyle x_{\y}$};
\node at ($ (1.6,4-\y) + (0,.1) $) {$\scriptstyle y_{\y}$};
\node at ($ (2,4-\y) + (0,.1) $) {$\scriptstyle a_{\y}$};
\node at ($ (3,4-\y) + (0,.1) $) {$\scriptstyle a_{\y}$};
\node at ($ (3.4,4-\y) + (0,.1) $) {$\scriptstyle y_{\y}$};
\node at ($ (3.8,4-\y) + (0,.1) $) {$\scriptstyle x_{\y}$};
}
\draw[red, thick] (1.4,1)  -- (1.4,.8) to[out=-90,in=-90] (3.6,.8) -- (3.6,1);
\draw[thick, red] (1.4,2) to[out=-90,in=90] (1.8,1);
\draw[thick, red] (1.8,2) to[out=90,in=-90] (1.6,3.3) arc (180:90:.5cm) -- (2.9,3.8) arc (90:0:.5cm) to[out=-90,in=90] (3.2,2);
\draw[thick, red] (3.6,2) to[out=-90,in=90] (3.2,1);
\filldraw[fill=white, rounded corners=5pt, very thick] (2.2,.7) rectangle (2.8,3.3);
\node at (2.5,2) {$\varphi$};
\node[red] at (1.5,2.5) {$\scriptstyle r_2$};
\node[red] at (1.4,1.5) {$\scriptstyle r_3$};
\node[red] at (1.2,.7) {$\scriptstyle r_4$};
\filldraw[purple] (1.4,1) circle (.05cm);
\filldraw[purple] (3.6,1) circle (.05cm);
\filldraw[orange] (1.8,1) circle (.05cm);
\filldraw[orange] (3.2,1) circle (.05cm);
\filldraw[yellow] (1.4,2) circle (.05cm);
\filldraw[yellow] (3.6,2) circle (.05cm);
\filldraw[green] (1.8,2) circle (.05cm);
\filldraw[green] (3.2,2) circle (.05cm);
}
$$
We then use the fusion relation \eqref{eq:Fusion} to contract along the $r_2$ string to obtain
$$
d_{a_1}d_{a_2}
\frac{\sqrt{d_{a_3}}}{D_\cX^{4}\sqrt{d_{x_2}d_{x_3}}}
\sum_{\substack{
r_3,r_4,
\\
y_2,y_3 
\\
\in \Irr(\cC)
}}
\sqrt{d_{y_2}d_{r_4}}
\tikzmath{
\draw (.8,3) to[out=0,in=180] (1.7,4.2) --node[above,yshift=-.1cm]{$\scriptstyle x_1$} (3.3,4.2) to[out=0,in=180] (4.2,3);
\draw (2.2,3) node[left,yshift=-.1cm]{$\scriptstyle a_1$} arc (270:90:.3cm) -- (2.8,3.6) arc (90:-90:.3cm) node[right,yshift=-.1cm]{$\scriptstyle a_1$};
\draw (.8,2) --node[above,yshift=-.1cm]{$\scriptstyle x_2$} (1.2,2) to[out=0,in=-90] node[left]{$\scriptstyle y_2$} (1.3,3) to[out=90,in=180] (1.9,4) -- (3.1,4) to[out=0,in=90] (3.7,3) to[out=-90,in=180] (3.8,2) --node[above,yshift=-.1cm]{$\scriptstyle x_2$} (4.2,2);
\draw (2.2,2) node[left,yshift=-.1cm]{$\scriptstyle a_2$} to[out=180,in=-90] (1.6,3.3) arc (180:90:.5cm) -- (2.9,3.8) arc (90:0:.5cm) to[out=-90,in=0] (2.8,2) node[right,yshift=-.1cm]{$\scriptstyle a_2$};
\draw (.8,1) -- (4.2,1);
\node at (1.2,1.1) {$\scriptstyle x_3$};
\node at (1.6,1.1) {$\scriptstyle y_3$};
\node at (2,1.1) {$\scriptstyle a_3$};
\node at (3,1.1) {$\scriptstyle a_3$};
\node at (3.4,1.1) {$\scriptstyle y_3$};
\node at (3.8,1.1) {$\scriptstyle x_3$};
\draw[red, thick] (1.4,1)  -- (1.4,.8) to[out=-90,in=-90] (3.6,.8) -- (3.6,1);
\draw[thick, red] (1.2,2) to[out=-90,in=90] (1.8,1);
\draw[thick, red] (3.8,2) to[out=-90,in=90] (3.2,1);
\filldraw[fill=white, rounded corners=5pt, very thick] (2.2,.7) rectangle (2.8,3.3);
\node at (2.5,2) {$\varphi$};
\node[red] at (1.3,1.5) {$\scriptstyle r_3$};
\node[red] at (1.2,.7) {$\scriptstyle r_4$};
\filldraw[purple] (1.4,1) circle (.05cm);
\filldraw[purple] (3.6,1) circle (.05cm);
\filldraw[orange] (1.8,1) circle (.05cm);
\filldraw[orange] (3.2,1) circle (.05cm);
\filldraw[yellow] (1.2,2) circle (.05cm);
\filldraw[yellow] (3.8,2) circle (.05cm);
}\,.
$$
At this point, it is clear that we can then contract along the $y_2$ string using \eqref{eq:Fusion}, followed by the $r_3$ string afterward to obtain
\begin{align*}
\frac{d_{a_1}d_{a_2}d_{a_3}}{D_\cC^{4}\sqrt{d_{x_3}}}
\tr_\cC(\varphi)
\sum_{\substack{
r_4,y_3 
\\
\in \Irr(\cC)
}}
\sqrt{d_{y_3}d_{r_4}}
\cdot
\tikzmath{
\draw (1,1) --node[above,yshift=-.1cm]{$\scriptstyle x_3$} (1.5,1) --node[above,yshift=-.1cm]{$\scriptstyle y_3$} (2.5,1) --node[above,yshift=-.1cm]{$\scriptstyle x_3$} (3,1);
\draw (1,1.5) --node[above,yshift=-.1cm]{$\scriptstyle x_2$} (3,1.5);
\draw (1,2) --node[above,yshift=-.1cm]{$\scriptstyle x_1$} (3,2);
\draw[red, thick] (1.5,1) to[out=-90,in=-90] (2.5,1);
\node[red] at (1.4,.7) {$\scriptstyle r_4$};
\filldraw[purple] (1.5,1) circle (.05cm);
\filldraw[purple] (2.5,1) circle (.05cm);
}
&=
\frac{d_{a_1}d_{a_2}d_{a_3}}{D_\cC^{4}}
\tr_\cC(\varphi)
\sum_{\substack{
r_4
\in \Irr(\cC)
}}
d_{r_4}
\cdot
\tikzmath{
\draw (1,1) --node[above,yshift=-.1cm]{$\scriptstyle x_3$} (3,1); 
\draw (1,1.5) --node[above,yshift=-.1cm]{$\scriptstyle x_2$} (3,1.5);
\draw (1,2) --node[above,yshift=-.1cm]{$\scriptstyle x_1$} (3,2);
\draw[red, thick] (2,.5) circle (.3cm);
\node[red] at (1.4,.5) {$\scriptstyle r_4$};
}
\\&=
\frac{d_{a_1}d_{a_2}d_{a_3}}{D_\cC^{3}}
\tr_\cC(\varphi)
\cdot
\tikzmath{
\draw (1,1) --node[above,yshift=-.1cm]{$\scriptstyle x_3$} (3,1); 
\draw (1,1.5) --node[above,yshift=-.1cm]{$\scriptstyle x_2$} (3,1.5);
\draw (1,2) --node[above,yshift=-.1cm]{$\scriptstyle x_1$} (3,2);
}\,.
\end{align*}

We have just proven the following proposition.

\begin{prop}
\label{prop:CanonicalState}
The canonical state $\psi_\fF$ on the boundary net of algebras $\fF$ is given by
$$
\psi_\fF(\Gamma_\varphi)=
\frac{1}{D_\cX^n}
\sum_{\substack{
c_1,\dots, c_n
\\\in \Irr(\cC)
}}
d_{c_1}\cdots d_{c_n} \tr_\cC(\varphi \cdot p_{c_1\otimes\cdots\otimes c_n})
\qquad\qquad\qquad
\varphi\in \End_\cC(X^n)
$$
where $p_{c_1\otimes\cdots\otimes c_n}\in \End_\cC(X^n)$ is the orthogonal projection onto $c_1 \otimes \cdots \otimes c_n \subset X^n$.
\end{prop}

It is clear that the above formula gives a well-defined inductive limit state $\psi_\fF$ on $\fF$.

\begin{rem}
\label{rem:MarkovTrace}
Surprisingly, $\psi$ is a trace on $\fB\cong \fF$ if and only if $\cC$ is pointed, i.e., $d_c=1$ for all $c\in \Irr(\cC)$.
In this case, $\fF$ has a unique Markov trace, namely $\tr :=\varinjlim d_X^{-n}\tr_\cC: \varinjlim\End_\cC(X^n)\to \bbC$, satisfying
$$
\tr\left(
\tikzmath{
\draw (-.3,.3) -- (-.3,.7);
\draw (.3,.3) -- (.3,.7);
\draw (.5,.3) -- (.5,.7);
\node at (.05,.5) {$\cdots$};
\draw (-.3,-.3) -- (-.3,-1.2);
\draw (.3,-.3) -- (.3,-1.2);
\node at (.05,-.5) {$\cdots$};
\draw (.5,-.3) arc (-180:0:.3cm) -- (1.1,.7);
\draw (.5,-1.2) arc (180:0:.3cm);
\roundNbox{}{(0,0)}{.3}{.2}{.4}{$\varphi$}
}\,\,\,
\right)
=
\tr(\varphi)
=
\tr\left(\,\,\,
\tikzmath{
\draw (-.3,.3) -- (-.3,.7);
\draw (.3,.3) -- (.3,.7);
\draw (-.5,.3) -- (-.5,.7);
\node at (.05,.5) {$\cdots$};
\draw (-.3,-.3) -- (-.3,-1.2);
\draw (.3,-.3) -- (.3,-1.2);
\node at (.05,-.5) {$\cdots$};
\draw (-.5,-.3) arc (0:-180:.3cm) -- (-1.1,.7);
\draw (-.5,-1.2) arc (0:180:.3cm);
\roundNbox{}{(0,0)}{.3}{.4}{.2}{$\varphi$}
}\,\,\,
\right)
\qquad\qquad\qquad
\forall\, \varphi\in \End_\cC(X^n),
$$
as the centers $Z(\fF(I))$ have dimension globally bounded by $\#\Irr(\cC)$ \cite{MR999799}.
Thus $\psi_\fF$ is the unique Markov trace, and $\fF''$ in the GNS representation of $\psi_\fF$ is a $\rm II_1$ factor.
Above, since $X$ is self-dual, 
$$
\tikzmath{
\draw (-.3,.5) arc (0:-180:.3cm);
\draw (-.3,-.5) arc (0:180:.3cm);
}
=
\coev_{X}\circ \coev_X^\dag 
= 
\ev_X^\dag \circ \ev_X
$$
where $(X,\ev_X,\coev_X)$ is any standard solution of the conjugate equations for $X$.
(Observe $$\bigoplus_{a,b,c,d} \cC(ab\to 1)\otimes \cC(1\to cd)$$ is of the form $\cK\otimes \cK^*$ for the Hilbert space $\cK=\cC(XX\to 1)$, and summing over an ONB and its adjoint is independent of the choice of ONB.)
\end{rem}

\subsection{The KMS condition for the Levin-Wen string net boundary state}
\label{sec:KMS}

We now show that the canonical state $\psi_\fF$ on the fusion categorical net $\fF$ is the unique KMS $\beta=1$ state for a dynamics coming from a certain unbounded operator.

First, we use $\psi_\fF$ to complete $\fF$ to a Hilbert space on which $\fF$ acts by left translation by bounded operators.
The operator $H_\partial=-\sum_{v\in\bbZ} C_v$  with 
$C_v := \sum_{c\in \Irr(\cC)}\log(d_c)\id_c \in \End_\cC(X)$
acting locally at site $v$
is an unbounded operator acting in $L^2(\fF,\psi_\fF)$, containing each $\varphi \in \fF(I)=\End_\cC(X^n)$ in its domain for each bounded interval $I$ (we write $n=\#I$ here). 
Observe that $\exp(-itC_v) = \bigoplus_{c\in \Irr(\cC)} d_c^{-it}\id_c$ acting at site $v$,
and thus $\exp(itH_\partial) = \prod_{v} \exp(-itC_v)$ acts locally at each site $v$.
For $t\in \bbR$ and 
$\varphi\in \cC(a_1\otimes\cdots \otimes a_n \to b_1\otimes\cdots\otimes b_n)\subset \fF(I)$,
we have
\begin{equation}
\label{eq:ModularAutomorphism}
\sigma_t(\varphi)
:=
\exp(it H_\partial)\varphi \exp(-itH_\partial) 
= 
\prod_{u,v\in I} \exp(-itC_u) \varphi \exp(itC_v)
=
\left(
\frac
{d_{a_1}^{it}\cdots d_{a_n}^{it}}
{d_{b_1}^{it}\cdots d_{b_n}^{it}}
\right)
\varphi.
\end{equation}
Clearly on such $\varphi$, $t\mapsto \sigma_t(\varphi)$ can be analytically continued to an entire function, and we observe
$$
\sigma_{it}(\varphi)
=
\left(
\frac
{d_{a_1}^{-t}\cdots d_{a_n}^{-t}}
{d_{b_1}^{-t}\cdots d_{b_n}^{-t}}
\right)
\varphi
=
\left(
\frac
{d_{b_1}^{t}\cdots d_{b_n}^{t}}
{d_{a_1}^{t}\cdots d_{a_n}^{t}}
\right)
\varphi.
$$

Recall that $\psi_\fF$ is a KMS-$\beta$ state for $(\fF,\sigma)$ and $\beta>0$ if for all $x,y\in \fF$ with $y$ entire ($t\mapsto \sigma_t(y)$ extends to an entire function), 
$\psi_\fF(x \sigma_{i\beta}(y))= \psi_\fF(yx)$.
When $\varphi:a_1\otimes\cdots\otimes a_n \to b_1\otimes\cdots\otimes b_n$ and $\phi: b_1\otimes\cdots\otimes b_n \to a_1\otimes\cdots\otimes a_n$ are in $\fF(I)$, 
we calculate
\begin{align*}
\psi_\fF(\varphi \sigma_{i\beta}(\phi))
&=
\left(
\frac
{d_{a_1}^{\beta}\cdots d_{a_n}^{\beta}}
{d_{b_1}^{\beta}\cdots d_{b_n}^{\beta}}
\right)
\psi_\fF(\varphi\phi)
=
\frac{1}{\cD_X^n}
\left(
\frac
{d_{a_1}^{\beta}\cdots d_{a_n}^{\beta}}
{d_{b_1}^{\beta}\cdots d_{b_n}^{\beta}}
\right)
(d_{b_1}\cdots d_{b_n})
\tr_\cC(\varphi\phi)
\\&\underset{(\beta=1)}{=}
\frac{1}{\cD_X^n}
(d_{a_1}\cdots d_{a_n})
\tr_\cC(\phi\varphi)
=
\psi_\fF(\phi\varphi).
\end{align*}
We have thus shown $\psi_\fF$ is KMS-1 for $(\fF,\sigma)$.

Moreover, $\psi_\fF$ is the unique KMS-1 state on $\fF$ by \cite[Prop.~4.1]{MR1772604} and uniqueness of the Frobenius-Perron eigenvector (up to scaling), as the Bratteli diagram
of $\fF=\varinjlim \fF([-n,n])$ is connected and stationary.
Hence letting $\fF''$ be the von Neumann algebra generated by $\fF$ in the GNS representation $L^2(\fF,\psi_\fF)$, we have that $\fF''$ is a factor.
Recall that $\psi_\fF$ and $\sigma$ both extend to $\fF''$, and this extension is still a KMS-1 state~\cite[Cor. 5.3.4]{MR1441540}, so the modular automorphism group is given by $t\mapsto \sigma_t$.
We write $(\psi_\fF, \sigma)$ for this extension again.

\begin{lem}
\label{lem:CanonicalStateFaithful}
The canoncial state $\psi_\fF$ is faithful and normal on $\fF''$.
Since $\fF$ is simple, $\psi_\fF$ is faithful on $\fF$.
\end{lem}
\begin{proof}
By \cite[Cor.~5.3.9]{MR1441540}, the canonical cyclic vector $\Omega_{\psi_\fF}$ in the GNS representation $L^2(\fF,\psi_\fF)$ is separating for $\fF''$, and thus $\psi_\fF = \langle \Omega_{\psi_\fF}| \,\cdot\,|\Omega_{\psi_\fF}\rangle$ is normal and faithful.
The last claim is then immediate.
\end{proof}

Our next task is to prove the following theorem.

\begin{thm}
\label{thm:InnerIffPointed}
If there is a $c\in \Irr(\cC)$ with $d_c\neq 1$,
then for all but countably many $t\in \bbR$, $\sigma_t$ from \eqref{eq:ModularAutomorphism} is outer. In particular, $\fF''$ is a type $\rm III$ factor.
\end{thm}

To prove this theorem, we make the following definition.

\begin{defn}
Given a von Neumann algebra $M$ and a faithful state $\phi$, a \emph{$\phi$-central sequence} is a norm bounded sequence $(x_n)\subset M$ such that $\|yx_n-x_ny\|_\phi\to 0$ for all $y\in M^\phi$, where $M^\phi$ is the centralizer of $\phi$.
\end{defn}

\begin{lem} 
\label{lem:PhiCentralInner}
Let $\alpha$ be a $\phi$-preserving automorphism of $M$ and $(x_n)$ a $\phi$-central sequence. 
If $\alpha$ is inner, then $||\alpha(x_{n})-x_{n}||_{\phi}\rightarrow 0$.
\end{lem}
\begin{proof}
    If $\alpha=\Ad(u)$ is $\phi$-preserving, then $u\in M^{\phi}$, and $\phi(yu)=\phi(uy)$ for all $y\in M$. 
    Thus right multiplication by $u$ (and $u^{*}$) is a $|| \cdot ||_{\phi}$-isometry. 
    We conclude
\[
    0=\lim ||ux_{n}-x_{n}u||_{\phi}=\lim ||ux_{n}u^{*}-x||_{\phi}=\lim ||\alpha(x_{n})-x_{n}||_{\phi}.
\qedhere    
\]    
\end{proof}

We also need an observation about centralizers.
Recall that there is a canonical $\phi$-preserving conditional expectation $E: M\to M^\phi$ defined as follows.
First, consider the canonical injection $\iota: L^2(M^\phi,\phi)\hookrightarrow L^2(M,\phi)$.
Then for $x\in M$, $E(x):=\iota^* x \iota \in B(L^2(M^\phi, \phi))$ and commutes with the bounded right $M^\phi$-action, and thus defines an element of $M^\phi$. 

Now suppose $M=\overline{\bigcup M_n}^{SOT}$ is a hyperfinite von Neumann algebra with each $M_n$ finite dimensional.
We then get a conditional expectation $E_n: M_n\to M_n^\phi$ as above.

\begin{lem}
\label{lem:ApproximateCentralizer}
Let $t \mapsto \sigma_t^\phi$ be the modular automorphism group of $M$ for the state $\phi$.
Suppose $\sigma_t^\phi$ preserves $M_n$ for all $n$ and all $t\in\bbR$.
Then $E|_{M_n}=E_n$.
In particular, $M^\phi = \overline{\bigcup M_n^\phi}^{SOT}$.
\end{lem}
\begin{proof}
Since each $\sigma^\phi_t$ preserves each $M_n$, $L^2(M_n,\phi)$ is an invariant subspace for each unitary $\Delta^{it}$ on $L^2(M,\phi)$, where $\Delta$ is the modular operator.
This means that $\Delta^{it}$ commutes with the orthogonal projection $p_n$ onto each $L^2(M_n, \phi)$ for all $n$, so the von Neumann algebra generated by all $\Delta^{it}$ and $p_n$ is abelian.  
Hence for each $t \in \bbR$, the spectral projection $e_t$ onto the $\lambda=1$ eigenspace for $\Delta^{it}$ commutes with each $p_n$.  
The orthogonal projection $e$ onto the intersection of these eigenspaces is given by $e = \prod_{t \in \bbR} e_t$, where the limit is taken in SOT.  
Hence $e$ is in the von Neumann algebra generated by the $\Delta^{it}$ and $p_n$, so it commutes with each $p_n$.  
Since each $e_t$ projects onto the invariant subspace for $\Delta^{it}$, it follows that $e$ projects onto $L^2(M^\phi, \phi)$ embedded in $L^2(M, \phi)$, so $E(x) \Omega = e x \Omega$ for $x \in M$.
Thus for $x\in M_n$,
\[
E(x)\Omega = e x \Omega = ep_nx\Omega = p_n ex\Omega = p_nE(x)\Omega,
\]
and thus $E(x)\Omega \in L^2(M_n,\phi)$.
We conclude that $E_n(x)=E(x)$.

To prove the final claim, suppose $x\in M^\phi$ and $x_n\in M_n$ with $x_n\to x$ SOT.
Then $E_n(x_n) = E(x_n)\to E(x)=x$ SOT. 
Since $E_n(x_n)\in M_n^\phi$, we are finished.
\end{proof}

\begin{proof}[Proof of Theorem \ref{thm:InnerIffPointed}]
Suppose there is a $c\in \Irr(\cC)$ with $d_c\neq 1$.
Let $t\in \bbR$ such that $d_c^{2it}\neq 1$.
We define a $\psi_\fF$-central sequence using the map $\ev_c : \overline{c}\otimes c\to 1_\cC$.
To begin, $x_0 = \ev_c$ localized at sites $0,1$ in our $\bbZ$-lattice $\cK$.
We then set $x_n:= \tau_n(\ev_c)$, where $\tau_n$ is translation by $n\in\bbN$, i.e.:
$$
x_n
=
\tikzmath{
\node at (-.9,0) {$\cdots$};
\draw (-.1,-.4) --node[left]{$\scriptstyle X$} (-.1,.4)  node[above]{$\scriptstyle n{-}1$};
\draw (.3,-.4) node[below]{$\scriptstyle n$} --node[left]{$\scriptstyle \overline{c}$} (.3,-.2) arc (180:0:.3cm) --node[right]{$\scriptstyle c$} (.9,-.4) node[below]{$\scriptstyle n{+}1$};
\draw (1.3,-.4) --node[right]{$\scriptstyle X$} (1.3,.4) node[above]{$\scriptstyle n{+}2$};
\node at (2.1,0) {$\cdots$};
}
$$
Note that $\sigma_i(x_n) = d_c^{-2} x_n$ for all $n$.
We claim that $(x_n)$ is $\psi_\fF$-central.
Indeed, for $y\in (\fF'')^{\psi_\fF}$, 
\begin{align}
\|yx_n\|_{\psi_\fF}^2
&=
\psi_\fF(x_n^*y^*yx_n) 
= 
\psi_\fF(yx_n\sigma_{i}(x_n^*)y^*)
=
d_c^{-2}\psi_\fF(yx_nx_n^*y^*)
\notag
\\&\leq
d_c^{-2}\|x_n\|^2\psi_\fF(y^*y)
=
d_c^{-2}\|x_n\|^2 \|y\|_{\psi_\fF}^2.
\label{eq:RightMultBounded}
\end{align}
Let $\varepsilon>0$, and choose a finite interval $I\subset \bbZ$ and $z\in \fF(I)^{\psi_\fF}$ such that 
$\|y-z\|_{\psi_\fF}<\varepsilon$ (which exists as $\bigcup \fF(I)$ is SOT-dense in $\fF''$ together with Lemma \ref{lem:ApproximateCentralizer}).
Pick $N\in\bbN$ large so that $n>N$ implies $x_nz-zx_n=0$.
We now calculate that when $n>N$,
\begin{align*}
\|yx_n - x_n y\|_{\psi_\fF}
&=
\|yx_n -zx_n + x_nz - x_n y \|_{\psi_\fF} 
\\&\leq
\|(y-z)x_n\|_{\psi_\fF} + \|x_n(z-y)\|_{\psi_\fF} 
\\&\leq 
(d_c^{-2}+1)\|x_n\| \|y-z\|_{\psi_\fF} 
&&\text{by \eqref{eq:RightMultBounded}}
\\&\leq
(d_c^{-2}+1)C \cdot \varepsilon
\end{align*}
where $C=\|x_n\|$ is independent of $n$ by translation invariance of $\psi_\fF$.
We conclude that $(x_n)$ is $\psi_\fF$-central.
However, 
$$
\|\sigma_t(x_n)-x_n\|_{\psi_\fF}
=
\|(d_c^{2it}-1)x_n\|_{\psi_\fF}
=
|d_c^{2it}-1| \cdot \|x_n\|_{\psi_\fF}
=
|d_c^{2it}-1| \cdot K
$$
where $K = \|x_n\|_{\psi_\fF}$ is independent of $n$.
By Lemma \ref{lem:PhiCentralInner}, $\sigma_t$ is not inner.

The last claim follows immediately by \cite[Thm.~VIII.3.14]{MR1943006}.
\end{proof}

\begin{rem}
\label{rem:LWBoundaryType}
Recall that type~\rm{III} factors can be further classified by a parameter $\lambda \in [0,1]$.
In  Appendix~\ref{app:typeiii} authored by Masaki Izumi,
he shows there are two cases, based on the fusion rules for $\cC$.
If the set
\[
    \set{ 
    \frac{d_a d_b}{d_c}
    }{
    a,b,c \in \Irr(\cC) \textrm{ and } N_{ab}^c \geq 1 
    }
\]
generates $\bbR_{>0}$ as a closed subgroup, then $\fF''$ is of type $\rm{III}_1$.
We call a triple $(a,b,c)$ such that $N_{ab}^c \geq 1$ \emph{admissable}.
If the generated subgroup is not dense, there is a $0 < \lambda < 1$ such that for all admissable triples $(a,b,c)$, we have
\[
    \frac{d_a d_b}{d_c} = \lambda^{Z^c_{ab}}
    \qquad\qquad\qquad\text{for some }Z^c_{ab}\in \bbZ,
\]
and the integers $Z^c_{ab}$ together generate $\bbZ$.
In this case, $\fF''$ is of type $\rm{III}_\lambda$.
\end{rem}

\begin{ex}
Let $\cC$ be the Fibonacci category with simple objects $1,\tau$ satisfying $\tau\otimes \tau \cong 1\oplus \tau$ and $d_\tau = \phi$, the golden ratio.
Since
$$
d_\tau^{it}=\phi^{it}=\exp(\log(\phi^{it}))=\exp(it\log(\phi)),
$$
whenever $t\in \frac{2\pi}{\log(\phi)}\bbZ$, $d_\tau^{it}=1$.
Since $1^{it}$ is always 1, this says that $\sigma_t$ is inner for these $t$.
We conclude that $\fF''$ is not type $\rm III_1$ \cite[Thm.~3.4.1]{MR0341115}.
In fact, using the results in Appendix~\ref{app:typeiii}, we can determine the type.
Note that the only non-trivial admissible triples are $(\tau, \tau, 1)$ and $(\tau, \tau, \tau)$, leading to the ratios $d_\tau^2$ and $d_\tau$.
Since these generate the non-zero part of $S(\fF'')$, it follows that $\lambda = d_\tau^{-1}$ and $\fF''$ is of type ${\rm{III}}_{d_\tau^{-1}}$.


The above example can be modified to any near-group or Tambara-Yamagami UFC.
\end{ex}

\begin{rem}
The classification program for topologically ordered phases of matter is about \emph{gapped} phases.
That is, one considers equivalence classes of (local) Hamiltonians which have a spectral gap in the thermodynamic limit.
This means that in the GNS repersentation, the (unbounded) Hamiltonian implementing the dynamics in the GNS representation of the ground state has a spectral gap between its lowest eigenvalue and the rest of the spectrum.
Our framework does not make reference to the spectral gap directly (although in our examples, the nets of projections come from gapped Hamiltonians), but a natural question to ask is if the canonical state associated to $(\fA, p)$ can be realized as the ground state of a gapped Hamiltonian.
Because we want the dynamics to be local, this is a non-trival question.

In the \emph{projected entangled pair state} (PEPS) setting, a particular class of tensor network states, it turns out that there is a relation between the spectral gap of the 2D bulk state, and properties of the 1D boundary state.%
\footnote{Many 2D states with topological order have a PEPS representation (including the states we consider here), but it is still an open problem if \emph{every} gapped quantum phase contains at least one PEPS representative~\cite{MR3995422}.}
It is known that every PEPS is the ground state of a local Hamiltonian, called the \emph{parent Hamiltonian}, but showing if this is gapped is generally very difficult.
However, important progress has been made.
For a 2D PEPS, there is a canonical way to define a 1D boundary state, and a corresponding boundary Hamiltonian~\cite{PhysRevB.83.245134}.
Based on numerical evidence, it was conjectured that the bulk Hamiltonian is gapped if and only if the boundary Hamiltonian is short-ranged.
Later, it was shown that if the boundary state is ``approximately factorizable'', then the bulk state indeed is gapped, and that 1D thermal states of finite range Hamiltonians satisfy this property~\cite{MR3614053,MR3927082}.

In the example above, we have shown that our boundary state $\psi_\fF$ is a thermal state for a local Hamiltonian.
Moreover, the bulk state is the Levin-Wen model, which has a spectral gap.
Hence this is consistent with the result in the PEPS setting, even if our framework does not make reference to Hamiltonians.
On the other hand, it is known that for the boundary states of \emph{topologically ordered} PEPS (such as the $G$-injective PEPS for $G$ a finite group), the local density operators do not have full rank, in contrast to KMS states.
One can, however, project down to the support of the local density operators and define a boundary Hamiltonian on this subspace (cf.~\S 5.6 of~\cite{MR3927082}).
But this is essentially what happens already on the algebra level in the construction of the boundary algebras.\footnote{As we remark at the end of \S\ref{sec:BoundaryNet}, we do not know if the boundary state is faithful in general, but it is for all our examples.}

Hence it is an interesting question how our results relate to the PEPS setting, also because we do not assume that the canonical state on the bulk has a PEPS representation.
For example, one could ask if our canonical boundary state coincides with that defined in~\cite{PhysRevB.83.245134}. We hope to return to this question at a later point.
\end{rem}

\subsection{Cone algebras}

In this section, we specialize to the case when our net $(\fA,p)$ is a quantum spin system as in Example \ref{ex:SpinSystem} with a translation-invariant net of projections.
We now discuss the connection between the canonical state on the boundary algebra $\fB$ and the analysis of the \emph{cone algebras} $\fA(\Lambda)''$ of the Toric Code and Kitaev quantum double models from \cite{MR4721705}.
A \emph{cone} $\Lambda \subset \bbZ^2$ is a region of the form
$$
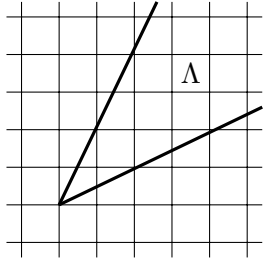
\tikzmath{
\draw[step=.5] (-1.7,-1.7) grid (1.7,1.7);
\draw[very thick] (-1,-1) -- (.3,1.7);
\draw[very thick] (-1,-1) -- (1.7,.3);
\node at (.75,.75) {$\Lambda$};
}
$$
and the algebra $\fA(\Lambda)$ is the $\rm C^*$-subalgebra of $\fA$ supported on sites in $\Lambda$.
The \emph{cone algebra} is then $\fA(\Lambda)''$, where the von Neumann completion is taken in the GNS representation on $L^2(\fA, \psi)$.  
These algebras are of interest as these cone regions are used to describe the excitations of topologically ordered spin systems using superselection theory \cite{MR2804555, MR3426207,MR4362722}.
In particular, the excitations for Toric Code \cite{MR2804555} and the abelian quantum double model \cite{MR3426207} have been described by localized and transportable endomorphisms of the quasi-local algebra $\fA$ (superselection sectors), where the localizing regions used are precisely these cones.  
Since the intertwining morphisms between two sectors live in the cone algebras, and since these endomorphisms can be uniquely extended to the cone algebras in a WOT-continuous fashion, the cone algebras become of interest when studying these models.  

When we have a net of projections $(p_\Lambda)$ satisfying \ref{LTO:Hastings}, we can use the canonical state $\psi$ to employ the argument in \cite{MR2804555} to show the cone algebras are factors.

\begin{prop}
\label{prop:ConeAlgebrasFactors}
Let $\Lambda$ be any subset of the lattice. 
Then $\fA(\Lambda)''$ in the GNS representation on $L^2(\fA,\psi)$ is a factor.
When $\Lambda$ is a cone, $\fA(\Lambda)''$ is infinite.
\end{prop}
\begin{proof}
This is a standard argument, but we repeat it here for convenience.
First, since $\psi : \fA \to \bbC$ is a pure state, we have that 
\[
\fA(\Lambda)'' \vee \fA(\Lambda^c)''
=
\fA''
=
B(L^2(\fA,\psi)).
\]
This immediately implies
\[
\cZ(\fA(\Lambda)'')'
=
(\fA(\Lambda)'' \cap \fA(\Lambda)')'
\supseteq
\fA(\Lambda)' \vee \fA(\Lambda)''
\supseteq
\fA(\Lambda^c)'' \vee \fA(\Lambda)''
=
B(L^2(\fA,\psi)),
\]
and thus $\fA(\Lambda)''$ is a factor.

When $\Lambda$ is a cone, the argument from \cite[Thm.~5.1]{MR2804555} shows that $\fA(\Lambda)''$ is infinite.
\end{proof}

In \cite{MR4721705}, Ogata proves that the cone algebras $\fA(\Lambda)''\subset B(L^2(\fA,\psi))$ with rough edges for the Kitaev quantum double model are type $\rm II_\infty$ factors by essentially\footnote{Really, Ogata proves this for a region $\Lambda'$ differing from the original cone by a finite number of edges.
She does this by intersecting $\Lambda'$ by larger and larger rectangles and showing that when you cut down by the support projection for the state $\psi$ on these finite regions, the obtained state is a trace.} 
showing there is a projection $p_\Lambda\in \fA(\Lambda)''$ such that the ground state restricted to the corner $p_\Lambda \fA(\Lambda)'' p_\Lambda$ is a trace.
Here, we use the suggestive notation $p_\Lambda$, as the chosen projection in \cite{MR4721705} is essentially $\prod_{s\subset \Lambda} A_s \prod_{p\subset \Lambda}B_p$, which exists as an infimum of projections in $\fA(\Lambda)''$.
(In Kitaev's quantum double model, $A_s,B_p$ are orthogonal projections, in contrast to $A_s,B_p$ from Kitaev's Toric Code model.)

Kitaev's quantum double model is simultaneously a model for $\fdHilb(G)$ and $\Rep(G)$; really, it is a model for the quantum double $\cD(G)$, which is the center of both $\fdHilb(G)$ and $\Rep(G)$.
(See also Remark \ref{rem:FiniteDepthQuantumCircuit}.)
For the $\cC=\fdHilb(G)$ Levin-Wen model, the compression of the cone algebra $p_\Lambda \fA(\Lambda)''p_\Lambda$ is almost\footnote{We get the boundary algebra on the nose if the sites in $\Lambda$ are connected by edges and plaquettes contained entirely in $\Lambda$. 
However, if sites are disconnected, we will get an amplification of a boundary algebra by a certain finite dimensional algebra.
} 
exactly the
von Neumann algebra $\fF''$ in the GNS representation of the canonical state $\psi_\fF$.
Indeed, consider $\Lambda$ to be the third quadrant
(including the sites on the axes).
$$
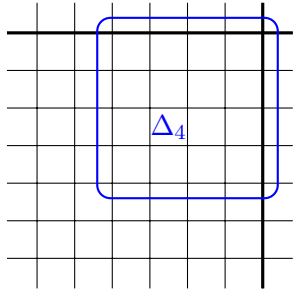
\tikzmath{
\draw[step=.5] (.4,.4) grid (-3.4,-3.4);
\draw[very thick] (-3.4,0) -- (.4,0);
\draw[very thick] (0,-3.4) -- (0,.4);
\draw[thick, blue, rounded corners = 5pt] (.2,.2) rectangle (-2.2,-2.2);
\node[blue] at (-1.25,-1.25) {$\Delta_4$};
}
$$
Let $\Delta_n$ be the $n^2$-rectangle in $\Lambda$ with northeast corner at the origin.
Given a local operator $x\in \fA(\Delta_n)\subset \fA(\Lambda)$, by an argument similar to the proof of \ref{LTO:Boundary} for Theorem \ref{thm:LTQO}, there is a unique $\varphi \in \End_\cC(X^{2n})$ such that for every $k>n$, $p_{\Delta_{k}} x p_{\Delta_k} = \Gamma_\varphi p_{\Delta_k}$, where $\Gamma_\varphi$ now glues $\varphi$ onto the northeast boundary of $\Delta_n$.

Observe now that
$\varinjlim p_{\Delta_{n+1}} \fA(\Delta_n)p_{\Delta_{n+1}} \cong \fF$;
on this algebra, the ground state is exactly the canonical state $\psi_\fF$ from Proposition \ref{prop:CanonicalState} above.
Since 
$\fA(\Lambda)=\varinjlim \fA(\Delta_n)$
and
$p_{\Delta_{k}}\to p_\Lambda$ SOT,
we conclude that $p_\Lambda \fA(\Lambda)''p_\Lambda$ is exactly the von Neumann algebra $\fF''$ in the GNS representation of the canonical state $\psi_\fF$.
In the case $\cC=\fdHilb(G)$, which is analogous to Kitaev's quantum double model, $\psi_\fF = \tr$, the unique Markov trace as discussed in Remark \ref{rem:MarkovTrace}.

The recent article \cite{MR4814524} computes the boundary algebras for Kitaev's quantum double model and indeed verifies that the canonical state is again a trace.  
This gives more direct confirmation of the result in \cite{MR4721705} using the argument we just outlined above.  

However, we remark that when $\cC$ is not pointed, the cone algebras are no longer type $\rm II_\infty$, but rather type $\rm III$!

\begin{cor}[Cor.~\ref{cor:ConeAlgebraTypes}]
The cone algebra $\fA(\Lambda)''$ is type $\rm II_\infty$ if $\cC$ is pointed; otherwise $\fA(\Lambda)''$ is type $\rm III$.
\end{cor}
\begin{proof}
By Proposition \ref{prop:ConeAlgebrasFactors}, $\fA(\Lambda)''$ is an infinite factor.
Since $\fF'' \cong p_\Lambda \fA(\Lambda)''p_\Lambda$ is either a type $\rm II_1$ factor when $\cC$ is pointed by Remark \ref{rem:MarkovTrace}
or a type $\rm III$ factor when $\cC$ is not pointed by Theorem \ref{thm:InnerIffPointed}, the result follows. 
\end{proof}

\subsection{Additional boundary states on the Levin-Wen boundary net}
\label{sec:AdditionalBoundaryStates}
In this section, we define some additional states on the Levin-Wen boundary net of algebras $\fF$ in terms of the UFC AF approximation $\End_\cC(X^n)$ where $X=\bigoplus_{c\in\Irr(\cC)} c$.
In future work, we will analyze the superselection theory for the boundary nets for these states.

\begin{ex}
\label{ex:UnitState}
Consider the inclusion isometry $\iota: 1_\cC\to X$ and its adjoint $\iota^\dag$.
We denote $\iota,\iota^\dag$ by a univalent vertex on the $X$-string.
We define a state $\phi_1$ on $\fF$ by 
$$
\phi_{1}(\varphi)
:=
\tikzmath{
\draw (-.3,-.7) -- (-.3,.7);
\filldraw (-.3,-.7) circle (.05cm);
\filldraw (-.3,.7) circle (.05cm);
\draw (.3,-.7) -- (.3,.7);
\filldraw (.3,-.7) circle (.05cm);
\filldraw (.3,.7) circle (.05cm);
\roundNbox{fill=white}{(0,0)}{.3}{.2}{.2}{$\varphi$}
\node at (.05,-.5) {$\cdots$};
\node at (.05,.5) {$\cdots$};
}
\qquad\qquad\qquad\qquad
\varphi\in \End_\cC(X^n).
$$
Observe that $\phi_1(\varphi) = \phi_1(\varphi \otimes \id_X) = \phi_1(\id_X\otimes \varphi)$ as $\iota^\dag\circ \iota = \id_{1}$, so we get a well-defined inductive limit state.
\end{ex}

The previous example can be generalized substantially.
Recall that that a \emph{Q-system} in $\cC$ is an algebra object $(Q,m,i)\in \cC$ satisfying the following axioms:
$$
\underset{\text{(Associative)}}{
\tikzmath{
\draw[thick, orange] (0,0) -- (.6,.6);
\draw[thick, orange] (.4,0) -- (.2,.2);
\draw[thick, orange] (.8,0) -- (.4,.4);
\filldraw[orange] (.2,.2) circle (.05cm);
\filldraw[orange] (.4,.4) circle (.05cm);
}
=
\tikzmath{
\draw[thick, orange] (0,0) -- (-.6,.6);
\draw[thick, orange] (-.4,0) -- (-.2,.2);
\draw[thick, orange] (-.8,0) -- (-.4,.4);
\filldraw[orange] (-.2,.2) circle (.05cm);
\filldraw[orange] (-.4,.4) circle (.05cm);
}
}
\quad\qquad
\underset{\text{(Unital)}}{
\tikzmath{
\draw[thick, orange] (0,0) -- (0,.6);
\draw[thick, orange] (.2,.2) -- (0,.4);
\filldraw[orange] (0,.4) circle (.05cm);
\filldraw[orange] (.2,.2) circle (.05cm);
}
=
\tikzmath{
\draw[thick, orange] (0,0) -- (0,.6);
}
=
\tikzmath{
\draw[thick, orange] (0,0) -- (0,.6);
\draw[thick, orange] (-.2,.2) -- (0,.4);
\filldraw[orange] (0,.4) circle (.05cm);
\filldraw[orange] (-.2,.2) circle (.05cm);
}
}
\quad\qquad
\underset{\text{(Frobenius)}}{
\tikzmath{
\draw[thick, orange] (0,0) -- (.6,.6);
\draw[thick, orange] (0,.6) -- (.2,.2);
\draw[thick, orange] (.6,0) -- (.4,.4);
\filldraw[orange] (.2,.2) circle (.05cm);
\filldraw[orange] (.4,.4) circle (.05cm);
}
=
\tikzmath{
\draw[thick, orange] (0,0) -- (.2,.2) -- (.4,0);
\draw[thick, orange] (0,.6) -- (.2,.4) -- (.4,.6);
\draw[thick, orange] (.2,.2) -- (.2,.4);
\filldraw[orange] (.2,.2) circle (.05cm);
\filldraw[orange] (.2,.4) circle (.05cm);
}
=
\tikzmath{
\draw[thick, orange] (0,0) -- (-.6,.6);
\draw[thick, orange] (0,.6) -- (-.2,.2);
\draw[thick, orange] (-.6,0) -- (-.4,.4);
\filldraw[orange] (-.2,.2) circle (.05cm);
\filldraw[orange] (-.4,.4) circle (.05cm);
}
}
\quad\qquad
\underset{\text{(Separable)}}{
\tikzmath{
\draw[thick, orange] (0,0) -- (0,.15);
\draw[thick, orange] (0,.45) -- (0,.6);
\draw[thick, orange] (0,.3) circle (.15cm);
\filldraw[orange] (0,.15) circle (.05cm);
\filldraw[orange] (0,.45) circle (.05cm);
}
=
\tikzmath{
\draw[thick, orange] (0,0) -- (0,.6);
}
}\,.
$$
Here, we denote $Q$ by an orange strand, the multiplication $m$ by a trivialent vertex, and the unit $i$ by a univalent vertex.
We denote adjoints by vertical reflections.
In the example below, we use a \emph{standard} Q-system which satisfies that $i^\dag\circ m \in \cC(Q\otimes Q\to 1)$ and $m^\dag\circ i \in \cC(1\to Q\otimes Q)$ is a standard (minimal) solution to the conjugate equations \cite{MR1444286}.
We refer the reader to \cite{MR4419534} for the basics of Q-systems in unitary tensor categories.

\begin{ex}
\label{ex:QSystemBoundary}
Pick a standard Q-system $Q\in \cC$, and for each $c\in \Irr(\cC)$, choose an ONB  $\{\alpha_c\}\subset \cC(c\to Q)$ using the isometry inner product, i.e., $\alpha_c^\dag \circ \alpha_c' = \delta_{\alpha_c = \alpha_c'}$ and $\sum_{\alpha_c} \alpha_c\circ \alpha_c^\dag$ is the projection onto the isotypic component of $c$ in $Q$.
Note that $\sum_{c\in \Irr(\cC)}\sum_{\alpha_c} \alpha_c\circ \alpha_c^\dag = \id_Q$ and $\alpha_c^\dag \circ \alpha_d' = 0$ when $c\neq d$.

We define $\phi_Q$ on $\fF$ by
$$
\phi_Q(\varphi)
:=
\frac{1}{d_Q}
\sum_{\substack{c_1,\dots, c_n\\
\in \Irr(\cC)}}
\sum_{\alpha_{c_1},\dots, \alpha_{c_n}}
\tikzmath{
\filldraw[orange] (0,1.1) circle (.05cm);
\filldraw[orange] (0,-1.1) circle (.05cm);
\draw[orange, thick] (0,1.1) -- (-.3,.7);
\draw[orange, thick] (0,-1.1) -- (-.3,-.7);
\draw[orange, thick] (0,1.1) -- (.3,.7);
\draw[orange, thick] (0,-1.1) -- (.3,-.7);
\draw[thick, orange] (0,1.1) arc (180:0:.5cm) -- (1,-1.1) arc (0:-180:.5cm);
\draw (-.3,-.7) -- (-.3,.7);
\filldraw (-.3,-.7) node[left]{$\scriptstyle\alpha_{c_1}^\dag$} circle (.05cm);
\filldraw (-.3,.7) node[left]{$\scriptstyle\alpha_{c_1}$} circle (.05cm);
\draw (.3,-.7) -- (.3,.7);
\filldraw (.3,-.7) node[right]{$\scriptstyle\alpha_{c_n}^\dag$} circle (.05cm);
\filldraw (.3,.7) node[right]{$\scriptstyle\alpha_{c_n}$} circle (.05cm);
\roundNbox{fill=white}{(0,0)}{.3}{.2}{.2}{$\varphi$}
\node at (.05,-.5) {$\cdots$};
\node at (.05,.5) {$\cdots$};
}
\qquad\qquad\qquad\qquad
\varphi\in \End_\cC(X^n),
\quad
\tikzmath{
\draw (0,-.3) -- (0,.3);}
=
X,
\quad
\tikzmath{
\draw[thick, orange] (0,-.3) -- (0,.3);}
=
Q.
$$
Here, we write a single multi-valent orange vertex to denote the product of $n$ copies of $Q$, and the orange cup and cap are the standard solution to the conjugate equations built from $i,m$ and their adjoints.
One checks that $\phi_Q(\varphi) = \phi_Q(\varphi \otimes \id_X) = \phi_Q(\id_X\otimes \varphi)$
using 
associativity and separability of $Q$ and sphericality of $\cC$.
We thus get a well-defined inductive limit state.

In \cite{MR2942952}, the authors study topological boundaries of Levin-Wen models in terms of module categories. 
A $Q$-system $Q\in \cC$ gives the module category $\cC_{Q}$ of right $Q$-modules in $\cC$. 
The states above should correspond to the ground states of the corresponding commuting projector Hamiltonians. 
In future work, we plan to analyze these states in more detail and rigorously study their superselection sectors, making this connection more explicit.
\end{ex}

\begin{rem}
Even though the canonical state $\psi_\fF$ is not always tracial, we may consider the canonical inductive limit trace $\tr_\fF$ as a state on the boundary algebra. This represents the unique infinite temperature boundary equilibrium state.
\end{rem}

\section{Bulk topological order from the boundary algebra net in 2+1 D}
\label{sec:BulkBoundaryCorrespondence}

The boundary algebras in Construction \ref{const:LTO-BoundaryNet} can yield highly non-trivial nets. 
Even starting with an ordinary spin system such as the Toric Code we already obtain a non-trivial net. 
In fact, Examples~\ref{ex:TCBulkBoundary} and~\ref{rem:LWBulkBoundary} below indicate that the algebraic structure of boundary nets contain information about the \emph{bulk} topological order.
We conjecture that this is a general phenomenon, and that the algebraic structure of the net of boundary algebras completely captures the topological order of the bulk Hamiltonian alone, without reference to a Hamiltonian, just the net of projections.
From a mathematical perspective, this observation is somewhat surprising.
In this section we provide the key points behind this idea.

The idea is to consider the \emph{category of DHR bimodules} associated to the boundary net of $\rm C^*$-algebras introduced in \cite[\S3]{MR4814692}, which was shown to be a braided unitary tensor category for lattices. 
The definitions utilize the theory of \textit{correspondences} over $\rm C^*$-algebras, for which we refer the reader to \cite[\S3.1]{MR4814692} and references therein.

\begin{defn} 
\label{defn:LocalizableCorrespodence}
Consider a rectangle $\Lambda\subseteq \cL$. 
A right finite correspondence $X$ over the quasi-local algebra $\fA$ is \emph{localizable in} $\Lambda$ if there exists a finite projective basis (PP-basis) 
$\{b_{i}\}\subseteq X$ such that $ab_{i}=b_{i}a$ for all $a\in \fA(\Lambda^{c}):={\rm C^*}\left(\bigcup_{\Delta\subset \Lambda^c} \fA(\Delta)\right)$.
Here $\{b_i \}$ being a PP-basis (after Pimsner and Popa) means that $\sum_{i} b_i \langle b_i \vert x \rangle = x$ for all $x \in X$, where $\langle \cdot | \cdot \rangle$ is the $\fA$-valued inner product on $X$.
\end{defn}

\begin{defn} 
A right finite correspondence $X$ is called \textit{localizable} if 
$X$ is localizable in all rectangles sufficiently large relative to $X$, i.e., 
there exists an $r>0$ (depending on $X$) such that for any $\Lambda\subseteq \cL$ containing an $r^n$-cube, 
$X$ is localizable in $\Lambda$. 
We denote by $\DHR(\fA)$ the full $\rm C^*$-tensor subcategory of right finite correspondences consisting of localizable bimodules.
\end{defn}

\begin{defn}
A net satisfies \textit{weak algebraic Haag duality} if 
there is a global $t>0$
such that for all sufficiently large rectangles $\Lambda$,
$\fA(\Lambda^{c})'\subset \fA(\Lambda^{+t})$, where the prime denotes the commutant in $\fA$ and $\Lambda^{+t}$ is as in Definition~\ref{defn:BoundedSpread}.
(This definition of weak algebraic Haag duality is equivalent to the one given in \cite{MR4814692}.)
\end{defn}

By \cite[Prop.~2.11]{MR4814692}, weak algebraic Haag duality is preserved by bounded spread isomorphism.
More importantly, it allows one to prove the following result, which is central to the bulk-boundary correspondence.
It tells us that the DHR bimodules can be endowed with a braiding, which is expected if they are to describe the topological excitations in the bulk.
Moreover, invariance under bounded spread isomorphism (up to equivalence) allows us to relate the boundary nets to fusion categorical nets $\fF$, for which $\DHR(\fF)$ can be found explicitly.

\begin{thm*}[{\cite[Thm.~B]{MR4814692}}] If a net $\fA$ satisfies weak algebraic Haag duality, then $\DHR(\fA)$ admits a canonical braiding. 
If $\fA$ is isomorphic to $\fB$ by a bounded spread isomorphism, then $\DHR(\fA)\cong\DHR(\fB)$.
\end{thm*}

We can define the braiding in fairly simple terms, and the technicalities arise in showing it is a well-defined bimodule intertwiner. 
(This definition and strategy is very similar to that used in \cite{MR4753059}.)
Let $X, Y$ be DHR bimodules, $\{b_{i}\}, \{c_{j}\}$ be projective bases localized in sufficiently large balls $F$ and $G$ which are sufficiently far apart (see \cite{MR4814692} for details). 
Define
$$
u^{F,G}_{X,Y}: X\boxtimes Y\rightarrow Y\boxtimes X
\qquad\qquad\text{by}\qquad\qquad
u^{F,G}_{X,Y}
\left(
\sum_{i,j}b_{i}\boxtimes c_{j}a_{ij}
\right)
:=
\sum_{i,j}c_{j}\boxtimes b_{i} a_{ij}.
$$
One can show that this is a well-defined, unitary bimodule intertwiner which does not depend on the balls $F,G$ as long as they are sufficiently large and sufficiently far apart. 
(In the case $n=1$, they must also have the same ordering, i.e., we require $f<g$ for all $f\in F$ and $g\in G$.)

\begin{conj}[Bulk-Boundary Correspondence] 
\label{conj:BulkBoundary}
For 2D Hamiltonians on a spin system satisfying \ref{LTO:Hastings}--\ref{LTO:Injective}, the bulk topological order is the braided DHR category of bimodules of the net of boundary algebras.
\end{conj}

Of course, for the category of DHR bimodules on the boundary to be braided, we need the boundary net to satisfy weak algebraic Haag duality. 
This happens in practice (see Examples~\ref{ex:TCBulkBoundary} and~\ref{rem:LWBulkBoundary}) but appears to not be automatic. 
Our justification for Conjecture~\ref{conj:BulkBoundary} is that it gives the correct result for the Toric Code and for Levin-Wen string nets, as we show in the following two examples, using the following theorem.

\begin{thm*}[{\cite[Thm.~C]{MR4814692}}]
If $\fF$ is the 1D net constructed from a unitary fusion category $\cC$ from Example \ref{ex:TensorCat}, and $X\in \cC$ strongly tensor generates $\cC$,
then $\DHR(\fF)\cong Z(\cC)$.
\end{thm*}

\begin{ex}
\label{ex:TCBulkBoundary}
Consider the Toric Code model from \S\ref{sec:toric}.
The boundary net $\fB$ is isomorphic to $\fC$ and $\fD$ by Corollary \ref{cor:Rough==Smooth} and Theorem \ref{thm:TC-LTO},
which is also isomorphic to the fusion category net $\fF$ from Example \ref{ex:TensorCat} for $\cC=\fdHilb(\bbZ/2)$ and $X=1\oplus g$.
Because the object $X=1\oplus g$ strongly tensor generates $\fdHilb(\bbZ/2)$, the net $\fF$ satisfies weak algebraic Haag duality by~\cite[Thm.~C]{MR4814692}.
Since $\fB$ and $\fD$ are isomorphic to $\fF$ by bounded spread isomorphism, they also satisfy weak algebraic Haag duality~\cite[Prop.~2.11]{MR4814692}.
Furthermore, by \cite[Thm.~B]{MR4814692}, the category of DHR bimodules is preserved under bounded spread isomorphism, and thus 
$$
\DHR(\fB)
\underset{\text{\scriptsize\cite[Thm.~B]{MR4814692}}}{\cong}
\DHR(\fD)
\underset{\text{\scriptsize\cite[Thm.~B]{MR4814692}}}{\cong}
\DHR(\fF)
\underset{\text{\scriptsize\cite[Thm.~C]{MR4814692}}}{\cong}
Z(\fdHilb(\bbZ/2)).
$$
This recovers the well-known bulk topological order of the Toric Code, supporting Conjecture~\ref{conj:BulkBoundary}.
\end{ex}

\begin{ex}
\label{rem:LWBulkBoundary}
Just as in the Toric Code case, the boundary net $\fB$ for the Levin-Wen model for $\cC$ is isomorphic to the fusion categorical net $\fF$ for the UFC $\cC$ by Remark \ref{rem:IdentifyBoundaryLW}.
Thus
$$
\DHR(\fB)
\cong
\DHR(\fF)
\underset{\text{\scriptsize\cite[Thm.~C]{MR4814692}}}{\cong}
Z(\cC).
$$
This recovers the well-known bulk topological order of the Levin-Wen string net model, providing further support for Conjecture~\ref{conj:BulkBoundary}. 
\end{ex}

\subsection{Heuristic for DHR bimodules}

Physically, we can think of the category of DHR bimodules as existing in the emergent `time' direction. 
This is compatible with the viewpoint of \cite{2205.05565}, and allows us to think of DHR bimodules as topologically Wick-rotated point defects, which gives some justification for this correspondence. 

Here, we give a heuristic in terms of string operators that is somewhat model independent. 
Suppose we have a topological string operator, terminating in an excitation in the boundary $\cK\subset \cL$.
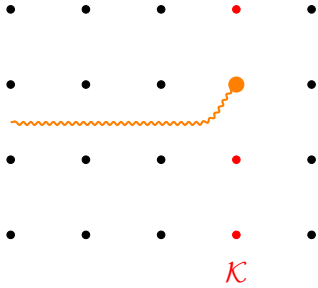
$$
\tikzmath{
\foreach \y in {0,1,2,3}{
\foreach \x in {0,1,2,4}{
\filldraw (\x,\y) circle (.05cm);
}
\filldraw[red] (3,\y) circle (.05cm);
}
\node[red] at (3,-.5) {$\cK$};
\filldraw[orange] (3,2) circle (.1cm);
\draw[orange, thick, decorate, decoration={snake, segment length=1mm, amplitude=.2mm}] (3,2) to[out=-135,in=0] (2.5,1.5) -- (0,1.5);
}
$$
Given a string operator $S$ with endpoint localized at a vertex $v\in \cK\subset \cL$, we define a DHR bimodule $\cY$ for the boundary algebra $\fB$ as follows.
We define a defect Hilbert space localized at $v$, and we define $\cY=\varinjlim \cY_n$ where each $\cY_n$ is defined in a certain rectangle $\Delta_n$ (see below for details).
Without loss of generality, we assume that $\cK$ is oriented vertically and we have chosen a distinguished half-plane $\bbH$ to the left of $\cK$.

First, we define $\Lambda_n$ to be the smallest sufficiently large rectangle whose right edge has length $2n+1$ and is centered at $v$, and $\Delta_n$ is the smallest rectangle such that $\Lambda_n \Subset_s \Delta_n$ with $\partial \Lambda_n \cap \partial \Delta_n \subset \cK$.
Here is a cartoon of $\Lambda_1\Subset_s \Delta_1$ with $r=3$ and $s=1$:
$$
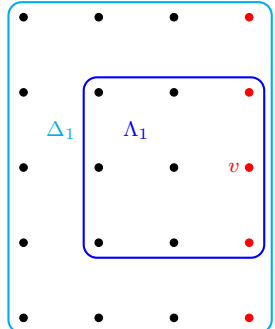
\tikzmath{
\foreach \y in {0,1,2,3,4}{
\foreach \x in {1,2,3}{
\filldraw (\x,\y) circle (.05cm);
}
\filldraw[red] (4,\y) circle (.05cm);
}
\draw[thick, cyan, rounded corners=5pt] (.8,-.2) rectangle (4.3,4.2);
\node[cyan] at (1.5,2.5) {$\scriptstyle \Delta_1$};
\draw[thick, blue, rounded corners=5pt] (1.8,.8) rectangle (4.2,3.2);
\node[blue] at (2.5,2.5) {$\scriptstyle \Lambda_1$};
\node[red] at (3.8,2) {$\scriptstyle v$};
}
$$
We define $\cY_n$ to be the space of homs between the tensor product Hilbert spaces localized in $\Lambda_n$ from sites in our original lattice to sites in our lattice with our new defect Hilbert space.
However, we cut down on both sides by projectors; we precompose with $p_{\Delta_n}$, and we post-compose with the projector $q_{\Delta_n}$ corresponding to the new tensor product space which carries the defect space at $v\in \Delta_n$.
For example, 
$$
\cY_1:=
q_{\Delta_1}
\Hom\left(\quad
\tikzmath{
\foreach \y in {1,2,3}{
\foreach \x in {2,3}{
\filldraw (\x,\y) circle (.05cm);
}
\filldraw[red] (4,\y) circle (.05cm);
}
\draw[thick, blue, rounded corners=5pt] (1.8,.8) rectangle (4.2,3.2);
\node[blue] at (2.5,2.5) {$\scriptstyle \Lambda_1$};
\node[red] at (3.8,2) {$\scriptstyle v$};
}
\qquad
\longrightarrow
\qquad
\tikzmath{
\foreach \y in {1,2,3}{
\foreach \x in {2,3}{
\filldraw (\x,\y) circle (.05cm);
}
\filldraw[red] (4,\y) circle (.05cm);
}
\draw[thick, blue, rounded corners=5pt] (1.8,.8) rectangle (4.2,3.2);
\node[blue] at (2.5,2.5) {$\scriptstyle \Lambda_1$};
\filldraw[orange] (4,2) node[left]{$\scriptstyle v$} circle (.1cm);
\draw[orange, thick, decorate, decoration={snake, segment length=1mm, amplitude=.2mm}] (4,2) to[out=-135,in=0] (3.5,1.5) -- (1.6,1.5);
}
\quad\right)
p_{\Delta_1}.
$$

The right action of $\fB(I):=\fB(\Lambda_I \Subset \Delta_I)$ where $\partial\Lambda\cap \cK=I$
is the standard composition of bounded operators, where we implicitly embed $\fB(I)\hookrightarrow \fB(I)p_{\Delta_n}$ using \ref{LTO:Injective}. 
That is, the right action of $a\in \fB(I)$ is given by
$\cY_n \ni f \mapsto f\circ a \in \cY_n$.
Using the same arguments as in Construction \ref{const:LTO-BoundaryNet}, this action should stabilize for large $n$.

The left action of $\fB(I)$ is more interesting.
First, we use a unitary hopping operator $h$ to move our excitation completely out of $\Delta_n$ where $n$ is large so that $\Delta_I\subset \Delta_n$.
We then post-compose with $a$, and then re-apply the adjoint $h^\dag$ of the hopping operator $h$ to bring the excitation back where it started, i.e.,
$\cY_n \ni f \mapsto h^\dag\circ a \circ h\circ f\in \cY_n$. 
Again, using the arguments of Construction \ref{const:LTO-BoundaryNet}, this action should stabilize for large $n$.

\subsection{Implementation of the DHR bimodules for Levin-Wen}
\label{sec:DHR-Implementation}

We now give a brief construction of a DHR bimodule for the fusion categorical net $\fF$ of the UFC $\cC$ associated to a string operator for $z\in Z(\cC)$ following the above heuristic.
These give all DHR bimodules for $\fF$ by \cite[Thm.~C]{MR4814692}.
First, we define a \emph{defect Hilbert space}
$$
\tikzmath{
\draw[thick] (-.5,0) node[left]{$\scriptstyle a$} -- (.5,0) node[right]{$\scriptstyle d$};
\draw[thick] (0,-.5) node[below]{$\scriptstyle b$} -- (0,.5) node[above]{$\scriptstyle c$};
\filldraw (0,0) circle (.05cm);
\draw[thick, orange, decorate, decoration={snake, segment length=1mm, amplitude=.2mm}] (0,0) -- (-.3,.5) node[left]{$\scriptstyle z$}; 
\node at (-.2,-.2) {$\scriptstyle v$};
\draw[blue!50, very thin] (-.5,.25) -- (.5,-.25);
{\draw[blue!50, very thin, -stealth] ($ (-.3,.15) - (.06,.12)$) to ($ (-.3,.15) + (.06,.12)$);}
{\draw[blue!50, very thin, -stealth] ($ (.3,-.15) - (.06,.12)$) to ($ (.3,-.15) + (.06,.12)$);}
}
\leftrightarrow
\cD_v:=
\bigoplus_{\substack{
a,b,c,d\in \cC
\\
z\in \Irr(Z(\cC))
}}
\cC(a\otimes b \to F(z)\otimes c\otimes d)
$$
where $F:Z(\cC)\to \cC$ is the forgetful functor.
Choosing a distinguished vertex $v\in \bbZ^2$, we can modify the total Hilbert space by replacing $\cH_v$ with $\cD_v$ at $v$.
For a rectangle $\Lambda$ containing $v$, 
whereas we defined $\cH(\Lambda):= \bigotimes_{u\in \Lambda} \cH_u$, we now define $\widehat{\cH}(\Lambda,v):=\cD_v\otimes\bigotimes_{u\in \Lambda\setminus\{v\}} \cH_u$.
We define the projector $C_v$ at $v$ to be the operator which selects the copy of $1\in Z(\cC)$ on the orange wavy string.
We also modify the plaquette operator in the northwest plaquette to $v$ to incorporate the half-braiding for $z$ as follows:
\begin{align*}
\frac{1}{D_\cC}\sum_{r\in \Irr(\cC)}
d_r\cdot
\tikzmath{
\draw[step=1.0,black,thin] (0.5,0.5) grid (2.5,2.5);
\filldraw[thick, blue, rounded corners=5pt, fill=gray!30] (1.15,1.15) rectangle (1.85,1.85);
\draw[thick, orange, knot, decorate, decoration={snake, segment length=1mm, amplitude=.2mm}] (2,1) -- ($ (2,1) + (-.4,.4)$);
\node at (2.3,1.2) {$\scriptstyle \xi_{2,1}$};
\node at (.7,2.2) {$\scriptstyle \xi_{1,2}$};
\node at (.7,1.2) {$\scriptstyle \xi_{1,1}$};
\node at (2.3,2.2) {$\scriptstyle \xi_{2,2}$};
\node[blue] at (1.3,1.5) {$\scriptstyle r$};
\node at (1.5,.85) {$\scriptstyle g$};
\node at (2.15,1.5) {$\scriptstyle h$};
\node at (1.5,2.15) {$\scriptstyle i$};
\node at (.85,1.5) {$\scriptstyle j$};
}
&=
\frac{1}{D_\cC}\sum_{r,k,\ell,m,n \in \Irr(\cC)}
\frac{\sqrt{d_kd_\ell d_md_n}}{d_r\sqrt{d_gd_hd_id_j}}
\cdot
\tikzmath{
\draw[step=1.0,black,thin] (0.5,0.5) grid (2.5,2.5);
\fill[gray!30, rounded corners=5pt] (1.15,1.15) rectangle (1.85,1.85);
\draw[thick, blue] (1.3,1) -- (1,1.3);
\draw[thick, blue] (1.7,1) -- (2,1.3);
\draw[thick, blue] (1.3,2) -- (1,1.7);
\draw[thick, blue] (1.7,2) -- (2,1.7);
\draw[thick, orange, knot, decorate, decoration={snake, segment length=1mm, amplitude=.2mm}] (2,1) -- ($ (2,1) + (-.4,.4)$);
\fill[fill=red] (1.3,1) circle (.05cm);
\fill[fill=red] (1.7,1) circle (.05cm);
\fill[fill=purple] (2,1.3) circle (.05cm);
\fill[fill=purple] (2,1.7) circle (.05cm);
\fill[fill=yellow] (1.3,2) circle (.05cm);
\fill[fill=yellow] (1.7,2) circle (.05cm);
\fill[fill=green] (1,1.3) circle (.05cm);
\fill[fill=green] (1,1.7) circle (.05cm);
\node at (2.3,.8) {$\scriptstyle \xi_{2,1}$};
\node at (.7,2.2) {$\scriptstyle \xi_{1,2}$};
\node at (.7,.8) {$\scriptstyle \xi_{1,1}$};
\node at (2.3,2.2) {$\scriptstyle \xi_{2,2}$};
\node at (1.15,.85) {$\scriptstyle g$};
\node at (1.85,.85) {$\scriptstyle g$};
\node at (2.15,1.15) {$\scriptstyle h$};
\node at (2.15,1.85) {$\scriptstyle h$};
\node at (1.15,2.15) {$\scriptstyle i$};
\node at (1.85,2.15) {$\scriptstyle i$};
\node at (.85,1.15) {$\scriptstyle j$};
\node at (.85,1.85) {$\scriptstyle j$};
\node at (1.5,.85) {$\scriptstyle k$};
\node at (2.15,1.5) {$\scriptstyle \ell$};
\node at (1.5,2.15) {$\scriptstyle m$};
\node at (.85,1.5) {$\scriptstyle n$};
}
\\&=
\sum_{\eta}
C(\xi,\eta)
\tikzmath{
\draw[step=1.0,black,thin] (0.5,0.5) grid (2.5,2.5);
\fill[gray!30, rounded corners=5pt] (1.1,1.1) rectangle (1.9,1.9);
\draw[thick, orange, knot, decorate, decoration={snake, segment length=1mm, amplitude=.2mm}] (2,1) -- ($ (2,1) + (-.4,.4)$);
\node at (2.3,1.2) {$\scriptstyle \eta_{2,1}$};
\node at (.7,2.2) {$\scriptstyle \eta_{1,2}$};
\node at (.7,1.2) {$\scriptstyle \eta_{1,1}$};
\node at (2.3,2.2) {$\scriptstyle \eta_{2,2}$};
}\,.
\end{align*}
By an argument similar to \cite[Prop.~3.2]{2305.14068}, the coefficient 
$$
C(\xi,\xi')
=
\delta_{z=z'}\delta_{a=a'}\cdots\delta_{h=h'}
\frac{1}{D_\cC\sqrt{d_a\cdots d_h}}
\cdot
\tikzmath{
\clip (-1.8,-1.8) rectangle (4.8,4.8);
\draw (0,0) rectangle (1,1);
\draw (2,2) rectangle (3,3);
\draw (1,1) to[out=0, in=-90] (2,2);
\draw (1,1) to[out=90, in=180] (2,2);
\draw (1,0) to[out=0, in=-90] (3,2);
\draw (1,0) .. controls ++(-90:2cm) and ++(0:2cm) .. (3,2);
\draw (0,1) to[out=90, in=180] (2,3);
\draw (0,1) .. controls ++(180:2cm) and ++(90:2cm) .. (2,3);
\draw (0,0) .. controls ++(180:5cm) and ++(90:5cm) .. (3,3);
\draw (0,0) .. controls ++(-90:5cm) and ++(0:5cm) .. (3,3);
\draw[thick, orange, knot, decorate, decoration={snake, segment length=1mm, amplitude=.2mm}] (1,0) to[out=135, in=-135] (1,2) node[above, yshift=.05cm]{$\scriptstyle F(z)$} to[out=45,in=135] (3,2);
\node at (.7,-.2) {$\scriptstyle \xi_{2,1}'$};
\node at (.3,.8) {$\scriptstyle \xi_{1,2}'$};
\node at (-.3,-.2) {$\scriptstyle \xi_{1,1}'$};
\node at (1.3,.8) {$\scriptstyle \xi_{2,2}'$};
\node at (3.3,2.2) {$\scriptstyle \xi_{2,1}^\dag$};
\node at (2.3,2.7) {$\scriptstyle \xi_{1,2}^\dag$};
\node at (3.3,2.7) {$\scriptstyle \xi_{1,1}^\dag$};
\node at (2.3,1.7) {$\scriptstyle \xi_{2,2}^\dag$};
\node at (-1,.3) {$\scriptstyle a$};
\node at (0,-1) {$\scriptstyle b$};
\node at (.9,-.5) {$\scriptstyle c$};
\node at (1.2,.2) {$\scriptstyle d$};
\node at (1.8,1.1) {$\scriptstyle e$};
\node at (1.3,1.3) {$\scriptstyle f$};
\node at (.2,1.2) {$\scriptstyle g$};
\node at (-.2,.8) {$\scriptstyle h$};
}\,,
$$
and thus this modified plaquette operator is a self-adjoint projector.
For rectangles $\Lambda$ containing the distinguished vertex $v$, we define the projector $q_\Lambda:= \prod_{\ell\subset \Lambda} A_\ell \prod_{p\subset \Lambda} B_p$, where the plaquette operator $B_p$ for the plaquette northwest to $v$ has been modified as above.

For any two $u,v\in \cL$, rectangle a $\Lambda$ containing $u,v$, and a localized excitation $z\in \Irr(\cC)$ at $u$, we have a unitary \emph{hopping operator} $H^z_{v,u}:\widehat{\cH}(\Lambda,u)\to \widehat{\cH}(\Lambda,v)$ given by 
$$
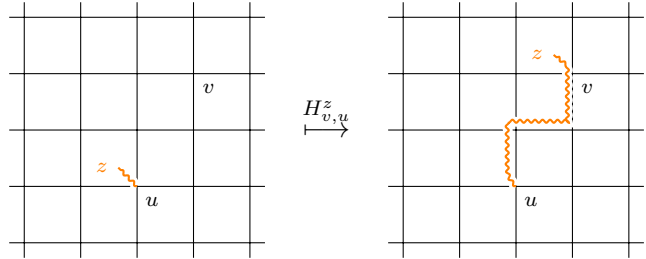
\tikzmath{
\draw[step=.75] (-.2,-.2) grid (3.2,3.2);
\draw[thick, orange, knot, decorate, decoration={snake, segment length=1mm, amplitude=.2mm}] (1.5,.75) -- (1.25,1) node[left]{$\scriptstyle z$};
\node at (1.7,.55) {$\scriptstyle u$};
\node at (2.45,2.05) {$\scriptstyle v$};
}
\quad
\overset{H^z_{v,u}}{\longmapsto}
\quad
\tikzmath{
\draw[step=.75] (-.2,-.2) grid (3.2,3.2);
\draw[thick, orange, knot, decorate, decoration={snake, segment length=1mm, amplitude=.2mm}] (1.5,.75) to[out=135, in=-90] (1.4,1) -- (1.4,1.6) -- (2.19,1.6) -- (2.19,2.34) -- (2,2.5) node[left]{$\scriptstyle z$};
\node at (1.7,.55) {$\scriptstyle u$};
\node at (2.45,2.05) {$\scriptstyle v$};
}\,
$$
where we implicitly use the fusion relation to resolve the $z$-excitation into each edge along the path.

We now focus on the case that the distinguished vertex lies in our codimension 1 hyperplane $\cK\subset \bbZ^2$.
The DHR-bimodule discussed in the heuristic above is given by 
\begin{equation}
\label{eq:DHRbimoduleXn}
\cY^z
=
\varinjlim
\cY^z_n 
=
\varinjlim
q_{\Delta_n} 
\Hom\left(
\cH(\Lambda_n)
\to
\widehat{\cH}(\Lambda_n,v)
\right)
p_{\Delta_n}.
\end{equation}
We now show $\cY_n^z$ carries commuting left and right actions of $\fF(I)$ with $I=\partial \Lambda_n \cap \cK$.
Below, we write $\Delta=\Delta_n$ to ease the notation.
By Lemma \ref{lem:KongSkeinModule},
$p_\Delta \cH(\Delta)=p_\Delta \bigotimes_{u\in \Delta} \cH_u \cong \cS(\#\partial \Delta)$, which can be identified with
$$
\bigoplus_{\vec{c}_{\partial\Delta\setminus I},\vec{c}_{ I}}\cC(\vec{c}_{\partial\Delta\setminus I} \to \vec{c}_I)
\cong
\cC(X^{\#\partial\Delta\setminus I} \to X^{\# I})
\qquad\qquad\qquad
X:=\bigoplus_{c\in \Irr(\cC)} c
$$
with the skein module inner product,
where 
$\vec{c}_{\partial\Delta\setminus I}= c_{i_1}\cdots c_{i_m}$
is a tensor product of simples in $\cC$ over the sites of $\partial\Delta\setminus I$, and
$\vec{c}_I = c_{j_1}\otimes \cdots\otimes c_{j_n}$ 
is a tensor product of simples over the sites of $I$.
Similarly, by \cite[Thm.~3.4]{2305.14068}, $q_\Delta\widehat{\cH}(\Delta) =q_\Delta\left(\cD_v\otimes\bigotimes_{u\in \Delta\setminus\{v\}} \cH_u\right)$ is isomorphic to an `enriched' skein module of the form
$$
\bigoplus_{\vec{c}_{\partial\Delta\setminus I},\vec{c}_{ I}}\cC(\vec{c}_{\partial\Delta\setminus I} \to \vec{c}_I \otimes F(z))
\cong
\cC(X^{\#\partial\Delta\setminus I} \to  X^{\# I}\otimes F(z))
$$
where $F:Z(\cC)\to \cC$ is the forgetful functor.
Hence operators in $\cY_n^z$ from \eqref{eq:DHRbimoduleXn} can be viewed as operators 
$$
\cC(X^{\#\partial\Delta\setminus I} \to X^{\# I})
\to
\cC(X^{\#\partial\Delta\setminus I} \to X^{\# I}\otimes F(z))
$$
which commute with the left $\End_\cC(X^{\#\partial\Delta\setminus I})$-action.
By the Yoneda Lemma, we can identify 
$$
\cY^z_n=\cC(X^{\#I} \to X^{\# I}\otimes F(z)).
$$
The right $\fF(I)$-action on $\cY_n^z$ is exactly precomposition, and the left $\fF(I)$-action which uses the hopping operator $H^z_{v,u}$ to move the $z$-excitation out of $I$ before acting and then move it back is exactly postcomposition on the $X^{\#I}$ tensorand.
In diagrams:
$$
a\rhd \xi\lhd b
:=
\tikzmath{
\draw[thick, orange, knot, decorate, decoration={snake, segment length=1mm, amplitude=.2mm}] (.5,0) -- (1,0) node[right]{$\scriptstyle F(z)$} ;
\draw (-.3,.3) --node[left]{$\scriptstyle X$} (-.3,.7) -- (-.3,1.7) node[above]{$\scriptstyle X$};
\draw (.3,.3) --node[right]{$\scriptstyle X$} (.3,.7) -- (.3,1.7) node[above]{$\scriptstyle X$};
\node at (.05,.5) {$\cdots$};
\node at (.05,1.5) {$\cdots$};
\draw (-.3,-.3) --node[left]{$\scriptstyle X$} (-.3,-.7) -- (-.3,-1.7) node[below]{$\scriptstyle X$};
\draw (.3,-.3) --node[right]{$\scriptstyle X$} (.3,-.7) -- (.3,-1.7) node[below]{$\scriptstyle X$};
\node at (.05,-.5) {$\cdots$};
\node at (.05,-1.5) {$\cdots$};
\roundNbox{fill=white}{(0,0)}{.3}{.2}{.2}{$\xi$}
\roundNbox{fill=white}{(0,1)}{.3}{.2}{.2}{$a$}
\roundNbox{fill=white}{(0,-1)}{.3}{.2}{.2}{$b$}
}
\qquad\qquad \xi\in \cY^z_n,
\,a,b\in \End_\cC(X^{\#I}),
$$
where the orange $F(z)$-string should be viewed as in the target of the above morphism.

Now suppose $J$ is obtained from $I$ by adding $k$ boundary points below $I$ and $k$ boundary points above $I$, where we view $I$ as the right hand side of $\partial \Lambda_n$.
Then as we have 
\[
\begin{tikzcd}[column sep=0, row sep =0]
\Lambda_{n+2k} & \Subset_s & \Delta_{n+2k}
\\
\cup &&\cup
\\
\Lambda_n & \Subset_s & \Delta_n
\end{tikzcd}
\qquad\qquad
\text{such that} 
\qquad
\partial \Lambda_n \cap \partial \Delta_n =I= \partial \Lambda_{n}\cap \partial \Delta_{n+2k},
\]
we get an inclusion $\cY^z_n \hookrightarrow \cY^z_{n+2k}$ by adding $2k$ through strings as follows:
$$
\tikzmath{
\draw (-1.3,-.7) node[below]{$\scriptstyle X$} -- (-1.3,.7) node[above]{$\scriptstyle X$};
\draw (-.7,-.7) node[below]{$\scriptstyle X$} -- (-.7,.7) node[above]{$\scriptstyle X$};
\node at (1.05,.5) {$\cdots$};
\draw (1.3,-.7) node[below]{$\scriptstyle X$} -- (1.3,.7) node[above]{$\scriptstyle X$};
\draw (.7,-.7) node[below]{$\scriptstyle X$} -- (.7,.7) node[above]{$\scriptstyle X$};
\node at (-.95,0) {$\cdots$};
\draw[thick, orange, knot, decorate, decoration={snake, segment length=1mm, amplitude=.2mm}] (.5,0) -- (1.5,0) node[right]{$\scriptstyle F(z)$} ;
\draw (-.3,.3) -- (-.3,.7) node[above]{$\scriptstyle X$};
\draw (.3,.3) -- (.3,.7) node[above]{$\scriptstyle X$};
\node at (.05,.5) {$\cdots$};
\draw (-.3,-.3) -- (-.3,-.7) node[below]{$\scriptstyle X$};
\draw (.3,-.3) -- (.3,-.7)node[below]{$\scriptstyle X$};
\node at (.05,-.5) {$\cdots$};
\roundNbox{fill=white}{(0,0)}{.3}{.2}{.2}{$\xi$}
}\,.
$$
Observe that the inclusion $\cY^z_{n}\hookrightarrow \cY^z_{n+2k}$ is compatible with the inlcusions $\fF(I)\hookrightarrow \fF(J)$ from Lemma \ref{lem:StabilityOfBoundaryAlgebras} under both the left and right actions.
We thus get an inductive limit $\fF-\fF$ bimodule $\cY^z$, and this bimodule is exactly the one constructed for $z\in Z(\cC)$ in \cite{MR4814692}.
It follows that this bimodule is localizable.

\subsection{Boundary states and \texorpdfstring{$\rm W^*$}{W*}-algebras in \texorpdfstring{$\DHR(\fB)$}{DHR(B)}}

In this section, we focus on the case of a translation invariant 2D lattice model $(\fA,p)$ satifying \ref{LTO:Hastings}--\ref{LTO:Injective},
and let $\fB$ be the 1D boundary net. 
For a separable $\rm C^*$-algebra $A$, we write $\Rep(A)$ for the $\rm W^*$-category of separable Hilbert space representations~\cite{MR808930}.
If $H,K\in \Rep(A)$, we say $H$ is \emph{quasi-contained in} $K$, denoted $H\preceq K$, if $H$ is isomorphic to a summand of $K^{n}$ for some $n\in \bbN\cup \{\infty\}$. The reader should compare the following definition to ~\cite{MR1234107,MR1463825}.

\begin{defn}
\label{defn:BoundarySuperSelection}
    Let $\fB$ be a net of finite dimensional $\rm C^*$-algebras on the lattice $\mathbbm{Z}$ as in Definition~\ref{defn:LocalNet}, and let $\phi$ be a state on $\fB$. 
    A \textit{superselection sector} of $\phi$ is a Hilbert space representation $H$ of $\fB$ satisfying the following property:
    \begin{itemize}
        \item There exists an $r>0$ such that for any interval $I$ of length at least $r$, 
    $H|_{\fB(I^{c})}\preceq L^{2}(\fB, \phi)|_{\fB(I^{c})}$.
    \end{itemize}
    Here, $\fB(I^c)$ is the unital $\rm C^*$-subalgebra of $\fB$ generated by the $\fB(J)$ for all $J\subset I^c$.
    We denote by $\Rep_{\phi}(\fB)$ the full $\rm W^*$-subcategory of $\Rep(\fB)$ of superselection sectors of $\phi$.
\end{defn}

Note that by definition, $\Rep_{\phi}(\fB)$ is unitarily Cauchy complete (closed under orthogonal direct sums and orthogonal summands). 
We now assume that $\DHR(\fB)$ is a unitary tensor category, so that each $X\in \DHR(\fB)$ is dualizable and $\fB$ has trivial center, which is satisfied in all of our examples.
We refer the reader to \cite{MR1624182,MR2085108,MR4419534} for dualizable criteria for Hilbert $\rmC^*$ $\fB-\fB$ correspondences.
 
\begin{prop}$\Rep_{\phi}(\fB)$ is a $\DHR(\fB)$-module $\rm W^*$-category with action given by $X\triangleright H:=X\boxtimes_{\fB} H$. 
\end{prop}

\begin{proof}
Suppose $H\in \Rep_{\phi}(\fB) $. Then there exists some $r>0$ so that for any interval $I$ with at least $r$ sites, $H|_{\fB(I^{c})}\preceq L^{2}(\fB, \phi)|_{\fB(I^{c})}$. 
For $X\in \DHR(\fB)$, there exists some $s>0$ such that $X$ is localizable in any interval with at least $s$ sites. 
We claim $X\boxtimes_{\fB} H\in \Rep_{\phi}(\fB)$ with localization constant $t=\max\{r,s\}$. 
 
Indeed, let $I$ be an interval with at least $t$ sites, and let $\{b_{i}\}^{n}_{i=1}$ be a PP-basis of $X$ such that $xb_i=b_ix$ for all $x\in \fB(I^c)$, as in Definition~\ref{defn:LocalizableCorrespodence}.
Observe that $\langle b_{i} | b_{j}\rangle\in \fB(I^{c})' \cap \fB$.
Since the action of $\fB(I^{c})' \cap \fB$ commutes with $\fB(I^{c})$ in any representation, we have $M_{n}(\fB(I^{c})' \cap \fB)\subset \End_{\fB(I^{c})}(H^{\oplus n})$.

Now consider the orthogonal projection $P=(\langle b_{i} | b_{j} \rangle)_{i,j}\in M_{n}(\fB(I^{c})' \cap \fB)$, so 
$K:=PH^{\oplus n}$ is a summand of $(H|_{\fB(I^{c})})^{\oplus n}$. We claim $\left(X\boxtimes_{\fB} H\right)|_{\fB(I^{c})}$ is unitarily isomorphic to $K$.
Consider the map $v:X\otimes H\rightarrow K$ given by 
$$
v\left(
\sum_{i=1}^n b_{i}\otimes \xi_{i}
\right)
:=
P
\begin{bmatrix} 
\xi_{1} \\ \vdots \\ \xi_n 
\end{bmatrix}.
$$
By definition of $P$, $v$ extends to a unitary isomorphism
$\tilde{v}:X\boxtimes_{\fB} H \rightarrow K$. 
Observe $\tilde{v}$
intertwines the $\fB(I^{c})$-actions as each $b_{i}$ centralizes $\fB(I^{c})$.

We have thus shown $X\boxtimes_{\fB} H|_{\fB(I^{c})} \preceq H |_{\fB(I^{c})}$. 
Since $H |_{\fB(I^{c})}\preceq L^{2}(\fB, \phi)|_{\fB(I^{c})}$ and $\preceq$ is transitive, the claim follows.
\end{proof}

Note there is a distinguished object $L^{2}(\fB, \phi)\in \Rep_{\phi}(\fB)$, which gives this module $\rm W^*$-category a canonical pointing (more specifically, it gives the full module subcategory generated by $\DHR(\fB)$ and $L^{2}(\fB, \phi)$ a pointing).
By \cite{MR3687214}, we get a canonical $\rm W^*$-algebra object $A_\phi\in \Vect(\DHR(\fB)):=\Fun(\DHR(\fB)^{\rm op}\to \Vect)$ 
associated to $(\Rep_\phi(\fB), L^2(\fB,\phi))$ given by
$$
A_\phi(\cY)
:=
\Hom_{\Rep_\phi(\fB)}(\cY\boxtimes_{\fB} L^2(\fB,\phi) \to L^2(\fB,\phi)).
$$
Here, the $\rm W^*$-algebra object $A_\phi=\underline{\End}_{\DHR(\fB)}(L^2(\fB,\phi))$ is in general too large to live in $\DHR(\fB)$, but rather it lies in the \emph{ind-completion} $\Vect(\DHR(\fB))$.
We refer the reader to \cite{MR3687214} for more details.

\begin{defn}
For a boundary state $\phi$ on $\fB$, we define the \textit{boundary order} to be the isomorphism class of the $\rm W^*$-algebra object $A_{\phi}\in \Vect(\DHR(\fB))$.

We call the state $\phi$ a \textit{topological boundary} if $A_{\phi}\in \Vect(\DHR(\fB))$ is Lagrangian.
(Recall that a $\rm W^*$-algebra object in a braided tensor category is called \emph{Lagrangian} if it is commutative and its category of local modules is trivial.)
\end{defn}

\begin{ex}
When $\cC$ is a fusion category, the \emph{canonical Lagrangian} in $Z(\cC)$ is $I(1_\cC)$ where $I: \cC\to Z(\cC)$ is adjoint to the forgetful functor $Z(\cC)\to \cC$, as $I(1_\cC)=\underline{\End}_{Z(\cC)}(1_\cC)$ for the module action of $Z(\cC)$ on $\cC$ given by the forgetful functor.
\end{ex}

\begin{rem}
In future work we will study fusion of superselection sectors for topological boundaries.
One could actually use the above definition of topological boundary to define fusion of superselection sectors, as $A_\phi$-modules in $\Vect(\DHR(\fB))$ again form a tensor category.
\end{rem}

We now specialize to the case that $\fB=\fF$, the fusion categorical net from our UFC $\cC$.
All the states constructed from Q-systems in Example \ref{ex:QSystemBoundary} above give topological boundaries. 
Indeed, every Q-system in the fusion category $\cC$ gives a canonical Lagrangian algebra in the center $Z(\cC)\cong \DHR(\fF)$. 
We explicitly illustrate this for the case $\phi=\phi_1$ for the trivial Q-system $1\in \cC$, which corresponds to the canonical Lagrangian algebra in $Z(\cC)$.
The argument for the other Q-systems in $\cC$ is analogous, and we plan to carry out a more systematic analysis of these states in future work.

Let $\cM_\phi$ be the full subcategory of $\Rep_\phi(\fF)$ generated by $L^2(\fF,\phi_1)$ under the $\DHR(\fF)$-action under taking orthogonal direct sums and orthogonal subobjects.

\begin{construction}
\label{const:Z(C)ModuleFunctor}
We now build a left $Z(\cC)$-module functor $\cC\to \cM_\phi$.

First, for each non-empty interval $I\subset \bbZ$, we have a fully faithful functor 
$H_{-}(I): \cC^{\rm op}\to \Rep(\fF(I))$ given by
$$
H_a(I):=\cC(a\to X^{\#I})
\qquad\qquad\qquad
\langle\eta|\xi\rangle
:=
\tr_\cC(\eta^\dag\circ \xi).
$$
Observe that $H_a(I)$ is a left
$\fF(I)$-module where the action is given by postcomposition.
That this $\fF(I)$-action is compatible with $\dag$ in $\fF$ follows by the unitary Yoneda embedding \cite[Rem.~2.28]{MR3687214} (see also \cite[Rem.~3.61]{MR4598730}).
Moreover, precomposition with $\varphi \in \cC(a\to b)$ gives a bounded $\fF(I)$-linear map $\varphi_*:H_b(I)\to H_a(I)$, and $(\varphi_*)^\dag= (\varphi^\dag)_*$ again by the unitary Yoneda embedding.
We thus have a unitary functor $H_{-}(I):\cC^{\rm op}\to \Rep(\fF(I))$ which is fully faithful by the Yoneda Lemma as every simple object is a subobject of $X^{\#I}$.

For $I\subset J$, we have an isometry $H_a(I)\hookrightarrow H_a(J)$ given by tensoring with $J\setminus I$ copies of the unit isometry $\iota: 1_\cC\to X$.
$$
H_a(I)
\ni
\tikzmath{
\draw (-.3,.3) -- (-.3,.7);
\draw (.3,.3) -- (.3,.7);
\draw[thick, blue] (0,-.7) -- node[right]{$\scriptstyle a$} (0,-.3);
\roundNbox{fill=white}{(0,0)}{.3}{.2}{.2}{$\varphi$}
\node at (.05,.5) {$\cdots$};
}
\longmapsto
\tikzmath{
\draw (-.7,.3) -- (-.7,.7);
\draw (-1.3,.3) -- (-1.3,.7);
\draw (.7,.3) -- (.7,.7);
\draw (1.3,.3) -- (1.3,.7);
\filldraw (-.7,.3) circle (.05cm);
\filldraw (-1.3,.3) circle (.05cm);
\filldraw (.7,.3) circle (.05cm);
\filldraw (1.3,.3) circle (.05cm);
\draw (-.3,.3) -- (-.3,.7);
\draw (.3,.3) -- (.3,.7);
\draw[thick, blue] (0,-.7) -- node[right]{$\scriptstyle a$} (0,-.3);
\roundNbox{fill=white}{(0,0)}{.3}{.2}{.2}{$\varphi$}
\node at (.05,.5) {$\cdots$};
\node at (1.05,.5) {$\cdots$};
\node at (-.95,.5) {$\cdots$};
}
\in H_a(J).
$$
Moreover, the inclusion isometry $H_a(I)\hookrightarrow H_a(J)$ is compatible with the left actions of $\fF(I)\hookrightarrow \fF(J)$, so the inductive limit $H_a:=\varinjlim H_a(I)$ has a left $\fF$-action.
We thus get a unitary functor $H: \cC^{\rm op}\to \Rep(\fF)$.
Finally, we precompose with the canonical unitary duality $\cC\to \cC^{\rm mop}$,\footnote{Here, $\rm mop$ means taking both the monoidal and arrow opposite of $\cC$.} noting that $\cC^{\rm mop}\cong \cC^{\rm op}$ as categories where we have forgotten the monoidal structure.
This gives us a functor $\check{H}: \cC\to \Rep(\fF)$ given by $\check{H}_a:= H_{\overline{a}}$.

We can identify $\check{H}_{1_\cC}\cong L^{2}(\fF, \phi)$ as follows.
We have a unitary isomorphism
$L^2(\fF(I),\phi)\cong \check{H}_{1_\cC}(I)=\cC(1_\cC\to X^{\#I})$ 
given by
$$
\varphi \Omega
=
\tikzmath{
\draw (-.3,-1) -- (-.3,-1.4);
\draw (-.3,-.7) -- (-.3,.7);
\filldraw (-.3,-.7) circle (.05cm);
\filldraw (-.3,-1) circle (.05cm);
\draw (.3,-1) -- (.3,-1.4);
\draw (.3,-.7) -- (.3,.7);
\filldraw (.3,-.7) circle (.05cm);
\filldraw (.3,-1) circle (.05cm);
\roundNbox{fill=white}{(0,0)}{.3}{.2}{.2}{$\varphi$}
\node at (.05,-.5) {$\cdots$};
\node at (.05,-1.2) {$\cdots$};
\node at (.05,.5) {$\cdots$};
}
\,\,
\mapsto 
\,\,
\tikzmath{
\draw (-.3,-.7) -- (-.3,.7);
\filldraw (-.3,-.7) circle (.05cm);
\draw (.3,-.7) -- (.3,.7);
\filldraw (.3,-.7) circle (.05cm);
\roundNbox{fill=white}{(0,0)}{.3}{.2}{.2}{$\varphi$}
\node at (.05,-.5) {$\cdots$};
\node at (.05,.5) {$\cdots$};
}\,,
$$
and this map clearly intertwines the left $\fF(I)$-action.
Taking inductive limits, we get an $\fF$-linear unitary $L^2(\fF,\phi)\cong \check{H}_{1_\cC}$.

We now construct a 
natural unitary isomorphism $\mu_{z,c}:\cY^{z}\boxtimes_\fF \check{H}_c
\to \check{H}_{z\rhd c}$,
where $\cY^z$ is the DHR bimodule associated to $z\in Z(\cC)$ from \S\ref{sec:DHR-Implementation} above.
This isomorphism will clearly satisfy unital and associative coherences, which will endow $\check{H}$ with the structure of a 
$Z(\cC)$-module functor.
Observe that $Z(\cC)\cong Z(\cC)^{\rm rev}\cong Z(\cC^{\rm op})$ as tensor categories where we have forgotten the braiding, and the forgetful functor $Z(\cC^{\rm op})\to \cC^{\rm op}$ is dominant.
Since $\check{H}_1\cong L^2(\fF,\phi)\in \Rep_\phi(\fF)$, this will show that the image of $\check{H}$ lies in $\cM_\phi$.

Now for every interval $I\subset \bbZ$ containing the defect point for $\cY^z$, we have a unitary isomorphism 
$$
\mu_{z,c}(I):
\cY^{z}(I)\boxtimes_{\fF(I)} \check{H}_c(I) 
= 
\cY^z(I)\boxtimes_\fF H_{\overline{c}}(I)
\to 
\check{H}_{z\rhd c}(I)
=
H_{\overline{c}\otimes F(\overline{z})}(I)
$$ 
by gluing diagrams and bending the $z$-string down and to the right of $\overline{c}$:
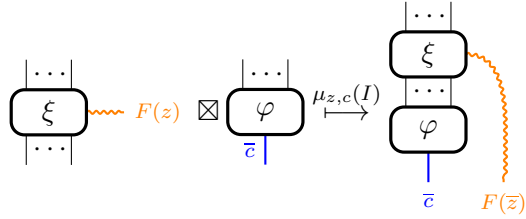
$$
\tikzmath{
\draw[thick, orange, knot, decorate, decoration={snake, segment length=1mm, amplitude=.2mm}] (.5,0) -- (1,0) node[right]{$\scriptstyle F(z)$} ;
\draw (-.3,.3) -- (-.3,.7);
\draw (.3,.3) -- (.3,.7);
\node at (.05,.5) {$\cdots$};
\draw (-.3,-.3) -- (-.3,-.7);
\draw (.3,-.3) -- (.3,-.7);
\node at (.05,-.5) {$\cdots$};
\roundNbox{fill=white}{(0,0)}{.3}{.2}{.2}{$\xi$}
}
\boxtimes
\tikzmath{
\draw (-.3,.3) -- (-.3,.7);
\draw (.3,.3) -- (.3,.7);
\draw[thick, blue] (0,-.7) -- node[left]{$\scriptstyle \overline{c}$} (0,-.3);
\roundNbox{fill=white}{(0,0)}{.3}{.2}{.2}{$\varphi$}
\node at (.05,.5) {$\cdots$};
}
\overset{\mu_{z,c}(I)}{\longmapsto}
\tikzmath{
\draw[thick, blue] (0,-1.3) -- (0,-1.7) node[below]{$\scriptstyle \overline{c}$};
\draw[thick, orange, knot, decorate, decoration={snake, segment length=1mm, amplitude=.2mm}] (.5,0) to[out=0,in=90] (1,-1) to[out=-90,in=90] (1,-1.7) node[below]{$\scriptstyle F(\overline{z})$};
\draw (-.3,.3) -- (-.3,.7);
\draw (.3,.3) -- (.3,.7);
\node at (.05,.5) {$\cdots$};
\draw (-.3,-.3) -- (-.3,-.7);
\draw (.3,-.3) -- (.3,-.7);
\node at (.05,-.5) {$\cdots$};
\roundNbox{fill=white}{(0,0)}{.3}{.2}{.2}{$\xi$}
\roundNbox{fill=white}{(0,-1)}{.3}{.2}{.2}{$\varphi$}
}
$$
The isomorphism $\mu_{z,c}(I)$ clearly satisfies the unit and associativity axioms for a modulator.
Moreover, if $I\subset J$, we get a commutative square
$$
\begin{matrix}
\cY^{z}(J)\boxtimes_{\fF(J)} \check{H}_c(J)
&
\underset{\phantom{A}}{\xrightarrow{\mu_{z,c}(J)}}
&
\check{H}_{z\rhd c}(J)
\\
\rotatebox{90}{$\hookrightarrow$} && \rotatebox{90}{$\hookrightarrow$}
\\
\cY^{z}(I)\boxtimes_{\fF(I)} \check{H}_c(I) 
&
\xrightarrow{\mu_{z,c}(I)}
&
\check{H}_{z\rhd c}(I).
\end{matrix}
$$
We thus get a well-defined unitary $\mu_{z,c}:=\varinjlim \mu_{z,c}(I): \cY^z\boxtimes_\fF \check{H}_c\to \check{H}_{z\rhd c}$, which endow $\check{H}$ with the structure of a $Z(\cC)$-module functor.
\end{construction}

With these considerations, we can prove the following theorem.

\begin{thm} 
Let $\cC$ be a UFC, let $\fF$ be the associated fusion categorical net, and let $\phi=\phi_1$ be the boundary state associatied to the trivial Q-system from Example \ref{ex:UnitState} above.
The functor $\check{H}: \cC\to \Rep_\phi(\fF)$ from Construction \ref{const:Z(C)ModuleFunctor} above is a $Z(\cC)$-module equivalence onto $\cM_\phi$.
In particular, $\phi$ is a topological boundary.
\end{thm}
\begin{proof} 
We first show $\check{H}$ is fully faithful.
To show $\Rep(\fF)(\check{H}_a\to \check{H}_b)\cong \cC(\overline{b}\to \overline{a})\cong \cC(a\to b)$ for $a,b\in\cC$, 
we show that every $\Rep(\fF)$-intertwiner 
$T: \check{H}_a\to \check{H}_b$
maps $\check{H}_a(I)\to \check{H}_b(I)$ for every interval $I\subset \bbZ$.
Since every $\check{H}_{-}(I)$ is fully faithful, and since $\check{H}=\varinjlim \check{H}_{-}(I)$, this will prove the result.

For $I\subset \bbZ$ and $n\geq \#I$, we define an orthogonal projection $q_n(I)$ which is $\id_X$ on sites in $I$ and $\iota\iota^\dag$ on the $n$ sites to the left and on the $n$ sites to the right of $I$, where $\iota: 1_\cC\to X$ is the inclusion.
$$
q_n(I):=
\tikzmath{
\draw (-.3,-.5) -- (-.3,.5);
\draw (.3,-.5) -- (.3,.5);
\node at (.05,0) {$\cdots$};
\draw (-.7,.2) -- (-.7,.5);
\draw (-.7,-.2) -- (-.7,-.5);
\draw (-1.3,.2) -- (-1.3,.5);
\draw (-1.3,-.2) -- (-1.3,-.5);
\draw (.7,.2) -- (.7,.5);
\draw (.7,-.2) -- (.7,-.5);
\draw (1.3,.2) -- (1.3,.5);
\draw (1.3,-.2) -- (1.3,-.5);
\filldraw (-.7,.2) circle (.05cm);
\filldraw (-.7,-.2) circle (.05cm);
\filldraw (-1.3,.2) circle (.05cm);
\filldraw (-1.3,-.2) circle (.05cm);
\filldraw (.7,.2) circle (.05cm);
\filldraw (.7,-.2) circle (.05cm);
\filldraw (1.3,.2) circle (.05cm);
\filldraw (1.3,-.2) circle (.05cm);
\node at (1.05,.35) {$\cdots$};
\node at (1.05,-.35) {$\cdots$};
\node at (-.95,.35) {$\cdots$};
\node at (-.95,-.35) {$\cdots$};
}
\in
\fF(n+I+n).
$$
Since $T$ intertwines the $\fF$-actions, we have that $Tq_n(I)=q_n(I)T$.
Thus if $x\in \check{H}_a(I)$, $x=q_n(I)x$ for all $n\in \bbN$, and thus
$Tx=Tq_n(I)x=q_n(I)Tx$ for all $n\in \bbN$.
We claim this means $Tx\in \check{H}_b(I)$.
Indeed, $\check{H}_b$ is filtered by the finite dimensional subspaces $\check{H}_b(I)$, and thus we can write $\check{H}_b$ as an orthogonal direct sum
$$
\check{H}_b
=
\check{H}_b(I)
\oplus \bigoplus_{n\in \bbN}  [\check{H}_b(n+I+n) \ominus \check{H}_b((n-1)+I+(n-1))].
$$
Observe that the subspace $\check{H}_b(n+I+n) \ominus \check{H}_b((n-1)+I+(n-1))$ on the right is given by adding a $\id_X-\iota\iota^\dag$ on both outer-most strands.
Now writing $Tx = (Tx)_I+ \sum_{n\in \bbN} (Tx)_n$ in this decomposition, we see that $Tx=q_n(I)Tx$ for all $n$, and $q_n(I)(Tx)_k=0$ for all $k=1,\dots, n-1$.
We conclude that $(Tx)_n=0$ for all $n$, and thus $Tx = (Tx)_I\in \check{H}_b(I)$ as claimed.

Now since $\cM_\phi$ was defined as the full $\rm W^*$-subcategory of $\Rep_\phi(\fF)$ generated by $L^2(\fF,\phi_1)$ under the $\DHR(\fF)\cong Z(\cC)$-action, $\check{H}$ is dominant.
Since both $\cC$ and $\cM_\phi$ are unitarily Cauchy complete, $\check{H}$ is a unitary equivalence.
\end{proof}

\appendix
\section{Path algebras and type \texorpdfstring{$\rm III$}{III} factors
\texorpdfstring{\\}{} 
by Masaki Izumi}
\label{app:typeiii}
In this appendix authored by Masaki Izumi,
we further classify the type of the factor $\fF''$ for the boundary states $\psi_\fF$ (see \S\ref{sec:LWBoundaryState}) of Levin-Wen models (defined in \S\ref{sec:LWStringNet}) in the type $\rm{III}$ case. 
We first prove a general result on the type of the von Neumann algebra $M=B''$ generated by a certain KMS state $\phi$ on an AF-algebra $B$ defined in terms of a finite oriented graph $\cG$.
The key point in the proof is that we can identify Krieger's ratio set~\cite{MR0414823}, which gives us Connes' invariant $S(M)$~\cite{MR0341115}.
See~\cite[\S 2]{MR0578730} or~\cite[\S XIII.2]{MR1943007} for the equality between the two sets.
This result can then be applied to the setting of the Levin--Wen model, leading to a proof of the claim in Remark~\ref{rem:LWBoundaryType}.

\subsection{Path \texorpdfstring{$\rm C^*$}{C*}-algebras}
\label{sec:PathAlgebra}
Throughout this appendix, let $\cG$ be a finite oriented graph.
We write $\cG^{(0)}$ and $\cG^{(1)}$ for the set of vertices and edges respectively.
With $s,r: \cG^{(1)} \to \cG^{(0)}$ we denote the source and range maps.
We furthermore assume the choice of a distinguished vertex $\ast \in \cG^{(0)}$.
For simplicity, we make the following assumption:
\begin{itemize}
\item 
For any pair of vertices $v_1, v_2$, there is an edge $e\in\cG^{(1)}$ with $s(e) = v_1$ and $r(e) = v_2$.
\end{itemize}
We define
\[
    \Path^n_v(\cG) :=  \set{\xi_1 \xi_2 \cdots \xi_n}{r(\xi_i) = s(\xi_{i+1}),\,\, s(\xi_1) = v}
\]
as the set of all paths of length $n$ starting at the vertex $v$.
For any finite path $\xi$, we will write $|\xi|$ for its length, and $r(\xi)$ for the range of the last edge in the path.

The paths starting in the distinguished vertex $*$ can be used to define a tower of finite dimensional $\rm C^*$-algebras $(B_n)$ as in \cite{MR996454,MR999799,MR1642584,MR1473221,MR2812459}.
For each pair $\xi, \eta \in \Path^n_{\ast}(\cG)$, define the formal symbol $\pbk{\xi}{\eta}$ and let
\[
    B_n := \operatorname{span}_\bbC \set{\pbk{\xi}{\eta}}{\xi,\eta \in \Path_{\ast}^n(\cG),\,\, r(\xi) = r(\eta)},
\]
with multiplication and $*$-structure given by
\[
    \pbk{\xi_1}{\xi_2} \cdot \pbk{\eta_1}{\eta_2} := \delta_{\xi_2= \eta_1} \pbk{\xi_1}{\eta_2}
    \qquad\qquad\text{and}\qquad\qquad
    \pbk{\xi}{\eta}^* := \pbk{\eta}{\xi}.
\]
Since $\cG$ is finite, $B_n$ is a $\rm C^*$-algebra
which is the direct sum of $\left| \cG^{(0)} \right|$ full matrix algebras.
For each $v \in \cG^{(0)}$, the corresponding summand is an $n_v \times n_v$ matrix algebra, where $n_v$ is the number of paths from $\ast$ to $v$ of length $n$.
One can think of $\pbk{\xi}{\eta}$ as a loop of length $2n$ on $\cG$ starting at $*$ where you first travel $\xi$ to $r(\xi)=r(\eta)$, and then travel $\eta$ in the reverse direction.
This gives the loop algebra convention of \cite{MR2812459}.

The inclusion $\iota_{n}: B_n \hookrightarrow B_{n+1}$ is given by
\begin{equation}
\label{eq:BnInclusion}
    \iota_{n}\left(\pbk{\xi_1}{\xi_2}\right) := \sum_{ \substack{ r(\xi_1) = s(\eta) \\ |\eta| = 1}}  \pbk{\xi_1 \eta}{\xi_2 \eta}.
\end{equation}
That is, we extend each pair of paths by one edge in all possible ways satisfying the loop condition, and sum the result.
By our assumptions on $\cG$, $\iota_n:B_n\hookrightarrow B_{n+1}$ is injective and unital, and the Bratteli diagram can be identified with $\cG$.
We write
\[
        B := \varinjlim_{n \to \infty} B_n
\]
for the corresponding AF-algebra.

\subsection{Edge and vertex weighting and a KMS state}
\label{sec:GraphKMS}
To define a KMS state on $B$ we need to specify a weight on our graph $\cG$.

\begin{defn}
    Consider a \emph{weight} $w: \cG^{(0)} \cup \cG^{(1)} \to (0, \infty)$ on the vertices and edges of $\cG$.
    Then for $\xi \in \Path^n_v(\cG)$, we set $w(\xi) := w(\xi_1) w(\xi_2) \cdots w(\xi_n)$.
\end{defn}

We now assume that our weight $w$ on $\cG$ satisfies the following condition:
\begin{itemize}
\item 
there is a $\delta >0$ satisfying that for all $v \in \cG^{(0)}$, we have
\begin{equation}
\label{eq:WeightCondition}
\sum_{\substack{ |\eta| = 1 \\ s(\eta) = v }} w(\eta) w(r(\eta)) = \delta w(v).
\end{equation}
\end{itemize}
In the final subsection of this appendix we will give an example of such a weight.

\begin{rem}
Observe the condition above is
a system of linear equations, with one equation and variable for each vertex $v\in\cG^{(0)}$.
So any solution of this system is a formula for the weights of the vertices as a function of the weights of the edges, with possibly some free parameters.
If all edges have weight 1, then $w$ is a Frobenius-Perron eigenvector associated to $\delta$.
If all vertices are weighted 1, this is roughly the $\delta$-fairness condition of \cite{MR3420332} (without the balancing and even number of loops conditions).
\end{rem}

Let us now introduce dynamics on $B$.
For $t \in \bbR$ and $\pbk{\xi}{\eta} \in B_n$, define
\begin{equation}
\label{eq:IzumiDynamics}
    \sigma_t( \pbk{\xi}{\eta} ) := 
    \frac{w(\xi)^{it}}{w(\eta)^{it}} \pbk{\xi}{\eta}.
\end{equation}
This induces a one-parameter group of $*$-automorphism on $B_n$ which is compatible with the inclusion $B_n \hookrightarrow B_{n+1}$ defined above.
Hence we can extend it to a one-parameter group $\sigma_t \in \Aut(B)$.

Next define a state $\phi$ on $B_n$ by
\begin{equation}
\label{eq:PathAlgebraState}
    \phi( \pbk{\xi}{\eta}) 
    := 
    \delta_{\xi=\eta} \frac{1}{\delta^n} w(\xi) w(r(\xi)), \quad\quad \xi, \eta \in \Path^n_{\ast}(\cG).
\end{equation}
Using~\ref{eq:WeightCondition}, 
it is straightforward to see that
$\phi$ is a positive linear functional on $B_n$ with $\phi(1) = 1$ which is compatible with the inclusion $B_n\hookrightarrow B_{n+1}$. 
Hence $\phi$ extends uniquely to a state on $B$ by Corollary~\ref{cor:ExtendAF-UCP}.

Furthermore, observe that every $\pbk{\xi}{\eta}$ is entire with respect to $\sigma_t$, and when $\beta=-1$,
\begin{align*}
\phi(\pbk{\xi_1}{\eta_1} \cdot \sigma_{-i}(\pbk{\xi_2}{\eta_2}))
&=
\frac{w(\xi_2)}{w(\eta_2)}
\phi(\pbk{\xi_1}{\eta_1} \cdot \pbk{\xi_2}{\eta_2})
\\&=
\delta_{\eta_1=\xi_2}
\delta_{\xi_1=\eta_2}
\frac{w(\xi_2)}{\delta^n w(\eta_2)}w(\xi_1)w(r(\xi_1))
\\&=
\delta_{\eta_1=\xi_2}
\delta_{\xi_1=\eta_2}
\frac{1}{\delta^n}
w(\xi_2)w(r(\eta_2))
\\&=
\delta_{\eta_1=\xi_2}
\delta_{\xi_1=\eta_2}
\frac{1}{\delta^n}
w(\xi_2)w(r(\xi_2))
\\&=
\phi(\pbk{\xi_2}{\eta_2} \cdot \pbk{\xi_1}{\eta_1})
\end{align*}
for all $\pbk{\xi_1}{\eta_1},\pbk{\xi_2}{\eta_2}\in B_n$
(recall here that $r(\xi_2)=r(\eta_2)$ as $\pbk{\xi_2}{\eta_2}\in B_n$).

Since the Bratteli diagram for $B$ is connected and stationary, by \cite[Prop.~4.1]{MR1772604} and uniqueness of the Frobenius-Perron eigenvector (up to scaling), $\phi$ is the unique $\beta=-1$ KMS state for the dynamics $\sigma$.

We have just proved the following proposition.

\begin{prop}
The state $\phi$ is the unique $\beta=-1$ KMS state on $B$ for the dynamics $\sigma$.
Thus $M:=B''$ acting on $L^2(B,\phi)$ is a factor.
\end{prop}

Note that $B$ is simple because $\cG$ is connected.
Hence as before (see \S\ref{sec:KMS}), $\phi$ extends to a faithful normal state on $B''$ acting on $L^2(B, \phi)$, and the extension of $\sigma$ to $B''$ corresponds to the modular automorphism group with respect to the canonical cyclic vector in $L^2(B, \phi)$.

\subsection{Type of the von Neumann completion using Krieger's ratio set}

Recall that a maximal abelian von Neumann subalgebra $A$ of a von Neumann algebra $M$ is called a \emph{Cartan subalgebra} if
there exists a faithful normal conditional expectation $E: M\to A$
and its normalizer generates $M$, i.e., $\cN_M(A)''=M$ where
\[
    \cN_{M}(A) := \set{u \in \cU(M)} { u A u^* = A}.
\]
This latter condition, called \emph{regularity} of $A\subset M$, is equivalent to 
$M = A \vee \cN_{M}(A)$.

We now consider the von Neumann algebra $M:=B''$ in the GNS representation $L^2(B,\phi)$
and the abelian von Neumann subalgebra 
$$
A := \set{ \pbk{\xi}{\xi} }{\xi \in \bigcup_n \Path_{\ast}^n (\mathcal{G})}''.
$$

\begin{lem}
\label{lem:Cartan}
$A$ is a Cartan subalgebra of $M$.
\end{lem}
\begin{proof}
Clearly $A$ is abelian.
Note that for each $n \in \bbN$, the algebra $B_n$ is globally fixed by all $\sigma_t$, that is $\sigma_t(B_n) = B_n$ for all $t \in \bbR$.
Since $\sigma$ is the modular automorphism group for $\phi$, by a theorem of Takesaki~\cite[Thm. IX.4.2]{MR1943006} this implies the existence of a sequence of normal conditional expectations $E_n : M \to B_n$ such that $\phi \circ E_n = \phi$.
Then for all $x \in M$, we have that $E_n(x) \to x$ in the strong-$*$ topology.
Indeed, to see this, recall that on the unit ball of $M$, the topology induced by $\| \cdot \|_\phi$ coincides with the $\sigma$-strong topology~\cite[Prop. III.5.4]{MR1873025}.
It follows that for $x \in M$, we can choose $(x_i) \in \bigcup_n B_n$ with $\| x_i - x \|_\phi \to 0$.
Then
\begin{align*}
    \lim_n \| E_n(x) - x \|_\phi 
    &=
    \lim_n \left\| E_n(x-x_i) - (x-x_i)  + E_n(x_i) - x_i \right\|_\phi \\
    &\leq 
    2 \|x - x_i \|_\phi + \underbrace{\lim_{n \to \infty} \| E_n(x_i) - x_i \|_\phi}_{= 0}
\end{align*}
for all $i$,  and similarly for $x^*$.
Hence $\| E_n(x) -x \|_\phi + \| E_n(x^*) -x^* \|_\phi \to 0$ for all $x \in M$.
Since $(E_n(x) - x)_n$ and $(E_n(x^*) - x^*)_n$ are both norm-bounded sequences, the claim follows.

Now consider $a' \in M \cap A'$, so that $E_n(a') \in (A \cap B_n)' \cap B_n$.
Because $A \cap B_n$ is maximal abelian in $B_n$, it follows that $E_n(a') \in A \cap B_n$, and from the strong-$*$ convergence of $E_n(a')$ to $a'$, it follows that $a' \in A$, and hence $A$ is a MASA.

Since $A$ is also globally fixed by all $\sigma_t$, as above there exists a normal conditional expectation $E : M \to A$ such that $\phi \circ E = \phi$.
Since $\phi$ is faithful, it follows that $E$ is faithful as well.

It remains to check that $A \subset M$ is regular.
For each pair of paths $\xi, \eta \in \Path_{*}^{n}(\cG)$ with $r(\xi) = r(\eta)$ but $\xi \neq \eta$, define the operator
\begin{equation}
    \label{eq:PathGroup}
    u_{\xi,\eta} := \pbk{\xi}{\eta} + \pbk{\eta}{\xi} + 1 - \pbk{\xi}{\xi} - \pbk{\eta}{\eta}.
\end{equation}
It is straightforward to check that these are self-adjoint unitaries, and that $u_{\xi,\eta} A u_{\xi,\eta}^* = A$, and therefore $u_{\xi,\eta} \in \cN_M(A)$.
Finally, note that $u_{\xi,\eta} \pbk{\eta}{\eta} = \pbk{\xi}{\eta}$, from which it is clear that $A$ together with the $u_{\xi,\eta}$ generate $M = B''$.
\end{proof}

By the Feldman--Moore Theorem~\cite[Thm.~1]{MR0578730} $M\cong L\cR$ where $\cR$ is a countable standard relation on a Borel measure space $\Omega$
and
$A \cong L^\infty(\Omega, \mu)$.
We will show this explicitly for $M$ by identifying a measure space $(\Omega, \mu)$ such that $A = L^\infty(\Omega, \mu)$, together with an action $G \curvearrowright \Omega$, where $G$ is the group generated by the unitaries of \eqref{eq:PathGroup} in Lemma~\ref{lem:Cartan}.
The equivalence relation $\cR$ is then the orbit equivalence relation.
This explicit description makes it possible to compute Krieger's ratio sets and determine the type of $M$.

We first claim that $A = L^\infty(\Omega, \mu)$, where we define the measure space $\Omega$ as follows.
Let $\Omega$ be the Cantor set $\Path_*^\infty(\cG)$ of infinite paths on $\cG$.
For $\xi \in \Path_*^n(\cG)$, define the cylinder set
\[
    E_\xi  := \set{ \eta \in \Omega\,\,}{\,\,\xi_i = \eta_i \text{ for all } i \leq |\xi| }.
\]
For two finite paths $\xi, \eta$ starting at $*$ and $|\xi| \leq |\eta|$ we have $E_\xi \cap E_\eta = \emptyset$ if the first $|\xi|$ edges of $\eta$ do not agree with those of $\xi$, or $E_\eta \subset E_\xi$ otherwise.
Thus the cylinder sets form a basis for a topology on $\Omega$.

Let $\chi_{E_\xi}$ be the indicator function for a cylinder set.
Then we can map $\chi_{E_\xi} \mapsto \pbk{\xi}{\xi} \in A$, and this is compatible with multiplication on both sides:
 $\chi_{E_{\xi}} \cdot \chi_{E_{\eta}} = \chi_{E_{\xi} \cap E_\eta} \mapsto \pbk{\xi}{\xi} \cdot \pbk{\eta}{\eta}$.
For $\xi \in \Path^n_*(\cG)$, define
\begin{equation}
    \label{eq:CylinderMeasure}
    \mu(E_\xi) := \phi(\pbk{\xi}{\xi}) = \frac{1}{\delta^n} w(\xi) w(r(\xi)).
\end{equation}
Since $\phi$ is a normal state, it follows from the discussion above that $\mu$ extends to a regular Borel measure on $\Omega$.
The map $\pi:= \chi_{E_\xi} \mapsto \pbk{\xi}{\xi}$ mentioned earlier extends to a $*$-homomorphism of a (SOT) dense subalgebra of $L^\infty(\Omega, \mu)$ onto a dense subalgebra of $A$.
Write $\mu$ again for the state on $L^\infty(\Omega,\mu)$ induced by the probability measure.
Since $\mu(f) = \phi(\pi(f))$ on this dense subalgebra and both states are normal and faithful, it follows that $\pi$ extends to a spatial isomorphism of von Neumann algebras, and we may write $A = L^\infty(\Omega, \mu)$.

Now let $G$ be the (countable) group generated by our normalizing unitaries $u_{\xi,\eta}$ as in \eqref{eq:PathGroup}
which together with $A$ generate $M$.
We now identify the action of $G$ on $A$ under the identification $A=L^\infty(\Omega,\mu)$.
Since $M=L\cR$ acts on $L^2(\Omega,\mu)$,
we first identify the action of the ket-bra operators 
$\pbk{\xi}{\eta}$ for $\xi, \eta \in \Path_*^n(\mathcal{G})$ such that $r(\xi) = r(\eta)$.
Given $\omega \in \Omega$, we can write $\omega = \zeta \zeta'$, where $\zeta$ consists of the first $n$ edges in the path, so $\omega\in E_\zeta$, and $\zeta'$ the rest of the infinite path.
We have
$$
\pbk{\xi}{\eta} \cdot \omega =
\begin{cases}
\xi \zeta' & \text{if }\eta=\zeta, \text{ equivalently } \omega\in E_\eta=E_\zeta\\
0 &\text{otherwise,}
\end{cases}
$$
and thus
the action of (the generators of) $G$ on $\Omega$ is given as follows:
\begin{equation}
    \label{eq:GAction}
    u_{\xi,\eta} \cdot \omega := 
    \begin{cases}
        \eta\zeta'  & \textrm{if } \zeta = \xi, \text{ equivalently } \omega\in E_\xi=E_\zeta
        \\
        \xi\zeta'  & \textrm{if } \zeta = \eta, \text{ equivalently } \omega\in E_\eta=E_\zeta
        \\
        \omega & \textrm{otherwise.}
    \end{cases}
\end{equation}

In particular, $u_{\xi,\eta}$ sends single points in $\Omega$ to single points, since $r(\xi) = r(\eta)$.
Note that $u_{\xi,\eta}$ replaces the first part of the path $\omega$ with the path $\eta$ (resp.~$\xi$) if $\omega$ starts with the path $\xi$ (resp.~$\eta)$.
In all other cases it acts trivially.
Thus \eqref{eq:GAction} is indeed a $G$-action on $\Omega$.
Since the cylinder sets generate the topology on $\Omega$, $G$ acts by Borel automorphisms on $\Omega$,
and we have the standard Borel orbit equivalence relation
\[
    \cR := \set{(\omega, g\cdot \omega)\,}{\,g \in G,\, \omega\in\Omega},
\]
which is a Borel subset of $\Omega\times\Omega$.
Note that two elements in $\Omega$ are equivalent if and only if the corresponding infinite paths differ in finitely many places only.

For every $u_{\xi,\eta} \in G$, we define a new measure via $u_{\xi,\eta} \cdot \mu(E) := \mu(u_{\xi,\eta}^{-1} \cdot E)$ as usual.
The Radon--Nikodym derivatives are given by
\begin{equation}
\label{eq:RadonNikodym}
\frac{d (u_{\xi,\eta} \cdot \mu)}{d\mu} (\omega)
=
\begin{cases}
    \frac{w(\eta)}{w(\xi)} & \omega \in E_\xi
    \\
    \frac{w(\xi)}{w(\eta)} & \omega \in E_\eta
    \\
    1 & \text{otherwise,}
\end{cases}
\end{equation}
which can be directly verified using~\eqref{eq:CylinderMeasure} and~\eqref{eq:GAction}.
In particular, it follows that $\mu$ is quasi-invariant with respect to the action of $G$.

To compute Connes' invariant $S(M)$ it suffices to determine Krieger's ratio set $r(\Omega, G, \mu)$ \cite{MR0259624}.
We recall the definition here  (adapted slightly to the case at hand).
\begin{defn}[\cite{MR0259624}]
\label{def:KriegerRatio}
The \emph{ratio set} $r(\Omega, G, \mu)$ is the set of all $r \in [0,\infty)$ such that for all measurable $A \subset \Omega$ with $\mu(A) > 0$ and $\epsilon > 0$, there is a measurable $B \subset A$ with $\mu(B) > 0$ together with a $g \in G$ such that $g \cdot B \subset A$ satisfying that
\begin{equation}
    \label{eq:RatioSet}
    \left| \frac{d (g \cdot \mu)}{d\mu} (\omega) -r \right| < \epsilon
\end{equation}
for almost every $\omega \in B$.
\end{defn}

We claim that the ratio set is generated by quotients of the form $w(\xi)/w(\eta)$, where $\xi$ and $\eta$ are finite paths starting in $*$ and ending in the same vertex.
To prove this we will use the following lemma.

\begin{lem}
\label{lem:GRadonNikodym}
For all $\xi, \eta \in \Path_*^n(\cG)$ with $r(\xi) = r(\eta)$, there is an $\ell > 0$ such that for all cylinder sets $E_\zeta$, the following holds:
There are $A,B \subset E_\zeta$, together with a $g \in G$ with $g^2 = 1$ such that $g \cdot A = B$ and the following bounds are satisfied:
\[
\mu(A) \geq \ell \mu(E_\zeta) \qquad \text{and} \qquad \mu(B) \geq \ell \mu(E_\zeta).
\]
Moreover, $A,B$ and $g$ can be chosen such that we have
\[
\frac{d (g \cdot \mu)}{d \mu} (\omega) = 
\begin{cases}
    \frac{w(\xi)}{w(\eta)} & \text{for } \omega \in A
    \\
    \frac{w(\eta)}{w(\xi)} & \text{for } \omega \in B,
\end{cases}
\]
and $g$ acts trivially on $E_{\zeta'}$ if $E_\zeta \cap E_{\zeta'} = \emptyset$.
\end{lem}
\begin{proof}
For each vertex $v \in \cG^{(0)}$, choose an edge $e_{v,*} \in \cG^{(1)}$ going from $v$ to $*$.
Define
\[
\ell := \frac{1}{\delta^{n+1}} \frac{m}{M} w(r(\eta)) \cdot \min \{ w(\xi), w(\eta) \},
\]
where $m := \min_{v \in \cG^{(0)}} w(e_{v,*})$ is the minimum weight of the edges ending in the distinguished point  and $M := \max_{v \in \cG^{(0)}} w(v)$ the maximum vertex weight.
Define the two sets $A := E_{\zeta e_{r(\zeta),*} \eta}$ and $B:= E_{\zeta e_{r(\zeta),*} \xi}$.
Then clearly $A,B \subset E_\zeta$.
Using~\eqref{eq:CylinderMeasure} it follows that
\[
\frac{\mu(A)}{\mu(E_\zeta)} 
=
\frac{1}{\delta^{n+1}} \frac{w(e_{r(\zeta),*}) w(\eta) w(r(\eta))}{w(r(\zeta))} \geq \ell.
\]
Since $r(\xi) = r(\eta)$, we obtain $\mu(B) \geq \ell \mu(E_\zeta)$ similarly.

Now set $g := u_{\zeta e_{r(\zeta),*} \xi, \zeta e_{r(\zeta),*} \eta}$.
Then $g \cdot E_\zeta \subset E_\zeta$, $g \cdot A = B$, and $g^2 = 1$.
The claim on the Radon--Nikodym derivative follows from \eqref{eq:RadonNikodym}.
Finally, if $E_\zeta \cap E_{\zeta'} = \emptyset$, 
our choice of $g$ acts trivially on $E_{\zeta'}$
by \eqref{eq:GAction}.
\end{proof}

\begin{thm}
\label{thm:RatioInKriegerSet}
     For all $n \in \mathbb{N}$ and $\xi, \eta \in \Path_{v}^n(\cG)$ with $r(\xi) = r(\eta)$, we have $\frac{w(\xi)}{w(\eta)} \in r(\Omega, G, \mu)$.
\end{thm}
\begin{proof}
We first consider the case where $v = *$.
Let $E \subset \Omega$ be a Borel set with $\mu(E) > 0$.
Since $\mu$ is regular, for all $\epsilon > 0$, there exists $U \subset \Omega$ open such that $E \subset U$ and $\mu(U \setminus E) < \epsilon \mu(E)$.
The cylinder sets generate the topology of $\Omega$, and since two cylinder sets are either disjoint or one is contained in the other, it follows that there is a sequence $(\xi_k)\subset \bigcup_n\Path_{*}^n(\cG)$ such that
\[
    U = \bigcup_{k=1}^\infty  E_{\xi_k}, \qquad\qquad E_{\xi_i} \cap E_{\xi_j} = \emptyset \text{ whenever } i \neq j,
\]
where we allow $E_{\xi_k}$ to be the empty set.
Choose $m \in \mathbb{N}$ such that $\mu(U) - \sum_{k=1}^m \mu(E_{\xi_k}) < \epsilon \mu(E)$ and set $E' := \bigcup_{k=1}^m E_{\xi_k}$.
Then
\[
\mu(E \Delta E')
= 
\mu(E \setminus E') + \mu(E' \setminus E)
\leq
\mu(U\setminus E') + \mu(U\setminus E) 
<
2 \epsilon \mu(E).
\]
For each $k=1,\dots,m$, choose sets $A_k, B_k \subset E_{\xi_k}$ and a group element $g_k$ as in the statement of Lemma~\ref{lem:GRadonNikodym}.
Set $A := \bigcup_k A_k$ and $B := \bigcup_k B_k$.
Since the sets $E_{\xi_k}$ are disjoint it follows that
\[
    \mu(A) \geq \ell \mu(E') \qquad \text{and} \qquad \mu(B) \geq \ell \mu(E'),
\]
where $\ell$ is the constant from Lemma~\ref{lem:GRadonNikodym} (note that it only depends on the choice of $\xi, \eta$).
Let $g= g_1 g_2 \cdots g_m$
(since the $g_k$ act non-trivially only on mutually disjoint sets,
they all commute). 
Again by Lemma~\ref{lem:GRadonNikodym},
\[
\frac{d (g \cdot \mu)}{d \mu} (\omega) = 
\begin{cases}
    r  & \text{for } \omega \in A
    \\
    r^{-1} & \text{for } \omega \in B,
\end{cases}
\]
where we set $r:= w(\xi)/w(\eta)$.

Using the results so far, we obtain the following estimates:
\begin{align*}
&\mu ((E \cap A) \Delta A) 
\leq
\mu(E \setminus E')
<
2 \epsilon \mu(E)
\\
&\mu( g (E \cap B) \Delta A)
=
\mu(g ((E \cap B) \Delta B))
=
\frac{1}{r} \mu((E \cap B) \Delta B)
<
\frac{2 \epsilon}{r} \mu(E)
\\
& \mu(A) 
\geq 
\ell \mu(E') 
> 
\ell(\mu(U) - \epsilon \mu(E))
\geq
\ell(1-\epsilon)\mu(E). 
\end{align*}
In the second line, we used that $gA = B$ and $g^2 = 1$ in the first equality, and that the Radon--Nikodym derivative is equal to $r^{-1}$ on $B$ in the second equality.
The last inequality in the last line follows because $E \subset U$.

Define the following sets:
\[
A_1 := E \cap A \qquad \text{ and } \qquad A_2 := g(E \cap B) \cap A. 
\]
Then clearly $\mu(A) \geq \mu(A_1) + \mu(A_2) - \mu(A_1 \cap A_2)$.
Without loss of generality (since $\ell$ is independent of $\epsilon$) we may assume that $\epsilon < \frac{\ell}{ (\ell + 2 + 2 r^{-1})}$
We claim that in this case, we have $\mu( A_1 \cap A_2) > 0$.
Indeed, we calculate using the estimates above
\begin{align*}
\mu(A_1 \cap A_2)
&\geq
\mu(A_1) + \mu(A_2) - \mu(A)
\\
&\geq
\mu(A) - 2 \epsilon \mu(E) + \mu(A) - \frac{2 \epsilon}{r} \mu(E) - \mu(A)
\\
&>
\left( \ell(1-\epsilon) - 2 \epsilon - \frac{2 \epsilon}{r} \right) \mu(E)
\\
&=
\left( \ell - \left(\ell + 2 + \frac{2}{r}\right) \epsilon \right) \mu(E)
\\
&> 0.
\end{align*}
Finally, define
\[
    A' = (E \cap A) \cap g(E \cap B) \subset E.
\]
Then by construction we have $g \cdot A' \subset E$, 
$\mu(A') \geq \mu(A_1 \cap A_2) > 0$,
and moreover $\frac{d (g\cdot \mu)}{d \mu}(\omega) = r$ for all $\omega \in A'$.
Hence $r \in r(\Omega,G,\mu)$ as claimed.

To prove the general case, pick an arbitrary finite path $\zeta$ with $s(\zeta) = *$ and $r(\zeta) = v$.
Then $\zeta\xi, \zeta\eta \in \Path_*^{n+|\zeta|}(\cG)$, so $w(\zeta \xi)/w(\zeta \eta) \in r(\Omega, G, \mu)$.
But since $w(\zeta\xi) = w(\zeta)w(\xi)$ and similarly for $w(\zeta\eta)$, the claim follows.
\end{proof}

We note that the Krieger ratio set can more generally be described in terms of the ``asymptotic range'' of a Radon--Nikodym cocycle (see~\cite[Defn. 8.2]{MR0578656}).
But the asymptotic range is equal to the Connes invariant $S$~\cite{MR0341115} of the Krieger factor by~\cite[Prop. 2.11]{MR0578730}.
Thus we arrive at the following corollary.

\begin{cor}
\label{cor:SInvariant}
Let $G_R$ be the subgroup of $\bbR_{>0}$ generated by the set
\[
\set{ \frac{w(\xi)}{w(\eta)}}{ n \in \bbN,\,\, \xi,\eta \in \Path_*^n(\cG),\, r(\xi) = r(\eta) }.
\]
Then the Connes invariant $S(M)$ is given by the closure of $G_R$.
\end{cor}
\begin{proof}
From the proof of Theorem~\ref{thm:RatioInKriegerSet} it is enough to consider paths starting in $*$.
By the same theorem,
each ratio $\frac{w(\xi)}{w(\eta)}$ is in $S(M)$.
Since $S(M)$ is a closed subset of $[0,\infty)$, and $S(M) \cap \bbR_{>0}$ is a multiplicative subgroup of $\bbR_{>0}$, the closure of $G_R$ is contained in $S(M)$.

Now each $g \in G$ is a finite product of operators of the form $u_{\xi_k,\eta_k}$ from \eqref{eq:PathGroup}.
One can then see (for example, using \eqref{eq:RadonNikodym} and the chain rule for Radon--Nikodym derivatives) that all values of $\frac{d (g \cdot \mu)}{d\mu}$ are products of ratios of the weights of a pair of paths as above.
Hence $S(M)$ cannot be any larger.
\end{proof}

\subsection{Application: type of the boundary algebra in the Levin-Wen model}

We associate a graph $\cG$ to a unitary fusion category $\cC$ as follows.
The set of vertices is $\cG^{(0)} = \Irr(\cC)$, with distinguished vertex $* =1_\cC$ the tensor unit.
For $c_1,c_2\in \cG^{(0)}$,
the set of edges $\cG^{(1)}$ is a disjoint union of orthonormal bases for $\cC(c_2 \to  a \otimes c_1 \otimes b)$ over $a,b\in\Irr(\cC)$, with the isometry inner product determined by 
$\langle f | g\rangle \id_{c_2} = f^\dag\circ g$.
Since $\cC$ is rigid, there is an edge between every pair of vertices.
Thus we can define algebras $B_n$ again as in~\S\ref{sec:PathAlgebra}.

To illustrate this, first consider the case $n=1$.
Recall that $X = \bigoplus_{c \in \Irr(\cC)} c$.
We can identify an edge $\xi$ from $*$ to $c$ with a morphism $\xi \in \cC(c \to a \otimes b)$, which we can identify with $\xi \in \cC(c \to X \otimes X)$.
If $\eta$ is another edge from $*$ to $c$, we have $\xi \circ \eta^\dag \in \End(X^{\otimes 2})$, which can be identified with $\pbk{\xi}{\eta}$.
By considering all $c \in \Irr(\cC)$, we see that $B_1 = \End_\cC(X^{\otimes 2})$.

If $\xi_1 \in \cC(c_2 \to a_1 \otimes c_1 \otimes b_1)$ and $\xi_2 \in \cC(c_3 \to a_2 \otimes c_2 \otimes b_2)$ are two edges, the path $\xi := \xi_1 \xi_2$ is defined by $(\id_{a_2} \otimes \xi_1 \otimes \id_{b_2}) \circ \xi_2$.
As before, we can again see $\xi$ as an element of $\cC(c_3 \to X \otimes X \otimes c_1 \otimes X \otimes X)$.
Generalising this to paths of length $n$, it follows that each algebra $B_n\cong\End_\cC(X^{\otimes 2n})$.
Moreover, the inclusion $\iota_n : B_n \hookrightarrow B_{n+1}$ as defined in~\eqref{eq:BnInclusion} is given by $\xi \mapsto \id_X \otimes \xi \otimes \id_X$.

Next we define a weight $w$ on $\cG$.
For vertices, set $w(c) := d_c$, where $d_c$ is the quantum dimension of $c \in \Irr(\cC)$.
Since an edge from $c_1$ to $c_2$ is an element $\eta \in \cC(c_2 \to a \otimes c_1 \otimes b)$ for some irreducible $a$ and $b$, we can set $w(\eta) := d_a d_b$.
We see \eqref{eq:WeightCondition} is satisfied with $\delta = \left( \sum_{c \in \Irr(\cC)} d_c^2 \right)^2$: 
for fixed $c_1 \in \Irr(\cC)$,
\begin{align*}
\sum_{a,b,c_2\in\Irr(\cC)}\sum_{\eta\in \operatorname{ONB}(c_2 \to a \otimes c_1 \otimes b)} w(\eta)w(r(\eta))
&=
\sum_{a,b,c_2,e\in\Irr(\cC)} N_{ae}^{c_2}N_{c_1b}^{e}d_ad_bd_{c_2}
\\&=
\sum_{a,b,e\in\Irr(\cC)} N_{c_1b}^{e}d_a^2d_bd_{e}
=
\left(\sum_{a,b\in\Irr(\cC)} d_a^2d_b^2\right)d_{c_1}.
\end{align*}

Now suppose that $c\prec a \otimes b$ for $a,b,c \in \Irr(\cC)$.
That is, $c$ appears in the direct sum composition of $a \otimes b$.
We call such a triple $(a,b,c)$ \emph{admissable}.

\begin{lem}
\label{lem:SGenerator}
The following set
\[
    \set{ \frac{d_a d_b}{d_c}}{\,(a,b,c) \text{ is admissable} } \subset \bbR_{>0}
\]
generates a dense subgroup of $S(M)$.
\end{lem}
\begin{proof}
Let $(a,b,c)$ be an admissable triple.
Then there exist non-zero
\[
    \xi \in \cC(c \to a \otimes 1 \otimes b) 
    \qquad \text{ and }\qquad 
    \eta \in \cC(c \to c \otimes 1 \otimes 1).
\]
We can choose both of them to be edges from $1$ to $c$, and identify $\xi$ and $\eta$ with paths of length one.
Note that $w(\xi) = d_a d_b$ and $w(\eta) = d_c$.
Hence by Corollary~\ref{cor:SInvariant} it follows that $\frac{d_a d_b}{d_c} \in S(M)$.

Now let $\xi = \xi_1 \cdots \xi_n$ be a path from $*$ to $c$, with $\xi_i \in \cC(c_{i+1} \to a_i \otimes c_i \otimes b_i)$, and thus $c_1 = *$ and $c_{n+1} = c$.
Using the fusion rules and semisimplicity, it follows that such a (non-zero) $\xi_i$ only exists if there is some $k$ such that $(c_i, b_i, k)$ and $(a_i,k,c_{i+1})$ are both admissable.
Using the first part, this means that
\[
\left(\frac{d_{c_i} d_{b_i}}{d_k} \right) \cdot \left( \frac{d_{a_i} d_k}{d_{c_{i+1}}}\right)
=
\frac{d_{a_i} d_{b_i} d_{c_i}}{d_{c_{i+1}}}
\]
is in the generated group.
Taking the product over all $i = 1, \dots, n$ gives $w(\xi)/d_c$.
If $\eta$ is another path from $*$ to $c$ of the same length, the same argument gives that $w(\eta)/d_c$ is in the generated group.
Thus $w(\xi)/w(\eta)$ is as well.
The claim then follows from Corollary~\ref{cor:SInvariant}.
\end{proof}

From the characterization of the net of algebras $(B_n)$, it follows that the limit algebra $B$ is the same as that of the fusion categorical net $I \mapsto \fF(I)$ for the Levin--Wen model studied in \S\ref{sec:LWBoundaryState}.
Using the proof of Lemma~\ref{lem:SGenerator} and the observations made earlier, we see that the state $\phi$ defined on $B$ via \eqref{eq:PathAlgebraState} is exactly the canonical boundary state $\psi_{\fF}$, which is described explicitly in Proposition~\ref{prop:CanonicalState}.
(Alternatively, since the dynamics~\eqref{eq:IzumiDynamics} on $B$ is inverse to the dynamics~\eqref{eq:ModularAutomorphism} on $\fF$, the $\beta=-1$ KMS state $\phi$ on $B$ agrees with the $\beta=1$ KMS state $\psi_\fF$ on $\fF$.)
Thus we have obtained the following lemma.

\begin{lem}
The von Neumann algebra $M$ is isomorphic to the boundary algebra $\fF''$ for the Levin--Wen model for $\cC$.
\end{lem}

We are now in a position to determine the type of the boundary algebra.
Since $S(M)$ is closed in $[0,\infty)$ and $S(M) \cap \bbR_{>0}$ is a subgroup of $\bbR_{>0}$ it follows that $S(M)$ must be one of the following four cases~\cite{MR0341115}:
\begin{itemize}
\item $S(M) = \{ 1 \}$, in which case $M$ is semi-finite.
From Lemma~\ref{lem:SGenerator} this is true if and only if all $d_a = 1$, i.e., $\cC$ is pointed.
In this case, it is also easy to see directly that the state $\phi$ defined in \S\ref{sec:GraphKMS} is a trace.
\item $S(M) = \{0,1\}$, in which case $M$ is type $\rm{III}_0$.
Note that $g \cdot \mu$ and $\mu$ are mutually absolutely continuous.
Hence the Radon--Nikodym derivative $\frac{d(g\cdot \mu)}{d\mu}$ cannot be zero on a non-neglible set.
By Corollary~\ref{cor:SInvariant}, if $0 \in S(M)$, it must be the limit of non-zero points in $S(M)$.
This is in contradiction with $1$ being the only other point in $S(M)$, so the type $\rm{III}_0$ case cannot occur.
\item
$S(M) = \{0\} \cup \set{\lambda^n}{n \in \bbZ}$ for some $\lambda \in (0,1)$, in which case $M$ is of type $\rm{III}_\lambda$.
Note that Lemma~\ref{lem:SGenerator} implies that for each admissable triple $(a,b,c)$, there is an integer $Z_{ab}^c$ such that
\[
\frac{d_a d_b}{d_c} = \lambda^{Z_{ab}^c},
\]
and the set $\{ Z_{ab}^c \}$ generates $\bbZ$ as a group.
\item $S(M) = [0,\infty)$, in which case $M$ is of type $\rm{III}_1$.
This is equivalent to the set
\begin{equation}
\label{eq:LogGroup}
    \set{ \frac{d_a d_b}{d_c}}{\,(a,b,c) \text{ admissable}}
\end{equation}
generating a dense subgroup of $\bbR_{>0}$.
\end{itemize}
This result refines Theorem~\ref{thm:InnerIffPointed}.

\bibliographystyle{alpha}
\bibliography{bibliography}

\end{document}